\documentclass[jair,11pt]{article}
\usepackage{jair, theapa, rawfonts}
\jairheading{74}{2022}{303-351}{05/2021}{05/2022}
\ShortHeadings{Adaptive Greedy versus Non-Adaptive Greedy for Influence Maximization}
{Chen, Peng, Schoenebeck, \& Tao}
\firstpageno{303}

\usepackage{amsmath,amsfonts,amsthm,amssymb}
\usepackage{algorithmic}
\usepackage[ruled]{algorithm2e}
\usepackage{blkarray}
\usepackage{enumitem}
\usepackage{graphicx}

\newtheorem{theorem}{Theorem}[section]
\newtheorem{lemma}[theorem]{Lemma}

\newtheorem{proposition}[theorem]{Proposition}

\theoremstyle{definition}
\newtheorem{definition}[theorem]{Definition}
\newtheorem{remark}[theorem]{Remark}

\newcommand{\Z}{{\mathbb Z}}
\newcommand{\R}{{\mathbb R}}
\newcommand{\z}{{\mathbf z}}

\newcommand{\Seed}{{\mathcal S}}
\newcommand{\alg}{{\mathcal A}}
\newcommand{\w}{\overline{w}}
\newcommand{\f}{{\mathfrak f}}
\newcommand{\m}{{\mathfrak m}}
\newcommand{\live}{{\texttt L}}
\newcommand{\block}{{\texttt B}}
\newcommand{\unknown}{{\texttt U}}
\newcommand{\IN}{{\textit IN}}

\DeclareMathOperator*{\E}{{\mathbb E}}
\DeclareMathOperator*{\argmax}{argmax}

\newcommand{\infmax}{{\sc InfMax}\xspace}
\newcommand{\ICM}{\texttt{ICM}\xspace}
\newcommand{\LTM}{\texttt{LTM}\xspace}

\title{Adaptive Greedy versus Non-Adaptive Greedy\\ for Influence Maximization}

\author{\name Wei Chen \email weic@microsoft.com \\
       \addr Microsoft Research, China
       \AND
       \name Binghui Peng \email bp2601@columbia.edu \\
       \addr Department of Computer Science, Columbia University
       \AND
       \name Grant Schoenebeck \email schoeneb@umich.edu \\
       \addr School of Information, University of Michigan
       \AND
       \name Biaoshuai Tao \email bstao@sjtu.edu.cn \\
       \addr School of Electronic Information and Electrical Engineering, \\Shanghai Jiao Tong University}

\begin{document}

\date{}
\maketitle

\begin{abstract}
We consider the \emph{adaptive influence maximization problem}: given a network and a budget $k$, iteratively select $k$ seeds in the network to maximize the expected number of adopters.  In the \emph{full-adoption feedback model}, after selecting each seed, the seed-picker observes all the resulting adoptions.  In the \emph{myopic feedback model}, the seed-picker only observes whether each neighbor of the chosen seed adopts.
Motivated by the extreme success of greedy-based algorithms/heuristics for influence maximization, we propose the concept of \emph{greedy adaptivity gap}, which compares the performance of the adaptive greedy algorithm to its non-adaptive counterpart.
Our first result shows that, for submodular influence maximization, the adaptive greedy algorithm can perform up to a $(1-1/e)$-fraction worse than the non-adaptive greedy algorithm, and that this ratio is tight.
More specifically, on one side we provide examples where the performance of the adaptive greedy algorithm is only a $(1-1/e)$ fraction of the performance of the non-adaptive greedy algorithm in four settings: for both feedback models and both the \emph{independent cascade model} and the \emph{linear threshold model}.
On the other side, we prove that in any submodular cascade, the adaptive greedy algorithm always outputs a $(1-1/e)$-approximation to the expected number of adoptions in the optimal non-adaptive seed choice.
Our second result shows that, for the general submodular diffusion model with full-adoption feedback, the adaptive greedy algorithm can outperform the non-adaptive greedy algorithm by an unbounded factor.
Finally, we propose a risk-free variant of the adaptive greedy algorithm that always performs no worse than the non-adaptive greedy algorithm.
\end{abstract}

\section{Introduction}
The \emph{influence maximization problem} (\infmax) is an optimization problem that asks which seeds a viral marketing campaign should target (e.g. by giving free products) so that the propagation from these seeds influences the most people in a social network.
That is, given a graph, a \emph{stochastic diffusion model} defining how each node is infected by its neighbors, and a limited budget $k$, how to pick $k$ seeds such that the expected number of total infected nodes in this graph at the end of the diffusion is maximized.
This problem has significant applications in viral marketing, outbreak detection, rumor controls, etc, and has been extensively studied (cf.~\shortciteA{chen2013information,li2018influence}).

For \infmax, most of the existing work has considered \emph{submodular} diffusion models, especially the \emph{independent cascade model} and the \emph{linear threshold model}~\shortcite{KempeKT03}.  Likewise, we also focus on submodular diffusion models.  In submodular diffusion models, a vertex $v$'s marginal probability of becoming infected after a new neighbor $t$ is infected given $S$ as the set of $v$'s already infected neighbors is at least the marginal probability that $v$ is infected after $t$ is newly infected given $T\supseteq S$ as the set of $v$'s already infected neighbors
(see the paragraph before Theorem~\ref{thm:submodular} for more details).   Intuitively, this means that the influence of infected nodes are substitutes and never have synergy.

When submodular \infmax is considered, nearly all the known algorithms/heuristics are based on a greedy algorithm that iteratively picks the seed that has the largest marginal influence.
Some of them improve the running time of the original greedy algorithm by skipping vertices that are known to be suboptimal~\shortcite{leskovec2007cost,goyal2011celf++}, while the others improve the scalability of the greedy algorithm by using more scalable algorithms to approximate the expected total influence~\shortcite{borgs2014maximizing,tang2014influence,tang2015influence,cheng2013staticgreedy,ohsaka2014fast} or computing a score of the seeds that is closely related to the expected total influence~\shortcite{chen2009efficient,chen2010scalable,ChenYZ10,goyal2011simpath,jung2012irie,galhotra2016holistic,tang2018online,schoenebeck2019influence}.
\shortciteA{arora2017debunking} benchmark most of the aforementioned variants of the greedy algorithms.

In this paper, we study the \emph{adaptive influence maximization problem}, where seeds are selected iteratively and feedback is given to the seed-picker after selecting each seed.
Two different feedback models have been studied in the past: the \emph{full-adoption feedback model} and the \emph{myopic feedback model}~\shortcite{golovin2011adaptive}.
In the full-adoption feedback model, the seed-picker sees the entire diffusion process of each selected seed, and in the myopic feedback model the seed-picker only sees whether each neighbor of the chosen seed is infected.

Past literature focused on the \emph{adaptivity gap}---the ratio between the performance of the \emph{optimal} adaptive algorithm and the performance of the \emph{optimal} non-adaptive algorithm~\shortcite{golovin2011adaptive,peng2019adaptive,chen2019adaptivity}.
However, even in the non-adaptive setting, \infmax is known to be APX-hard~\shortcite{KempeKT03,schoenebeck2019influence}.
As a result, in practice, it is not clear whether the adaptivity gap can measure how much better an adaptive algorithm can do.

In this paper, we define and consider the \emph{greedy adaptivity gap}, which is the ratio between the performance of the adaptive greedy algorithm and the non-adaptive greedy algorithm.
We focus on the gap between the greedy algorithms for three reasons.
First, as we mentioned, the APX-hardness of \infmax renders the practical implications of the adaptivity gap unclear.
Second, as we remarked at the beginning, the greedy algorithm is used almost exclusively in the context of influence maximization.
Third, the iterative nature of the original greedy algorithm naturally extends to the adaptive setting.

\subsection{Our Results}
We show that, for the general submodular diffusion models, with both the full-adoption feedback model and the myopic feedback model, the infimum of the greedy adaptivity gap is exactly $(1-1/e)$ (Section~\ref{sect:inf}).
In addition, this result can be extended to the two well-studied submodular diffusion models: the independent cascade model and the linear threshold model.
This is proved in two steps.

As the first step, in Section~\ref{sect:inf_gap_tight}, we show that there are \infmax instances where the adaptive greedy algorithm can only produce $(1-1/e)$ fraction of the influence of the solution output by the non-adaptive greedy algorithm.
This result is surprising: one would expect that the adaptivity is always helpful, as the feedback provides more information to the seed-picker, which makes the seed-picker refine the seed choices in future iterations.
Our result shows that this is not the case, and the feedback, if overly used, can make the seed-picker act in a more myopic way, which is potentially harmful.

As the second step, in Section~\ref{sect:inf_lower_bound}, we show that the adaptive greedy algorithm always achieves a $(1-1/e)$-approximation of the non-adaptive optimal solution, so its performance is always at least a $(1-1/e)$ fraction of the performance of the non-adaptive greedy algorithm.
In particular, combining the two steps, we see that when the adaptive greedy algorithm output only obtains a (nearly) $(1-1/e)$-fraction of the performance of the non-adaptive greedy algorithm, the non-adaptive greedy algorithm is (almost) optimal.
This worst-case guarantee indicates that the adaptive greedy algorithm will never be too bad.

As the second result, in Section~\ref{sect:sup}, we show that the supremum of the greedy adaptivity gap is infinity, for the general submodular diffusion model with full-adoption feedback.
This indicates that the adaptive greedy algorithm can perform significantly better than its non-adaptive counterpart.
We also show, with almost the same proof, that the adaptivity gap in this setting (general submodular model with full-adoption feedback) is also unbounded.

All the results above hold for the ``exact'' deterministic greedy algorithm where a vertex with the exact maximum marginal influence is chosen as a seed in each iteration.
However, most variants of the greedy algorithm used in practice are randomized algorithms that find a seed with a marginal influence \emph{close to} the maximum \emph{with high probability} in each iteration.
In Section~\ref{sect:robustness}, we discuss how our results for the exact greedy algorithm can be adapted to those greedy algorithms used in practice.

Finally, in Section~\ref{sect:variant}, we propose a risk-free but more conservative variant of the adaptive greedy algorithm, which always performs at least as well as the non-adaptive greedy algorithm.
In Section~\ref{sect:experiments}, we compare this variant of the adaptive greedy algorithm with the adaptive greedy algorithm and the non-adaptive greedy algorithm by implementing experiments on social networks in our daily lives.

\subsection{Related Work}
The influence maximization problem was initially posed by \shortciteA{DomingosR01,RichardsonD02}.
\shortciteA{KempeKT03} proposed the linear threshold model and the independent cascade model, and show that they are submodular.
Whenever a diffusion model is submodular, the greedy algorithm was shown to obtain a $(1-1/e)$-approximation to the optimal number of infections~\shortcite{Nemhauser78,KempeKT03,KempeKT05,MosselR10}.
For undirected networks, the greedy algorithm provides a $(1-1/e+c)$-approximation for some constant $c>0$ for the independent cascade model~\shortcite{khanna2014influence}, while the approximation guarantee of the greedy algorithm for the linear threshold model is still asymptotically $(1-1/e)$ (i.e., the constant improvement $+c$ to the approximation guarantee in the independent cascade model does not occur in the linear threshold model)~\shortcite{schoenebeck2020limitations}.

For adaptive \infmax, \shortciteA{golovin2011adaptive} showed that \infmax with the independent cascade model and full-adoption feedback is \emph{adaptive submodular}\footnote{Informally, if we have two observations $\varphi_1$ and $\varphi_2$ such that $\varphi_2$ includes $\varphi_1$ (i.e., the status of all the edges/vertices observed in $\varphi_1$ are also observed in $\varphi_2$), adaptive submodularity says that the marginal increment of the total expected number of infections when including an extra seed $s$ given $\varphi_1$ is no less than the marginal increment of the total expected number of infections when including the extra seed $s$ given $\varphi_2$. We refer the readers to Reference~\shortcite{golovin2011adaptive} for the formal definition of adaptive submodularity.}, which implies that the adaptive greedy algorithm obtains a $(1-1/e)$-approximation to the adaptive optimal solution.
On the other hand, \infmax for the independent cascade model with myopic feedback, as well as \infmax for the linear threshold model with both feedback models, are not adaptive submodular.
In particular, the adaptive greedy algorithm fails to obtain a $(1-1/e)$-approximation for the independent cascade model with myopic feedback~\shortcite{peng2019adaptive}.
\shortciteA{peng2019adaptive} showed that the adaptivity gap for the independent cascade model with myopic feedback is at most $4$ and at least $e/(e-1)$, and they also showed that both the adaptive and non-adaptive greedy algorithms perform a $0.25(1-1/e)$-approximation to the adaptive optimal solution.
\shortciteA{d2020improved} improved these results by showing that the adaptivity gap under the same setting is at most $2e/(e-1)\approx3.164$ and non-adaptive greedy algorithms perform a $0.5(1-1/e)$-approximation.
The adaptivity gap for the independent cascade model with full-adoption feedback, as well as the adaptivity gap for the linear threshold model with both feedback models, are still open problems, although there is some partial progress~\shortcite{chen2019adaptivity,d2020better}.

Our paper is not the first work studying the adaptive greedy algorithm.
Previous work focused on improving the running time of the adaptive greedy algorithm~\shortcite{han2018efficient,sun2018multi}.
However, to the best of our knowledge, our work is the first one that compares the adaptive greedy algorithm to its non-adaptive counterpart.

Finally, we remark that there do exist \infmax algorithms that are not based on greedy~\shortcite{BKS07,GL13,angell2017don,schoenebeck2017beyond,toct2019beyond,schoenebeck2019think}, but they are typically for non-submodular diffusion models.

We summarize the existing results about the adaptivity gap and our new results about the greedy adaptivity gap in Table~\ref{tab:results}.

\begin{table}[t]
    \centering
    \begin{tabular}{|l||l|l|l|}
    \hline
    model & AG & GAG inf & GAG sup \\
    \hline
    \hline
    \ICM, full-adoption  & 
    \begin{tabular}{l}
         $\geq e/(e-1)$  \\
         \shortcite{chen2019adaptivity}
    \end{tabular}
     &  
     \begin{tabular}{l}
          $1-1/e$  \\
          (Thm~\ref{thm:inf_gap}) 
     \end{tabular}
       & unknown\\
    \hline
    \ICM, myopic  &   
    \begin{tabular}{l}
         $\geq e/(e-1)$, $\leq4$ \\
         \shortcite{peng2019adaptive}\\
         $\leq 2e/(e-1)$\\
         \shortcite{d2020improved}
    \end{tabular}
    &  \begin{tabular}{l}
          $1-1/e$  \\
          (Thm~\ref{thm:inf_gap}) 
     \end{tabular} & 
    \begin{tabular}{l}
         $\leq4e/(e-1)$  \\
         \shortcite{peng2019adaptive}\\
         $\leq2e^2/(e-1)^2$\\
         \shortcite{d2020improved}
    \end{tabular}\\
    \hline
    \LTM, full-adoption & unknown &  \begin{tabular}{l}
          $1-1/e$  \\
          (Thm~\ref{thm:inf_gap}) 
     \end{tabular} & unknown\\
    \hline
    \LTM, myopic & unknown &  \begin{tabular}{l}
          $1-1/e$  \\
          (Thm~\ref{thm:inf_gap}) 
     \end{tabular} & unknown\\
    \hline
    GSM, full-adoption & $\infty$ (Thm~\ref{thm:adaptivityGap}) &  \begin{tabular}{l}
          $1-1/e$  \\
          (Thm~\ref{thm:inf_gap}) 
     \end{tabular} & $\infty$ (Thm~\ref{thm:sup_gap})\\
    \hline
    GSM, myopic &   
    \begin{tabular}{l}
         $\geq e/(e-1)$ \\
          \shortcite{peng2019adaptive}
    \end{tabular}& \begin{tabular}{l}
          $1-1/e$  \\
          (Thm~\ref{thm:inf_gap}) 
     \end{tabular} & unknown\\
    \hline
    \end{tabular}
    \caption{Results for the adaptivity gap (AG), the infimum of the greedy adaptivity gap (GAG inf) and the supremum of the greedy adaptivity gap (GAG sup), where \ICM stands for the independent cascade model, \LTM stands for the linear threshold model, and GSM stands for the general submodular diffusion model.}
    \label{tab:results}
\end{table}

\section{Preliminaries}
All graphs in this paper are simple and directed.
Given a graph $G=(V,E)$ and a vertex $v\in V$, let $\Gamma(v)$ and $\deg(v)$ be the set of in-neighbors and the in-degree of $v$ respectively.

\subsection{Triggering Model}
We consider the well-studied \emph{triggering model}~\shortcite{KempeKT03}, which is commonly used to capture ``general'' submodular diffusion models.
A more general way to capture submodular diffusion models is the \emph{general threshold model}~\shortcite{KempeKT03} with \emph{submodular local influence functions}. All our results hold under this setting as well. We will discuss this in Appendix~\ref{append:GTM}.
\begin{definition}[\shortciteA{KempeKT03}]\label{def:GTM}
  The \emph{triggering model}, $I_{G,F}$, is defined by a graph $G=(V,E)$ and for each vertex $v$
  a distribution $\mathcal{F}_v$ over all the subsets of its in-neighbors $\{0,1\}^{|\Gamma(v)|}$.
  Let $F=\{\mathcal{F}_v\mid v\in V\}$.

  On an input seed set $S \subseteq V$,  $I_{G,F}(S)$ outputs a set of infected vertices as follows:
  \begin{enumerate}[noitemsep,nolistsep]
    \item[1.] Initially, only vertices in $S$ are infected. Each vertex $v$ samples a subset of its in-neighbors $T_v\subseteq\Gamma(v)$ from $\mathcal{F}_v$ independently. We call $T_v$ the \emph{triggering set} of $v$.
    \item[2.]  In each subsequent round, a vertex $v$ becomes infected if a vertex in $T_v$ is infected in the previous round.
    \item[3.] After a round where no additional vertices are infected, the set of infected vertices is the output.
  \end{enumerate}
\end{definition}

$I_{G,F}$ in Definition~\ref{def:GTM} can be viewed as a random function $I_{G,F}:\{0,1\}^{|V|}\to\{0,1\}^{|V|}$.
In addition, if the triggering set $T_v$ is fixed for each vertex $v$, then $I_{G,F}$ is deterministic.
Given $v$, its triggering set $T_v$, and an in-neighbor $u\in\Gamma(v)$, we say that the edge $(u,v)$ is \emph{live} if $u\in T_v$, and we say that $(u,v)$ is \emph{blocked} if $u\notin T_v$.
It is easy to see that, when the triggering sets for all vertices are sampled, $I_{G,F}(S)$ is the set of all vertices that are reachable from $S$ when removing all blocked edges from the graph.

We define a \emph{realization} of a graph $G=(V,E)$ as a function $\phi:E\to\{\live,\block\}$ such that $\phi(e)=\live$ if $e\in E$ is live and $\phi(e)=\block$ if $e\in E$ is blocked.
Let $I_{G,F}^\phi:\{0,1\}^{|V|}\to\{0,1\}^{|V|}$ be the deterministic function corresponding to the triggering model $I_{G,F}$ with vertices' triggering sets following realization $\phi$.
We write $\phi\sim F$ to indicate that a realization $\phi$ is sampled according to $F=\{\mathcal{F}_v\}$.

The triggering model captures the well-known independent cascade and linear threshold models.
In the two definitions below, we define the two models in terms of the triggering model, which is sufficient for this paper.
In Appendix~\ref{append:originaldefinition}, we present the original definitions and give some intuitions for the two models for those readers who are not familiar with them.

\begin{definition}\label{def:ICM}
  The \emph{independent cascade model} \ICM is a special case of the triggering model $I_{G,F}$ where $G=(V,E,w)$ is an edge-weighted graph with $w(u,v)\in(0,1]$ for each $(u,v)\in E$ and $\mathcal{F}_v$ is the distribution such that each $u\in\Gamma(v)$ is included in $T_v$ with probability $w(u,v)$ independently.
\end{definition}

\begin{definition}\label{def:LTM}
  The \emph{linear threshold model} \LTM is a special case of the triggering model $I_{G,F}$ where $G=(V,E,w)$ is an edge-weighted graph with $w(u,v)>0$ for each $(u,v)\in E$ and $\sum_{u\in\Gamma(v)}w(u,v)\leq1$ for each $v\in V$, and $\mathcal{F}_v$ is the distribution defined as follows: order $v$'s in-neighbors $u_1,\ldots,u_{T}$ arbitrarily (where $T$ is the in-degree of $v$), sample a real number $r$ in $[0,1]$ uniformly, and
  $$T_v=\left\{\begin{array}{ll}
      \{u_t\} & \mbox{if }r\in\left[\sum_{i=1}^{t-1}w(u_i,v),\sum_{i=1}^tw(u_i,v)\right) \\
      \emptyset & \mbox{if }r\geq\sum_{i=1}^Tw(u_i,v)
  \end{array}\right..$$
  Intuitively, $T_v$ includes at most one of $v$'s in-neighbors such that each $u_t$ is included with probability $w(u_t,v)$.
\end{definition}

Given a triggering model $I_{G,F}$, let $\sigma_{G,F}:\{0,1\}^{|V|}\to\R_{\geq0}$ be the \emph{global influence function} defined as $\sigma_{G,F}(S)=\E_{\phi\sim F}[|I_{G,F}^\phi(S)|]$.
We drop the subscripts $G,F$ and write the global influence function as $\sigma(\cdot)$ when there is no ambiguity.

A function $f$ mapping from a set of elements to a non-negative value is \emph{submodular} if $f(A\cup\{v\})-f(A)\geq f(B\cup\{v\})-f(B)$ for any two sets $A,B$ with $A\subseteq B$ and any element $v\notin B$.

\begin{theorem}[\shortciteA{KempeKT03}]\label{thm:submodular}
For any triggering model $I_{G,F}$, $\sigma_{G,F}(\cdot)$ is submodular. In particular, $\sigma_{G,F}(\cdot)$ is submodular for both \ICM and \LTM. 
\end{theorem}

\subsection{\infmax and Adaptive \infmax}
\label{sect:Prelim_infmax}
\begin{definition}\label{def:infmax}
  The \emph{influence maximization problem} (\infmax) is an optimization problem which takes inputs $G=(V,E)$, $F$, and $k\in\Z^+$, and outputs a seed set $S$ that maximizes the expected total number of infections: $S\in\argmax_{S \subseteq V: |S| \leq k} \sigma(S)$.
\end{definition}

In the remaining part of this subsection, we define the adaptive version of the influence maximization problem.
We will define two different models: the \emph{full-adoption feedback model} and the \emph{myopic feedback model}.
Suppose a seed set $S\subseteq V$ is chosen by the seed-picker, and an underlying realization $\phi$ is given but not known by the seed-picker.
Informally, in the full-adoption feedback model, the seed-picker sees all the vertices that are infected by $S$ in all future iterations, i.e., the seed-picker sees $I_{G,F}^\phi(S)$. In the myopic feedback model, the seed-picker only sees the states of $S$'s neighbors, i.e., whether each vertex in $\{v\mid \exists s\in S:s\in\Gamma(v)\}$ is infected.

Define a \emph{partial realization} as a function $\varphi:E\to\{\live,\block,\unknown\}$ such that $\varphi(e)=\live$ if $e$ is known to be live, $\varphi(e)=\block$ if $e$ is known to be blocked, and $\varphi(e)=\unknown$ if the status of $e$ is not yet known.
We say that a partial realization $\varphi$ \emph{is consistent with} the full realization $\phi$, denoted by $\phi\simeq\varphi$, if $\phi(v)=\varphi(v)$ whenever $\varphi(v)\neq\unknown$.
For ease of notation, for an edge $(u,v)\in E$, we will write $\phi(u,v),\varphi(u,v)$ instead of $\phi((u,v)),\varphi((u,v))$.

\begin{definition}\label{def:fulladoption}
  Given a triggering model $I_{G=(V,E),F}$ with a realization $\phi$, the \emph{full-adoption feedback} is a function $\Phi_{G,F,\phi}^{\f}$ mapping a seed set $S\subseteq V$ to a partial realization $\varphi$ such that
  \begin{itemize}
      \item $\varphi(u,v)=\phi(u,v)$ for each $u\in I_{G,F}^\phi(S)$, and
      \item $\varphi(u,v)=\unknown$ for each $u\notin I_{G,F}^\phi(S)$.
  \end{itemize}
\end{definition}


\begin{definition}\label{def:myopic}
  Given a triggering model $I_{G=(V,E),F}$ with a realization $\phi$, the \emph{myopic feedback} is a function $\Phi_{G,F,\phi}^\m$ mapping a seed set $S\subseteq V$ to a partial realization $\varphi$ such that
  \begin{itemize}
      \item $\varphi(u,v)=\phi(u,v)$ for each $u\in S$, and
      \item $\varphi(u,v)=\unknown$ for each $u\notin S$.
  \end{itemize}
\end{definition}

An \emph{adaptive policy} $\pi$ is a function that maps a seed set $S$ and a partial realization $\varphi$ to a vertex $v=\pi(S,\varphi)$, which corresponds to the next seed the policy $\pi$ would choose given $\varphi$ and $S$ being the set of seeds that has already been chosen.
Naturally, we only care about $\pi(S,\varphi)$ when $\varphi=\Phi_{G,F,\phi}^\f(S)$ or $\varphi=\Phi_{G,F,\phi}^\m(S)$, although we define $\pi$ that specifies an output for any possible inputs $S$ and $\varphi$.
Notice that we have defined $\pi$ as a deterministic policy for simplicity, and our results hold for randomized policies.
Let $\Pi$ be the set of all possible adaptive policies.

Notice that an adaptive policy $\pi$ completely specifies a seeding strategy in an iterative way.
Given an adaptive policy $\pi$ and a realization $\phi$, let $\Seed^\f(\pi,\phi,k)$ be the first $k$ seeds selected according to $\pi$ with the underlying realization $\phi$ under the full-adoption feedback model.
By our definition, $\Seed^\f(\pi,\phi,k)$ can be computed as follows:
\begin{enumerate}
    \item initialize $S=\emptyset$;
    \item update $S=S\cup\{\pi(S,\Phi_{G,F,\phi}^\f(S))\}$ for $k$ iterations;
    \item output $\Seed^\f(\pi,\phi,k)=S$.
\end{enumerate}
Define $\Seed^\m(\pi,\phi,k)$ similarly for the myopic feedback model, where $\Phi_{G,F,\phi}^\m(S)$ instead of $\Phi_{G,F,\phi}^\f(S)$ is used in Step~2 above.

Let $\sigma^\f(\pi,k)$ be the expected number of infected vertices given that $k$ seeds are chosen according to $\pi$, i.e., $\sigma^\f(\pi,k)=\E_{\phi\sim F}[|I_{G,F}^\phi(\Seed^\f(\pi,\phi,k))|]$.
Define $\sigma^\m(\pi,k)$ similarly for the myopic feedback model.

\begin{definition}\label{def:infmax_adaptive}
  The \emph{adaptive influence maximization problem} (adaptive \infmax) is an optimization problem which takes as inputs $G=(V,E)$, $F$, and $k\in\Z^+$, and outputs an adaptive policy $\pi$ maximizing the expected total number of infections: $\pi\in\argmax_{\pi\in\Pi} \sigma^\f(\pi,k)$ or $\pi\in\argmax_{\pi\in\Pi} \sigma^\m(\pi,k)$ (depending on the feedback model used).
\end{definition}

\subsection{Adaptivity Gap and Greedy Adaptivity Gap}
\label{sect:prelim_AG}
The adaptivity gap is defined as the ratio between the performance of the optimal adaptive policy and the performance of the optimal non-adaptive seeding strategy.
In this paper, we only consider the adaptivity gap for triggering models.

\begin{definition}\label{def:adaptivitygap}
  The \emph{adaptivity gap with full-adoption feedback} is
  $$\sup_{G,F,k}\frac{\max_{\pi\in\Pi}\sigma^\f(\pi,k)}{\max_{S\subseteq V,|S|\leq k}\sigma(S)}.$$
  The \emph{adaptivity gap with myopic feedback} is defined similarly.
\end{definition}

The (non-adaptive) \emph{greedy algorithm} iteratively picks a seed that has the maximum marginal gain to the objective function $\sigma(\cdot)$:
\begin{enumerate}
    \item initialize $S=\emptyset$;
    \item update for $k$ iterations $S\leftarrow S\cup\{s\}$, where $s\in\argmax_{s\in V}(\sigma(S\cup\{s\})-\sigma(S))$ with tie broken in an arbitrarily consistent order;
    \item return $S$.
\end{enumerate}
Let $S^g(k)$ be the set of $k$ seeds output by the (non-adaptive) greedy algorithm.

The \emph{greedy adaptive policy} $\pi^g$ is defined as $\pi^g(S,\varphi)=s$ such that 
$$s\in\argmax_{s\in V}\E_{\phi\simeq\varphi}\left[\left|I_{G,F}^\phi\left(S\cup\{s\}\right)\right|-\left|I_{G,F}^\phi\left(S\right)\right|\right],$$
with ties broken in an arbitrary consistent order.

\begin{definition}\label{def:adaptivitygreedygap}
  Given a triggering model $I_{G,F}$ and $k\in\Z^+$, the \emph{greedy adaptivity gap with full-adoption feedback} is
  $\frac{\sigma^\f(\pi^g,k)}{\sigma(S^g(k))}$.
  The \emph{greedy adaptivity gap with myopic feedback} is defined similarly.
\end{definition}

Notice that, unlike the adaptivity gap in Definition~\ref{def:adaptivitygap}, we leave $G,F,k$ unspecified (instead of taking a supremum over them) when defining the greedy adaptivity gap.
This is because we are interested in both supremum and infimum of the ratio $\frac{\sigma^\f(\pi^g,k)}{\sigma(S^g(k))}$.
Notice that the infimum of the ratio $\frac{\max_{\pi\in\Pi}\sigma^\f(\pi,k)}{\max_{S\subseteq V,|S|\leq k}\sigma(S)}$ in Definition~\ref{def:adaptivitygap} is $1$: the optimal adaptive policy is at least as good as the optimal non-adaptive policy, as the non-adaptive policy can be viewed as a special adaptive policy; on the other hand, it is easy to see that there are \infmax instances such that the optimal adaptive policy is no better than non-adaptive one (for example, a graph containing $k$ vertices but no edges).
For this reason, we only care about the supremum of this ratio.

\section{Infimum of Greedy Adaptivity Gap}
\label{sect:inf}
In this section, we show that the infimum of the greedy adaptivity gap for the triggering model is exactly $(1-1/e)$, for both the full-adoption feedback model and the myopic feedback model.
This implies that the greedy adaptive policy can perform even worse than the conventional non-adaptive greedy algorithm, but it will never be significantly worse.
Moreover, we show that this result also holds for both
\ICM (Definition~\ref{def:ICM}) and \LTM (Definition~\ref{def:LTM}).
The theorem below formalizes those above-mentioned results.

\begin{theorem}\label{thm:inf_gap}
For the full-adoption feedback model,
$$\inf_{G,F,k:\mbox{ }I_{G,F}\text{ is \ICM}}\frac{\sigma^\f(\pi^g,k)}{\sigma(S^g(k))}=\inf_{G,F,k:\mbox{ }I_{G,F}\text{ is \LTM}}\frac{\sigma^\f(\pi^g,k)}{\sigma(S^g(k))}
=\inf_{G,F,k}\frac{\sigma^\f(\pi^g,k)}{\sigma(S^g(k))}=1-\frac1e.$$
The same result holds for the myopic feedback model.
\end{theorem}
In Section~\ref{sect:inf_gap_tight}, we show by providing examples that the greedy adaptive policy in the worst case will only achieve $(1-1/e+\varepsilon)$-approximation of the expected number of infected vertices given by the non-adaptive greedy algorithm, for both \ICM and \LTM.

In Section~\ref{sect:inf_lower_bound}, we show that the greedy adaptive policy has performance at least $(1-1/e)$ of the performance of the non-adaptive optimal seeds (Theorem~\ref{thm:lowerbound}).
Theorem~\ref{thm:lowerbound} provides a lower bound on the greedy adaptivity gap for the triggering model and is also interesting on its own.
At the end of Section~\ref{sect:inf_lower_bound}, we prove Theorem~\ref{thm:inf_gap} by putting the results from Section~\ref{sect:inf_gap_tight} and Section~\ref{sect:inf_lower_bound} together.

\subsection{Tight Examples}
\label{sect:inf_gap_tight}
In this subsection, we show that the adaptive greedy algorithm can perform worse than the non-adaptive greedy algorithm by a factor of $(1-1/e+\varepsilon)$, for both \ICM and \LTM and any $\varepsilon>0$.
This may be surprising, as one would expect that the feedback provided to the seed-picker will refine the seed choices in  future iterations.
Here, we provide some intuitions why adaptivity can sometimes hurt.
Suppose there are two promising sequences of seed selections, $\{s,u_1,\ldots,u_k\}$ and $\{s,v_1,\ldots,v_k\}$, such that
\begin{itemize}
    \item $s$ is the best seed which will be chosen first;
    \item $\{s,u_1,\ldots,u_k\}$ has a better performance;
    \item the influence of $u_1,\ldots,u_k$ are non-overlapping, the influence of $v_1,\ldots,v_k$ are non-overlapping, but the influence of $u_i,v_j$ overlaps for each $i,j$; moreover, if $u_1$ is picked as the second seed, the greedy algorithm, adaptive or not, will continue to pick $u_2,\ldots,u_k$, and if $v_1$ is picked as the second seed, $v_2,\ldots,v_k$ will be picked next;
\end{itemize}
Now, suppose there is a vertex $t$ elsewhere which can be infected by both $s$ and $v_1$, such that
\begin{itemize}
    \item if $t$ is infected by $s$, which slightly reduces the marginal influence of $v_1$, $v_1$ will be less promising than $u_1$;
    \item if $t$ is not infected by $s$, $v_1$ is more promising than $u_1$;
    \item in average, when there is no feedback, $v_1$ is still less promising than $u_1$, even after adding the increment in $t$'s infection probability to $v_1$'s expected marginal influence.
\end{itemize}
In this case, the non-adaptive greedy algorithm will ``go to the right trend'' by selecting $u_1$ as the second seed; the adaptive greedy algorithm, if receiving feedback that $t$ is not infected by $s$, will ``go to the wrong trend'' by selecting $v_1$ next.

As a high-level description of the lesson we learned, both versions of the greedy algorithms are intrinsically myopic, and the feedback received by the adaptive policy may make the seed-picker act in a more myopic way, which could be more hurtful to the final performance.

We will assume in the rest of this section that vertices can have positive integer weights, as justified in the following remark.
\begin{remark}\label{remark:weightedvertices}
For both \ICM and \LTM, we can assume without loss of generality that each vertex has a positive integer weight, so that, in \infmax, we are maximizing the expected total weight of the infected vertices instead of maximizing the expected number of infected vertices as before.
Suppose we want to make a vertex $v$ have weight $W\in\Z^+$.
We can construct $W-1$ vertices $w_1,\ldots,w_{W-1}$, and create $W-1$ directed edges $(v,w_1),\ldots,(v,w_{W-1})$ with weight $1$. (Recall from Definition~\ref{def:ICM} and Definition~\ref{def:LTM} that the graphs in both \ICM and \LTM are edge-weighted, and the weights of edges completely characterize the collection of triggering set distributions $F$.)
It is straightforward from Definition~\ref{def:ICM} and Definition~\ref{def:LTM} that, for both \ICM and \LTM, each of $w_1,\ldots,w_{W-1}$ will be infected with probability $1$ if $v$ is infected.
In addition, both the greedy algorithm and the greedy adaptive policy will never pick any of $w_1,\ldots,w_{W-1}$ as seeds, as seeding $v$ is strictly better.
Therefore, we can consider the subgraph consisting of $v,w_1,\ldots,w_{W-1}$ as a gadget that representing a vertex $v$ having weight $W$.
\end{remark}

The following lemma shows that, for both the full-adoption and the myopic feedback models, under \ICM, the greedy adaptive policy can perform worse than the non-adaptive greedy algorithm by a factor of almost $(1-1/e)$.

\begin{lemma}\label{lem:tightICM}
For any $\varepsilon>0$, there exists $G,F,k$ such that $I_{G,F}$ is an \ICM and
$$\frac{\sigma^\f(\pi^g,k)}{\sigma(S^g(k))}\leq1-\frac1e+\varepsilon,\qquad\frac{\sigma^\m(\pi^g,k)}{\sigma(S^g(k))}\leq1-\frac1e+\varepsilon.$$
\end{lemma}
\begin{proof}
We will construct an \infmax instance $(G=(V,E,w),k+1)$ with $k+1$ seeds allowed.
Let $W\in\Z^+$ be a sufficiently large integer divisible by $k^{2k}(k-1)$ and whose value is to be decided later.
Let $\Upsilon=W/k^2$.
The vertex set $V$ contains the following weighted vertices:
\begin{itemize}
    \item a vertex $s$ that has weight $2W$;
    \item a vertex $t$ that has weight $4k\Upsilon$;
    \item $2k$ vertices $u_1,\ldots,u_k,v_1,\ldots,v_k$ that have weight $1$;
    \item $2k^2$ vertices $\{w_{ij}\mid i=1,\ldots,2k;j=1,\ldots,k\}$
    \begin{itemize}
        \item $w_{11},\ldots,w_{1k}$ have weight $\frac Wk$;
        \item for each $i\in\{2,\ldots,k\}$, $w_{i1},\ldots,w_{ik}$ have weight $\frac1k(1-\frac1k)^{i-1}W+\frac{4k-2}{k-1}\Upsilon$;
        \item for each $i\in\{k+1,\ldots,2k\}$, $w_{i1},\ldots,w_{ik}$ have weight $\frac1k(1-\frac1k)^kW$.
    \end{itemize}
\end{itemize}
The edge set $E$ is specified as follows:
\begin{itemize}
    \item create two edges $(v_1,t)$ and $(s,t)$;
    \item for each $i=1,\ldots,k$, create $2k$ edges $(u_i,w_{1i}),(u_i,w_{2i}),\ldots,(u_i,w_{(2k)i})$, and create $k$ edges $(v_i,w_{i1}),(v_i,w_{i2}),\ldots,(v_i,w_{ik})$.
\end{itemize}
For the weights of edges, all the edges have weight $1$ except for the edge $(s,t)$ which has weight $1/k$.

It is straightforward to check that
\begin{eqnarray}
\sigma(\{s\})&=&\w(s)+\frac1k\w(t)=2W+4\Upsilon,\label{eqn:sigma_s}\\
\forall i\in\{1,\ldots,k\}:\sigma(\{u_i\})&=&\w(u_i)+\sum_{j=1}^{2k}\w(w_{ji})=1+W+(4k-2)\Upsilon,\label{eqn:sigma_ui}\\
\sigma(\{v_1\})&=&\w(v_1)+\w(t)+\sum_{j=1}^k\w(w_{1j})=1+4k\Upsilon+W,\label{eqn:sigma_v1}\\
\forall i\in\{2,\ldots,k\}:\sigma(\{v_i\})&=&\w(v_i)+\sum_{j=1}^k\w(w_{ij})=1+\left(1-\frac1k\right)^{i-1}W+\frac{k(4k-2)}{k-1}\Upsilon,\label{eqn:sigma_vi}
\end{eqnarray}
and the influence of the remaining vertices are significantly less than these.

Since $s$ has the highest influence, both the greedy algorithm and the greedy adaptive policy will choose $s$ as the first seed.

The non-adaptive greedy algorithm will choose $u_1,\ldots,u_k$ iteratively for the next $k$ seeds, and the expected number of infected vertices by the seeds chosen by non-adaptive greedy algorithm is
\begin{equation}\label{eqn:su1touk}
\sigma(\{s,u_1,\ldots,u_k\})=\w(s)+\frac1k\w(t)+\sum_{i=1}^k\w(u_i)+\sum_{i=1}^{2k}\sum_{j=1}^k\w(w_{ij})=(k+2)W+o(W).
\end{equation}
To show the former claim, letting $U_i=\{s,u_1,\ldots,u_i\}$ and $U_0=\{s\}$, and supposing without loss of generality that the non-adaptive greedy algorithm has chosen $U_i$ for the first $(i+1)$ seeds (notice the symmetry of $u_1,\ldots,u_k$), it suffices to show that, for any vertex $x$, we have
\begin{equation}\label{eqn:nonadaptivegreedy_marginal}
    \sigma(U_i\cup\{u_{i+1}\})-\sigma(U_i)\geq\sigma(U_i\cup\{x\})-\sigma(U_i).
\end{equation}
To consider an $x$ that makes the right-hand side large, it is easy to see that we only need to consider one of $u_{i+1},\ldots,u_{k},v_1,v_2$, as the remaining vertices clearly have less marginal influence.
By symmetry, $\sigma(U_i\cup\{u_{i+1}\})=\sigma(U_i\cup\{u_{i+2}\})=\cdots=\sigma(U_i\cup\{u_{k}\})$.
Therefore, we only need to consider $x$ being $v_1$ or $v_2$.
It is straightforward to check that
\begin{align}
\sigma(U_i\cup\{u_{i+1}\})-\sigma(U_i)&=1+W+(4k-2)\Upsilon,\label{eqn:nonadaptivegreedy_marginal_ui+1}\\
\sigma(U_i\cup\{v_1\})-\sigma(U_i)&=1+(4k-4)\Upsilon+\frac{k-i}kW\leq1+W+(4k-4)\Upsilon,\label{eqn:nonadaptivegreedy_marginal_v1}\\
\sigma(U_i\cup\{v_{2}\})-\sigma(U_i)&=1+\frac{k-i}k\left(1-\frac1k\right)W+\frac{(k-i)(4k-2)}{k-1}\Upsilon\nonumber\\
&\leq1+W-\frac{W}k+(4k+5)\Upsilon.\label{eqn:nonadaptivegreedy_marginal_v2}
\end{align}
Recall that $\Upsilon=W/k^2$, straightforward calculations show that $\sigma(U_i\cup\{u_{i+1}\})-\sigma(U_i)$ is maximum.

For the greedy adaptive policy, we have seen that $s$ will be the first seed chosen.
The second seed picked by the greedy adaptive policy will depend on whether $t$ is infected by $s$.
Notice that the status of $t$ is available to the policy in both the full-adoption feedback model and the myopic feedback model, so the arguments here, as well as the remaining part of this proof, apply to both feedback models.
By straightforward calculations, the greedy adaptive policy will pick $v_1$ as the next seed if $t$ is not infected by $s$, and the policy will pick a seed from $u_1,\ldots,u_k$ otherwise.

In the latter case, the policy will eventually pick the seed set $\{s,u_1,\ldots,u_k\}$, which will infect vertices with a total weight of
$$\w(s)+\w(t)+\sum_{i=1}^k\w(u_i)+\sum_{i=1}^{2k}\sum_{j=1}^k\w(w_{ij})=(k+2)W+o\left(W\right)$$
with probability $1$ (notice that we are in the scenario that $t$ has been infected by $s$).

In the former case, we can see that the third seed picked by the policy will be $v_2$ instead of any of $u_1,\ldots,u_k$.
In particular, $v_2$ contributes $1+(1-\frac1k)W+\frac{k(4k-2)}{k-1}\Upsilon$ infected vertices.
On the other hand, since $w_{11},\ldots,w_{ik}$ have already been infected by $v_1$, the marginal contribution for each $u_i$ is $\sigma(\{u_i\})-\w(w_{1i})=1+W+(k-1)\cdot\frac{4k-2}{k-1}\Upsilon-\frac1kW$, which is less than the contribution of $v_2$.
By similar analysis, we can see that the greedy adaptive policy in this case will pick the seed set $\{s,v_1,\ldots,v_k\}$, which will infect vertices with a total weight of
$$
\w(s)+\w(t)+\sum_{i=1}^k\w(v_i)+\sum_{i=1}^k\sum_{j=1}^k\w(w_{ij})=
\left(2+\sum_{i=1}^k\left(1-\frac1k\right)^{i-1}\right)W+o\left(W\right)$$
$$\qquad=
\left(2+k\left(1-\left(1-\frac1k\right)^k\right)\right)W+o\left(W\right)
$$
in expectation.

Since $t$ will be infected with probability $\frac{1}{k}$, the expected weight of infected vertices for the greedy adaptive policy is
$$\frac1k\left((k+2)W+o\left(W\right)\right)+\left(1-\frac1k\right)\cdot\left(\left(2+k\left(1-\left(1-\frac1k\right)^k\right)\right)W+o\left(W\right)\right)$$
$$\qquad\leq\left(3+k\left(1-\left(1-\frac1k\right)^k\right)\right)W+o\left(W\right).$$

Putting this together with Equation~(\ref{eqn:su1touk}), both $\frac{\sigma^\f(\pi^g,k)}{\sigma(S^g(k))}$ and $\frac{\sigma^\m(\pi^g,k)}{\sigma(S^g(k))}$ in this case are at most
$$\frac{\left(3+k\left(1-\left(1-\frac1k\right)^k\right)\right)W+o\left(W\right)}{(k+2)W+o\left(W\right)},$$
which has limit $1-1/e$ when both $W$ and $k$ tend to infinity.
\end{proof}

The following lemma shows that, for both the full-adoption and the myopic feedback models, under \LTM, the greedy adaptive policy can perform worse than the non-adaptive greedy algorithm by a factor of almost $(1-1/e)$.

\begin{lemma}\label{lem:tightLTM}
For any $\varepsilon>0$, there exists $G,F,k$ such that $I_{G,F}$ is an \LTM and
$$\frac{\sigma^\f(\pi^g,k)}{\sigma(S^g(k))}\leq1-\frac1e+\varepsilon,\qquad\frac{\sigma^\m(\pi^g,k)}{\sigma(S^g(k))}\leq1-\frac1e+\varepsilon.$$
\end{lemma}
\begin{proof}
We will construct an \infmax instance $(G=(V,E,w),k+1)$ with $k+1$ seeds allowed.
Let $W\in\Z^+$ be a sufficiently large integer divisible by $k^{2k}(k-1)$ and whose value are to be decided later.
Let $\Upsilon=W/k^2$.
The vertex set $V$ contains the following weighted vertices:
\begin{itemize}
    \item a vertex $s$ that has weight $2W$;
    \item a vertex $t$ that has weight $4k\Upsilon$;
    \item $k$ vertices $u_1,\ldots,u_k$ that have weight $1$;
    \item $k$ vertices $v_1,\ldots,v_k$ such that $\w(v_1)=W$ and $\w(v_i)=W(1-\frac1k)^{i-1}+\frac{4k^2-6k}{k-1}\Upsilon$ for each $i=2,\ldots,k$;
    \item $k$ vertices $v_{k+1},\ldots,v_{2k}$ such that $\w(v_{k+1})=\cdots=\w(v_{2k})=W(1-\frac1k)^k$.
\end{itemize}
The edge set $E$ and the weights of edges are specified as follows:
\begin{itemize}
    \item create two edges $(v_1,t)$ and $(s,t)$ with weights $1-\frac{1}{k}$ and $\frac{1}{k}$ respectively;
    \item create $2k^2$ edges $\{(u_i,v_j)\mid i=1,\ldots,k;j=1,\ldots,2k\}$, each of which has weight $\frac1k$.
\end{itemize}
It is easy to check that each vertex $v$ satisfy $\sum_{u\in\Gamma(v)}w(u,v)\leq1$, as required in Definition~\ref{def:LTM}.

The remaining part of the analysis is similar to the proof of Lemma~\ref{lem:tightICM}.
The first seed chosen by both algorithms is $s$.
After this, each $u_i$ has marginal influence $1+\frac1k\sum_{i=1}^{2k}\w(v_i)=1+W+(4k-6)\Upsilon+(1-\frac1k)\frac1k\cdot4k\Upsilon=1+W+(4k-2-\frac4k)\Upsilon$ (notice that each of $v_1,\ldots,v_{2t}$ is infected with probability $1/k$, and $t$'s probability of infection increases from $1/k$ to $1$ if $v_1$ is infected by $u_i$).
Since $t$ is infected by $s$ with probability $\frac{1}{k}$, the marginal influence of $v_1$ without any feedback is $(1-\frac{1}{k})\w(t)+\w(v_1)=W+(4k-4)\Upsilon$.
If $t$ is known to be infected, the marginal influence of $v_1$ is $W$, and the marginal influence of each $u_i$ is $1+W+(4k-6)\Upsilon$. 
If $t$ is known to be uninfected, then seeding $v_1$ will infect $t$ with probability $1$.
In this case, the marginal influence of $v_1$ is $W+4k\Upsilon$, and the marginal influence of each $u_i$ is $1+W+(4k-2)\Upsilon$.
By comparing these values, the non-adaptive greedy algorithm will pick one of $u_1,\ldots,u_k$ as the second seed, and the greedy adaptive policy will pick $v_1$ as the second seed if $t$ is not infected and one of $u_1,\ldots,u_k$ as the second seed if $t$ is infected.
(Notice that $\w(v_1)>\w(v_2)>\cdots>\w(v_k)>\w(v_{k+1})=\cdots=\w(v_{2k})$.)

Simple analyses show the non-adaptive greedy algorithm will choose $\{s,u_1,\ldots,u_k\}$, which will infect all of $v_1,\ldots,v_{2k}$ with probability $1$, and the greedy adaptive policy will choose $\{s,v_1,\ldots,v_k\}$ with a very high probability $1-\frac{1}{k}$, which will leave $v_{k+1},\ldots,v_{2k}$ uninfected.
Since $s,v_1,\ldots,v_{2k}$ are the only vertices with weight $\Theta(W)$ and we have both $\sum_{i=1}^k\w(v_i)=(4+k(1-(1-\frac1k)^k))W+o(W)$ and $\sum_{i=1}^{2k}\w(v_i)=(4+k)W+o(W)$, the lemma follows by taking the limit $W\rightarrow\infty$ and $k\rightarrow\infty$.
\end{proof}

\subsection{Lower Bound}
\label{sect:inf_lower_bound}
We prove the following theorem in this subsection, which states that, for both the full-adoption and myopic feedback models, under the general triggering model, the greedy adaptive policy can achieve at least $(1-1/e)$ fraction of the performance of the non-adaptive \emph{optimal} solution.
\begin{theorem}\label{thm:lowerbound}
For a triggering model $I_{G,F}$, we have both
$$\sigma^\f(\pi^g,k)\geq\left(1-\frac1e\right)\max_{S\subseteq V,|S|\leq k}\sigma(S)\qquad\mbox{and}\qquad
\sigma^\m(\pi^g,k)\geq\left(1-\frac1e\right)\max_{S\subseteq V,|S|\leq k}\sigma(S).$$
\end{theorem}
For a high-level idea of the proof, let $S$ with $|S|=i$ be the seeds picked by $\pi^g$ for the first $i$ iterations and $S^\ast$ be the optimal non-adaptive seed set: $S^\ast\in\argmax_{|S'|\leq k}\sigma(S')$.
Given $S$ as the existing seeds and any feedback (myopic or full-adoption) corresponding to $S$, we can show that the marginal increment to the expected influence caused by the $(i+1)$-th seed picked by $\pi^g$ is at least $1/k$ of the marginal increment to the expected influence caused by $S^\ast$.
Then, a standard argument showing that the greedy algorithm can achieve a $(1-1/e)$-approximation for any submodular monotone optimization problem can be used to prove this theorem.

Theorem~\ref{thm:lowerbound} is implied by the following three propositions.
In the remaining part of this section, we let $S^\ast$ be an optimal seed set for the non-adaptive \infmax: $S^\ast\in\max_{S\subseteq V,|S|\leq k}\sigma(S)$.

We first show that the global influence function after fixing a partial seed set $S$ and any possible feedback of $S$ is still submodular.

\begin{proposition}\label{prop:partial_submodular}
Given a triggering model $I_{G,F}$, any $S\subseteq V$, any feedback model (either full-adoption or myopic) and any partial realization $\varphi$ that is a valid feedback of $S$ (i.e., $\exists\phi:\varphi=\Phi_{G,F,\phi}^\f(S)$ or $\exists\phi:\varphi=\Phi_{G,F,\phi}^\m(S)$, depending on the feedback model considered), the function $\mathcal{T}:\{0,1\}^{|V|}\to\R_{\geq0}$ defined as $\mathcal{T}(X)=\E_{\phi\simeq\varphi}[|I_{G,F}^\phi(S\cup X)|]$ is submodular.
\end{proposition}
\begin{proof}
Fix a feedback model, $S\subseteq V$, and $\varphi$ that is a valid feedback of $S$.
Let $\overline{S}$ be the set of infected vertices indicated by the feedback of $S$.
Formally, $\overline{S}$ is the set of all vertices that are reachable from $S$ by only using edges $e$ with $\varphi(e)=\live$.

We consider a new triggering model $I_{G',F'}$ defined as follows:
\begin{itemize}
    \item $G'$ shares the same vertex set with $G$;
    \item The edge set of $G'$ is obtained by removing all edges $e$ in $G$ with $\varphi(e)\neq\unknown$;
    \item The distribution $\mathcal{F}_v'$ is normalized from $\mathcal{F}_v$. Specifically, for each $T_v\subseteq\Gamma(v)$, let $p(T_v)$ be the probability that $T_v$ is chosen as the triggering set under $\mathcal{F}_v$. Let $\Gamma'(v)$ be the set of $v$'s in-neighbors in $G'$, and we have $\Gamma'(v)\subseteq\Gamma(v)$ by our construction. Then, $\mathcal{F}_v'$ is defined such that $T_v\subseteq\Gamma'(v)$ is chosen as the triggering set with probability $p(T_v)/\sum_{T_v'\subseteq\Gamma'(v)}p(T_v')$.
\end{itemize}

A simple coupling argument reveals that
\begin{equation}\label{eqn:coupling}
    \mathcal{T}(X)=\E_{\phi\simeq\varphi}\left[\left|I_{G,F}^\phi(S\cup X)\right|\right]=
    \sigma_{G',F'}(\overline{S}\cup X).
\end{equation}
We define a coupling of a realization $\phi$ of $G$ with $\phi\simeq\varphi$ to a realization $\phi'$ of $G'$ in a natural way: $\phi(e)=\phi'(e)$ for all edges $e$ in $G'$.
From our construction of $F'=\{\mathcal{F}_v'\}$, it is easy to see that, when $\phi$ is coupled with $\phi'$, the probability that $\phi$ is sampled under $I_{G,F}$ conditioning on $\phi\simeq\varphi$ equals  the probability that $\phi'$ is sampled under $I_{G',F'}$.
Under this coupling, it is easy to see that $u$ is reachable from $S$ by live edges under $\phi$ if and only if it is reachable from $\overline{S}$ by live edges under $\phi'$.
This proves Equation~(\ref{eqn:coupling}).

Finally, by Theorem~\ref{thm:submodular}, $\sigma_{G',F'}(\cdot)$ is submodular.
Therefore, for any two vertex sets $A,B$ with $A\subseteq B$ and any $u\notin B$, 
$$\mathcal{T}(A\cup\{u\})-\mathcal{T}(A)=\sigma_{G',F'}(\overline{S}\cup A\cup\{u\})-\sigma_{G',F'}(\overline{S}\cup A)$$
is weakly larger than 
$$\mathcal{T}(B\cup\{u\})-\mathcal{T}(B)=\sigma_{G',F'}(\overline{S}\cup B\cup\{u\})-\sigma_{G',F'}(\overline{S}\cup B)$$
if $u\notin\overline{S}$,
and
$$\mathcal{T}(A\cup\{u\})-\mathcal{T}(A)=\mathcal{T}(B\cup\{u\})-\mathcal{T}(B)=0$$
if $u\in\overline{S}$.
In both case, the submodularity of $\mathcal{T}(\cdot)$ holds.
\end{proof}

Next, we show that the marginal gain to the global influence function after selecting one more seed according to $\pi^g$ is at least $1/k$ fraction of the marginal gain of including all the vertices in $S^\ast$ as seeds.

\begin{proposition}\label{prop:marginal}
Given a triggering model $I_{G,F}$, any $S\subseteq V$, any feedback model, and any partial realization $\varphi$ that is a valid feedback of $S$, let $s=\pi^g(S,\varphi)$ be the next seed chosen by the greedy policy. 
We have
$$\E_{\phi\simeq\varphi}\left[\left|I_{G,F}^\phi(S\cup\{s\})\right|\right]-\E_{\phi\simeq\varphi}\left[\left|I_{G,F}^\phi(S)\right|\right]
\geq\frac1k\left(\E_{\phi\simeq\varphi}\left[\left|I_{G,F}^\phi(S\cup S^\ast)\right|\right]-\E_{\phi\simeq\varphi}\left[\left|I_{G,F}^\phi(S)\right|\right]\right).$$
\end{proposition}
\begin{proof}
Let $S^\ast=\{s_1^\ast,\ldots,s_k^\ast\}$.
By the greedy nature of $\pi^g$, we have
$$\forall v:\E_{\phi\simeq\varphi}\left[\left|I_{G,F}^\phi(S\cup\{s\})\right|\right]\geq\E_{\phi\simeq\varphi}\left[\left|I_{G,F}^\phi(S\cup\{v\})\right|\right],$$
and this holds for $v$ being any of $s_1^\ast,\ldots,s_k^\ast$ in particular.

Let $S^\ast_i=\{s_1^\ast,\ldots,s_i^\ast\}$ for each $i=1,\ldots,k$ and $S^\ast_0=\emptyset$, the proposition concludes from the following calculations
\begin{align*}
    &\E_{\phi\simeq\varphi}\left[\left|I_{G,F}^\phi(S\cup\{s\})\right|\right]-\E_{\phi\simeq\varphi}\left[\left|I_{G,F}^\phi(S)\right|\right]\\
 \geq&\frac1k\sum_{i=1}^k\left(\E_{\phi\simeq\varphi}\left[\left|I_{G,F}^\phi(S\cup\{s_i^\ast\})\right|\right]-\E_{\phi\simeq\varphi}\left[\left|I_{G,F}^\phi(S)\right|\right]\right)\\
 \geq&\frac1k\sum_{i=1}^k\left(\E_{\phi\simeq\varphi}\left[\left|I_{G,F}^\phi(S\cup S^\ast_{i-1}\cup \{s_i^\ast\})\right|\right]
 -\E_{\phi\simeq\varphi}\left[\left|I_{G,F}^\phi(S\cup S^\ast_{i-1})\right|\right]\right)\tag{Proposition~\ref{prop:partial_submodular}}\\
 =&\frac1k\left(\E_{\phi\simeq\varphi}\left[\left|I_{G,F}^\phi(S\cup S^\ast)\right|\right]-\E_{\phi\simeq\varphi}\left[\left|I_{G,F}^\phi(S)\right|\right]\right),
\end{align*}
where the last equality is by a telescoping sum, by noticing that $S^\ast_i=S^\ast_{i-1}\cup\{s_i^\ast\}$ and $S^\ast=S^\ast_k$.
\end{proof}

Finally, we prove the following proposition which is a more general statement than Theorem~\ref{thm:lowerbound}.

\begin{proposition}\label{prop:generallowerbound}
For a triggering model $I_{G,F}$ and any $\ell\in\Z^+$, we have $\sigma^\f(\pi^g,\ell)\geq(1-(1-1/k)^\ell)\sigma(S^\ast)$, and the same holds for the myopic feedback model.
\end{proposition}
\begin{proof}
We will only consider the full-adoption feedback model, as the proof for the myopic feedback model is identical.
We prove this by induction on $\ell$.
The base step for $\ell=1$ holds trivially by Proposition~\ref{prop:marginal} by considering $S=\emptyset$ in the proposition.

Suppose the inequality holds for $\ell=\ell_0$.
We investigate the expected marginal gain to the global influence function by selecting the $(\ell_0+1)$-th seed.
For a seed set $S\subseteq V$ with $|S|=\ell_0$ and a partial realization $\varphi$, let $P(S,\varphi)$ be the probability that the policy $\pi^g$ chooses $S$ as the first $\ell_0$ seeds and $\varphi$ is the feedback.
That is, 
$P(S,\varphi)=\Pr_{\phi\sim F}\left(\Seed^\f\left(\pi^g,\phi,\ell_0\right)=S\land\Phi_{G,F,\phi}^\f(S)=\varphi\right).$
The mentioned expected marginal gain is
\begin{align*}
    &\sigma^\f\left(\pi^g,\ell_0+1\right)-\sigma^\f\left(\pi^g,\ell_0\right)\\
 =&\sum_{S,\varphi:|S|=\ell_0}P(S,\varphi)\left(\E_{\phi\simeq\varphi}\left[\left|I_{G,F}^\phi(S\cup\{\pi^g(S,\varphi)\})\right|\right]
 -\E_{\phi\simeq\varphi}\left[\left|I_{G,F}^\phi(S)\right|\right]\right)\\
 \geq&\sum_{S,\varphi:|S|=\ell_0}P(S,\varphi)\cdot\frac1k\left(\E_{\phi\simeq\varphi}\left[\left|I_{G,F}^\phi(S\cup S^\ast)\right|\right]
 -\E_{\phi\simeq\varphi}\left[\left|I_{G,F}^\phi(S)\right|\right]\right)\tag{Proposition~\ref{prop:marginal}}\\
 \geq&\sum_{S,\varphi:|S|=\ell_0}P(S,\varphi)\cdot\frac1k\left(\E_{\phi\simeq\varphi}\left[\left|I_{G,F}^\phi(S^\ast)\right|\right]
 -\E_{\phi\simeq\varphi}\left[\left|I_{G,F}^\phi(S)\right|\right]\right)\\
 =&\frac1k\sigma(S^\ast)-\frac1k\sigma^\f(\pi^g,\ell_0),
\end{align*}
where the last equality follows from the law of total probability.

By rearranging the above inequality and the induction hypothesis,
\begin{align*}
    \sigma^\f\left(\pi^g,\ell_0+1\right)&\geq\frac1k\sigma(S^\ast)+\frac{k-1}k\sigma^\f\left(\pi^g,\ell_0\right)\\
    &\geq\left(\frac1k+\frac{k-1}k\left(1-\left(1-\frac1k\right)^{\ell_0}\right)\right)\sigma(S^\ast)\\
    &=\left(1-\left(1-\frac1k\right)^{\ell_0+1}\right)\sigma(S^\ast),
\end{align*}
which concludes the inductive step.
\end{proof}

By taking $\ell=k$ and noticing that $1-(1-1/k)^k>1-1/e$, 
it is easy to see that Proposition~\ref{prop:generallowerbound} implies Theorem~\ref{thm:lowerbound}.

Finally, putting Theorem~\ref{thm:lowerbound}, Lemma~\ref{lem:tightICM} and Lemma~\ref{lem:tightLTM} together, Theorem~\ref{thm:inf_gap} can be concluded easily.

\begin{proof}[Proof of Theorem~\ref{thm:inf_gap}]
Since \ICM and \LTM are special cases of triggering models,
we have
$$\inf_{G,F,k:\mbox{ }I_{G,F}\text{ is \ICM}}\frac{\sigma^\f(\pi^g,k)}{\sigma(S^g(k))}
\geq
\inf_{G,F,k}\frac{\sigma^\f(\pi^g,k)}{\sigma(S^g(k))}
$$
and
$$
\inf_{G,F,k:\mbox{ }I_{G,F}\text{ is \LTM}}\frac{\sigma^\f(\pi^g,k)}{\sigma(S^g(k))}
\geq
\inf_{G,F,k}\frac{\sigma^\f(\pi^g,k)}{\sigma(S^g(k))}.$$
Lemma~\ref{lem:tightICM} and Lemma~\ref{lem:tightLTM} show that both
$$\inf_{G,F,k:\mbox{ }I_{G,F}\text{ is \ICM}}\frac{\sigma^\f(\pi^g,k)}{\sigma(S^g(k))}
\qquad\mbox{and}\qquad
\inf_{G,F,k:\mbox{ }I_{G,F}\text{ is \LTM}}\frac{\sigma^\f(\pi^g,k)}{\sigma(S^g(k))}
$$
are at most $1-1/e$.
On the other hand, Theorem~\ref{thm:lowerbound} implies
$$\frac{\sigma^\f(\pi^g,k)}{\sigma(S^g(k))}\geq\frac{\sigma^\f(\pi^g,k)}{\sigma(S^\ast)}\geq1-\frac1e$$
for any triggering model $I_{G,F}$ and any $k$, where $S^\ast$, as usual, denotes the optimal seeds in the non-adaptive setting.

Putting these together, Theorem~\ref{thm:inf_gap} concludes for the full-adoption feedback model.
Since all those inequalities hold for the myopic feedback model as well, Theorem~\ref{thm:inf_gap} concludes for all feedback models.
\end{proof}

\section{Supremum of Greedy Adaptivity Gap}
\label{sect:sup}
In this section, we show that, for the full-adoption feedback model, both the adaptivity gap and the supremum of the greedy adaptivity gap are unbounded.
As a result, in some cases, the adaptive version of the greedy algorithm can perform significantly better than its non-adaptive counterpart.

\begin{theorem}\label{thm:sup_gap}
The greedy adaptivity gap with full-adoption feedback is unbounded: there exists a triggering model $I_{G,F}$ and $k$ such that
$$\frac{\sigma^\f(\pi^g,k)}{\sigma(S^g(k))}=2^{\Omega(\log\log |V|/\log\log\log |V|)}.$$
\end{theorem}

\begin{theorem}\label{thm:adaptivityGap}
The adaptivity gap for the general triggering model with full-adoption feedback is infinity.
\end{theorem}

In Section~\ref{sect:prescribed}, we consider a variant of \infmax such that the seeds can only be chosen among a prescribed vertex set $\overline{V}\subseteq V$, where $\overline{V}$ is specified as an input to the \infmax instance.
We show that, under this setting with \LTM, both the adaptivity gap and the supremum of the greedy adaptivity gap with the full-adoption feedback model are unbounded (Lemma~\ref{lem:ltm_precribed}).
Since it is common in practice that only a subset of nodes in a network is visible or accessible to the seed-picker, Lemma~\ref{lem:ltm_precribed} is also interesting on its own.
In Section~\ref{sect:proofs}, we show that how Lemma~\ref{lem:ltm_precribed} can be used to prove Theorem~\ref{thm:sup_gap} and Theorem~\ref{thm:adaptivityGap}.
Notice that Theorem~\ref{thm:sup_gap} and Theorem~\ref{thm:adaptivityGap} hold for the standard \infmax setting without a prescribed set of seed candidates, but we do not know if they hold for \LTM (instead, they are for the more general triggering model).

We first present the following lemma revealing a special additive property for \LTM, which will be used later.
\begin{lemma}\label{lem:additive}
Suppose $I_{G,F}$ is \LTM. If $U_1,U_2\subseteq V$ with $U_1\cap U_2=\emptyset$ satisfy that there is no path from any vertices in $U_1$ to any vertices in $U_2$ and vice versa, then $\sigma(U_1)+\sigma(U_2)=\sigma(U_1\cup U_2)$.
\end{lemma}
\begin{proof}
For any seed set $S\subseteq V$, $\sigma(S)$ can be written as follows:
\begin{equation}\label{eqn:sum_g}
\sigma(S)=\sum_{\phi}\Pr(\phi\mbox{ is sampled})\cdot\left|I_{G,F}^\phi(S)\right|.
\end{equation}

For $U_1$ and $U_2$ in the lemma statement, since each vertex can only have at most one incoming live edge (in Definition~\ref{def:LTM}, each $T_v$ has size at most $1$), under any realization $\phi$, each vertex $v\in V\setminus(U_1\cup U_2)$ that is reachable from vertices in $U_1\cup U_2$ is reachable from either vertices in $U_1$ or vertices in $U_2$, but not both.
Therefore, $|I_{G,F}^\phi(U_1)|+|I_{G,F}^\phi(U_2)|=|I_{G,F}^\phi(U_1\cup U_2)|$ for any $\phi$, and the lemma follows from considering the decomposition of $\sigma(U_1)$ and $\sigma(U_2)$ according to (\ref{eqn:sum_g}).
\end{proof}

\subsection{On \LTM with Prescribed Seed Candidates}
\label{sect:prescribed}
\begin{definition}\label{def:infmax_v}
  The \emph{influence maximization problem with prescribed seed candidates} is an optimization problem which takes as inputs $G=(V,E)$, $F$, $k\in\Z^+$, and $\overline{V}\subseteq V$, and outputs a seed set $S\subseteq\overline{V}$ that maximizes the expected total number of infections: $S\in\argmax_{S \subseteq\overline{V}: |S| \leq k} \sigma(S)$.
  The \emph{adaptive influence maximization problem with prescribed seed candidates} has the same definition as it is in Definition~\ref{def:infmax_adaptive}, with the exception that the range of the function $\pi$ is now $\overline{V}$, and $\Pi$ is the set of all such policies.
\end{definition}

\begin{lemma}\label{lem:ltm_precribed}
For \infmax with prescribed seed candidates with \LTM and the full-adoption feedback, the adaptivity gap is infinity, and the greedy adaptivity gap is $2^{\Omega(\log |V|/\log\log|V|)}$.
\end{lemma}
\begin{proof}
For $d,W\in\Z^+$ with $d$ being sufficiently large and $W\gg d^{d+1}$, we construct the following (adaptive) \infmax instance with prescribed seed candidates:
\begin{itemize}
    \item the edge-weighted graph $G=(V,E,w)$ consists of a $(d+1)$-level directed full $d$-ary tree with the root node being the sink
     (i.e., an in-arborescence) and $W$ vertices each of which is connected \emph{from} the root node of the tree; the weight of each edge in the tree is $1/d$, and the weight of each edge connecting from the root to those $W$ vertices is $1$;
    \item the number of seeds is given by $k=2(\frac{d+1}2)^d$;
    \item the prescribed set for seed candidates $\overline{V}$ is the set of all the leaves in the tree.
\end{itemize}

Since the leaves are not reachable from one to another, Lemma~\ref{lem:additive} indicates that choosing any $k$ vertices among $\overline{V}$, i.e., the leaves, infects the same number of vertices in expectation.
It is easy to see that a single seed among the leaves will infect the root node with probability $1/d^d$, and those $W$ vertices will be infected with probability $1$ if the root of the tree is infected.
Thus, for any seed set $S\subseteq\overline{V}$, by assuming all vertices in the tree are infected (in the sake of finding an upper bound for $\sigma(S)$), we have
$\sigma(S)\leq\frac1{d^d}\cdot|S|\cdot W+\sum_{i=0}^dd^i<\frac{|S|W}{d^d}+d^{d+1}.$
This gives an upper bound for the performance of both the non-adaptive greedy algorithm and the non-adaptive optimal seed set.

Now, we analyze the seeds chosen by the greedy adaptive policy.
At a particular iteration when executing the greedy adaptive policy, we classify the internal tree nodes (i.e., the nodes that are neither leaves nor the root) into the following three types:
\begin{itemize}
    \item Unexplored: the subtree rooted at this internal node contains no seed.
    \item Explored: the subtree rooted at this internal node contains seeds, and no edge in the path connecting this internal node to the root is known to be blocked (i.e., all edges in the path have statuses either \live\xspace or \unknown).
    \item Dead: if an edge in the path connecting this internal node to the root is known to be blocked.
\end{itemize}

Here we give some intuitions for the behavior of the greedy adaptive policy.
Our objective is to infect the root, which will infect those $W$ vertices that constitute most vertices of the graph.
Before the root is infected, once an internal node is known to be ``dead'', the policy should never choose any seed from the leaves that are descendants of this node, as those seeds will never have a chance to infect the root (this explains our naming).
Moreover, as we will see soon, the greedy adaptive policy will keep ``exploring'' an explored node before starting to ``exploring'' an unexplored node, until this explored node becomes dead.

We will show that, \emph{if the root node is not infected yet, at any iteration of the greedy adaptive policy, each internal level of the tree can contain at most one explored node.}
This is a formal statement describing what we meaned just now by saying that we should keep exploring an explored node.

Firstly, since only one seed can be chosen at a single iteration, among all the nodes at a particular level of the tree, at most one of them can change the status from ``unexplored'' to ``explored''.
Suppose for the sake of contradiction that, at a particular iteration of the greedy adaptive policy, an internal node $v'$ which is previously unexplored become explored, while there is already another explored node $v$ at the same level of $v'$.
Suppose this is the first iteration we see two explored nodes at the same level.
Let $u$ be the least common ancestor of $v$ and $v'$.
Let $\ell_u$ be the level containing $u$.
It is easy to see that all the nodes on the path from $v$ to the root, which includes $u$, are explored (they cannot be unexplored, as the descendants of each of those nodes contains the descendants of $v$, which contain seeds; they cannot be dead, for otherwise $v$ is dead).
Since $v$ and $v'$ are the first pair of explored nodes at the same level, before the iteration where $v'$ is explored, all nodes on the path between $v'$ and $u$ are unexplored (excluding $u$).
Let $d_u$ be the number of $u$'s children that are not dead.
Given the feedback from previous iterations, since all the descendants of $v'$ and all the nodes on the path between $v'$ and $u$ (excluding $u$) are unexplored, the probability that a seed from a leaf that is a descendant of $v'$ infects $u$ is $\frac1{d^{\ell_u-1}\cdot d_u}$.
On the other hand, if at this same iteration we pick a seed from a leaf which is a descendant of $v$ and the path from this leaf to $v$ contains no blocked edge, the probability that this seed infects $u$ is at least $\frac1{d^{\ell_u-2}(d-1)d_u}$.
This is because there is at least one dead node that is a descendant of $v$ (we know that all the nodes on the path between $v$ and the root are explored and uninfected, and we know that seeds have been chosen among the leaves on the subtree rooted at $v$; the only reason that those seeds have not made the root infected is that there are dead nodes that ``block the cascade'', and we know there is no dead node on the path between $v$ and the root).
Since the only way that a seed corresponding to either $v$ or $v'$ can infect the root is to first infect $u$ and we have $\frac1{d^{\ell_u-2}(d-1)d_u}>\frac1{d^{\ell_u-1}\cdot d_u}$, the marginal influence of a seed corresponding to $v'$ is smaller than the marginal influence of a seed corresponding to $v$.
In other words, ``exploring'' $v'$ provide less marginal influence than ``exploring'' $v$, which leads to the desired contradiction.

Next, we evaluate the expected number of seeds required to infect the root, under the greedy adaptive policy.
Suppose the tree only has two levels (i.e., a star).
The number of seeds among the leaves required to infect the root is a random variable with uniform distribution on $\{1,\ldots,d\}$, with expectation $\frac{d+1}2$.
We will show that, by induction on the number of levels of the tree, with a $d$-level tree as it is in our case, the expected number of seeds required to infect the root is $(\frac{d+1}2)^d$, which equals to $\frac{k}2$.
Let $x_1,\ldots,x_d$ be the $d$ children of the root node.
By the claim we showed just now, at most one of $x_1,\ldots,x_d$ can be ``explored'' at any iteration.
The greedy adaptive policy will do the following: it first explores one of $x_1,\ldots,x_d$, say, $x_1$; it will continue exploring $x_1$ until $x_1$ is dead or until the root is infected.
The only situation that $x_1$ is dead is that $x_1$ is infected but the edge between $x_1$ and the root is blocked.
Therefore, the greedy adaptive policy will attempt to infect $x_1,x_2,x_3,\ldots$ one by one, until one of those children infects the root.
By the induction hypothesis, the expected number of seeds required to infect each of $x_1,\ldots,x_d$ is $(\frac{d+1}2)^{d-1}$.
Let $X$ be the random variable indicating the smallest $d'$ such that $x_{d'}$ is in the triggering set of the root (this means that the greedy adaptive policy will need to infect $x_1,\ldots,x_{d'}$ in order to infect the root).
Then the expect number of seeds required to infect the root is
$$\sum_{d'=1}^d\left(\Pr(X=d')\cdot d'\left(\frac{d+1}2\right)^{d-1}\right)=\left(\frac{d+1}2\right)^{d},$$
where $d'(\frac{d+1}2)^{d-1}$ is the expected number of seeds required to infected all of $x_1,\ldots,x_{d'}$ by the linearity of expectation.

After proving that the expected number of seeds required to infect the root is $(\frac{d+1}2)^{d}=\frac k2$, by Markov's inequality, the $k$ seeds chosen according to the greedy adaptive policy will infect the root with probability at least $1/2$.
Therefore, $\sigma^\f(\pi^g,k)\geq\frac12W$, and the optimal adaptive policy can only be better: $\max_{\pi\in\Pi}\sigma^\f(\pi,k)\geq\sigma^\f(\pi^g,k)\geq\frac12W$.

Putting these together, both the adaptivity gap and the supremum of the greedy adaptivity gap is at least
$$\frac{\frac12W}{\frac{kW}{d^d}+d^{d+1}}=\frac{\frac12W}{\frac1{2^{d-1}}(1+\frac1d)^dW+d^{d+1}\cdot}=\Omega\left(2^{d}\right),$$
if setting $W=d^{d+10}\gg d^{d+1}$.
The lemma concludes by noticing $d=\Omega(\frac{\log|V|}{\log\log|V|})$ (in particular, $|V|=W+o(W)=d^{d+10}+o(d^{d+10})$, so $\log|V|=\Theta(d\log d)+o(d\log d)$, $\log\log|V|=\Theta(\log d) +o(\log d)$, and $d=\Omega(\frac{\log|V|}{\log\log|V|})$).
\end{proof}

\subsection{Proof of Theorem~\ref{thm:sup_gap}, \ref{thm:adaptivityGap}}
\label{sect:proofs}
To prove Theorem~\ref{thm:sup_gap} and Theorem~\ref{thm:adaptivityGap}, we construct an \infmax instance with a special triggering model $I_{G,F}$ which is a combination of \ICM and \LTM.

\begin{definition}\label{def:LTMc}
The \emph{mixture of \ICM and \LTM} is a triggering model $I_{G,F}$ where $G=(V,E,w)$ is an edge-weighted graph with $w(u,v)\in(0,1]$ for each $(u,v)\in E$ and each vertex $v$ is labelled either \textbf{IC} or \textbf{LT} such that $T_v$ is sampled according to $\mathcal{F}_v$ described in Definition~\ref{def:ICM} if $v$ is labelled \textbf{IC} and $T_v$ is sampled according to $\mathcal{F}_v$ described in Definition~\ref{def:LTM} if $v$ is labelled \textbf{LT}.
In addition, each vertex $v$ labelled \texttt{L} satisfies $\sum_{u\in\Gamma(v)}w(u,v)\leq1$.
\end{definition}

To conclude Theorem~\ref{thm:sup_gap} and Theorem~\ref{thm:adaptivityGap}, we construct an edge-weighted graph $G=(V,E,w)$ on which the greedy adaptive policy significantly outperforms the non-adaptive greedy algorithm.
Let $M\gg d^{d+1}$ be a large integer.
We reuse the graph with a tree and $W$ vertices in the proof of Lemma~\ref{lem:ltm_precribed}.
We create $M$ such graphs and name them $T_1,\ldots,T_M$.
Let $L=d^d$ be the number of leaves in each $T_i$.
Let $\Z_L=\{1,\ldots,L\}$ and $\Z_L^M$ be the set of all $M$-dimensional vectors whose entries are from $\Z_L$.
For each $\z=(z_1,\ldots,z_M)\in\Z_L^M$, create a vertex $a_{\z}$ and create a directed edge from $a_{\z}$ to the $z_i$-th leaf of the tree $T_i$ for each $i=1,\ldots,M$.
The weight of each such edge is $1$.
Let $A=\{a_{\z}\mid\z\in\Z_L^M\}$.
Notice that $|A|=L^M$ and each $a_{\z}\in A$ is connected to $M$ vertices from $T_1,\ldots,T_M$ respectively.
The leaves of $T_1,\ldots,T_M$ are labelled as \textbf{IC}, and the remaining vertices are labelled as \textbf{LT}.
Finally, set $k=2(\frac{d+1}2)^d$ as before.

Due to that $M$ is large, it is more beneficial to seed a vertex in $A$ than a vertex elsewhere.
In particular, seeding a root in certain $T_i$ infects $W$ vertices, while seeding a vertex in $A$ will infects $M\cdot(\frac1d)^dW\gg W$ vertices in expectation.

It is easy to see that, in the non-adaptive setting, the optimal seeding strategy is to choose $k$ seeds from $A$ such that they do not share any out-neighbors, in which case the $k$ chosen seeds will cause the infection of exactly $k$ leaves in each $T_i$.
This is also what the non-adaptive greedy algorithm will do.
As before, to find an upper bound for any seed set $S$ with $|S|=k$, we assume that all vertices in each $T_i$ are infected, and we have
$\sigma(S)\leq M\left(k\cdot\frac1{d^d}W+\sum_{i=0}^dd^i\right).$

By the same analysis in the proof of Lemma~\ref{lem:ltm_precribed}, by choosing $k$ seeds among $A$ as described above, which is equivalent as choosing $k$ leaves in each of $T_1,\ldots,T_M$ simultaneously, the root in each $T_i$ is infected with probability at least $\frac12$.
Therefore, the expected total number of infected vertices is at least
$M\cdot\frac12W$.

It may seem problematic that the greedy adaptive policy may start to seed the roots among $T_1,\ldots,T_M$ when it sees that there are already a lot of infected roots (so seeding a root is better than seed a vertex in $A$). 
However, since $M\gg d^{d+1}$, by simple calculations, this can only happen when there are already $(1-o(1))M$ trees with infected roots, in which case the number of infected vertices is already much more than $M\cdot\frac12 W$.

Putting these together as before, both the adaptivity gap and the supremum of the greedy adaptivity gap is at least
$$\frac{M\cdot\frac12W}{M(\frac{kW}{d^d}+d^{d+1})}=\frac{\frac12W}{\frac1{2^{d-1}}(1+\frac1d)^dW+d^{d+1}\cdot}=\Omega\left(2^d\right),$$
if fixing $W=d^{d+10}\gg d^{d+1}$.
Theorem~\ref{thm:adaptivityGap} concludes by letting $d\rightarrow\infty$.
To conclude Theorem~\ref{thm:sup_gap}, we need to show that $d=\Omega(\log\log|V|/\log\log\log|V|)$.
To see this, we set $M=d^{d+10}$ which is sufficiently large for our purpose.
Since we have $L=d^d$, we have $|V|=L^M+o(L^M)=d^{d^{d+11}}+o(d^{d^{d+11}})$, which implies $d=\Omega(\log\log|V|/\log\log\log|V|)$.

\section{Greedy Algorithms in Practice and Robustness of Our Results}
\label{sect:robustness}
Recall that, in the greedy algorithm, we find a vertex $s$ that maximizes the marginal gain of the influence $\sigma(S\cup\{s\})-\sigma(S)$ in each iteration.
However, computing the function $\sigma(\cdot)$ is known to be \#P-hard for both \ICM~\shortcite{chen2010scalable2} and \LTM~\shortcite{chen2010scalable}.
In practice, the greedy algorithm is implemented with $\sigma(\cdot)$ estimated by Monte Carlo simulations (where $\sigma(\cdot)$ is approximated by sampling a sufficient number of realizations $\phi$, computing $|I_{G,F}^\phi(S)|$ for each realization $\phi$, and taking an average over them), reverse reachable sets coverage (see Section~\ref{sect:RR_set} for details), or other \emph{randomized approximation algorithms}.
As a result, in reality, when implementing the greedy algorithm, a vertex $s$ that \emph{approximately} maximizes the marginal gain of the influence is found in each iteration \emph{with high probability}.
In this section, we discuss the applicability of all our results in previous sections under this approximation setting.
In Section~\ref{sect:epsilon_delta_greedy}, 
We first define the \emph{$(\varepsilon,\delta)$-greedy algorithm} where in each iteration a vertex $s$ that approximately maximizes the marginal gain $\sigma(S\cup\{s\})-\sigma(S)$ within factor $(1-\varepsilon)$ is found with probability at least $(1-\delta)$, which captures the practical implementations of greedy algorithms.
In Section~\ref{sect:epsilon_delta_greedy_gap}, we discuss the robustness of our results by studying under what $\varepsilon$ and $\delta$ our results hold.

\subsection{$(\varepsilon,\delta)$-Greedy Algorithms}
\label{sect:epsilon_delta_greedy}
\begin{definition}\label{epsilon_delta_greedy}
  An \emph{$(\varepsilon,\delta)$-greedy algorithm} is a randomized iterative algorithm that satisfies the following:
  \begin{enumerate}
      \item the algorithm initializes $S=\emptyset$;
      \item for each of the $k$ iterations, with probability at least $1-\delta$, the algorithm finds $s\in V$ such that
      $$\sigma(S\cup\{s\})-\sigma(S)\geq(1-\varepsilon)\max_{s'\in V}\left(\sigma(S\cup\{s'\})-\sigma(S)\right),$$
      and update $S\leftarrow S\cup\{s\}$;
      \item the algorithm outputs $S$.
  \end{enumerate}
\end{definition}

Since an $(\varepsilon,\delta)$-greedy algorithm is an approximation version of the ``exact'' greedy algorithm, it achieves an approximation ratio that is close to $(1-1/e)$.
The proof is standard, and we include it here for completeness.
\begin{theorem}
For any $\varepsilon\leq\frac1k$, an $(\varepsilon,\frac\delta{k})$-greedy algorithm gives a $(1-1/e-\varepsilon)$-approximation for submodular \infmax with probability $(1-\delta)$.
\end{theorem}
\begin{proof}
Let $S^\ast=\{s_1^\ast,\ldots,s_k^\ast\}$ be an optimal solution which maximizes $\sigma(\cdot)$, and let $S=\{s_1,\ldots,s_k\}$ be the seed set output by any $(\varepsilon,\delta/k)$-greedy algorithm with $\varepsilon\leq1/k$.
Let $S^\ast_i=\{s_1^\ast,\ldots,s_i^\ast\}$ and $S_i=\{s_1,\ldots,s_i\}$.
In particular, let $S^\ast_0=S_0=\emptyset$.
Similar to Proposition~\ref{prop:marginal}, we will show that, for each $i=0,1,\ldots,k-1$,
\begin{equation}\label{eqn:each_step_epsilon_delta_greedy}
    \sigma(S_{i+1})-\sigma(S_i)\geq\frac{1-\varepsilon}k(\sigma(S_i\cup S^\ast)-\sigma(S_i))\qquad\mbox{with probability at least }1-\frac\delta{k}.
\end{equation}
This is because
\begin{align*}
    \sigma(S_{i+1})-\sigma(S_i)&\geq(1-\varepsilon)\left(\max_{s\in V}\left(\sigma(S_{i}\cup\{s\})-\sigma(S_{i})\right)\right)\tag{by definition of $(\varepsilon,\delta)$-greedy}\\
    &\geq(1-\varepsilon)\frac1k\sum_{j=1}^k\left(\sigma(S_{i}\cup\{s_j^\ast\})-\sigma(S_{i})\right)\tag{since $s$ is the maximizer}\\
    &\geq\frac{1-\varepsilon}k\sum_{j=1}^k\left(\sigma(S_{i}\cup S_j^\ast )-\sigma(S_{i}\cup S_{j-1}^\ast)\right)\tag{submodularity of $\sigma(\cdot)$}\\
    &=\frac{1-\varepsilon}k\left(\sigma(S_{i}\cup S^\ast)-\sigma(S_{i})\right).\tag{telescoping sum}
\end{align*}
Next, similar to Proposition~\ref{prop:generallowerbound}, we can prove by induction that for each $i=1,\ldots,k$
\begin{equation}\label{eqn:induction_epsilon_delta_greedy}
    \sigma(S_i)\geq\left(1-\left(1-\frac1k\right)^i-\varepsilon\right)\sigma(S^\ast)\qquad\mbox{with probability at least }1-\frac{i\delta}k.
\end{equation}
The base step for $i=1$ follows from Equation~(\ref{eqn:each_step_epsilon_delta_greedy}):
$$\sigma(S_1)=\sigma(S_1)-\sigma(S_0)\geq\frac{1-\varepsilon}k(\sigma(S_0\cup S^\ast)-\sigma(S_0))>\left(\frac1k-\varepsilon\right)\sigma(S^\ast).$$
For the inductive step, by Equation~(\ref{eqn:each_step_epsilon_delta_greedy}) again, we have, with probability at least $(1-\delta/k)$,
$$\sigma(S_{i+1})-\sigma(S_i)\geq\frac{1-\varepsilon}k(\sigma(S_i\cup S^\ast)-\sigma(S_i))\geq\frac{1-\varepsilon}k(\sigma(S^\ast)-\sigma(S_i)),$$
which implies
$$\sigma(S_{i+1})\geq\frac{1-\varepsilon}k\sigma(S^\ast)+\frac{k-1+\varepsilon}k\sigma(S_i).$$
By the induction hypothesis, with probability at least $(1-\frac{i\delta}k)$, we have $$\sigma(S_i)\geq\left(1-\left(1-\frac1k\right)^i-\varepsilon\right)\sigma(S^\ast).$$
Putting these together by a union bound, with probability at least $1-\frac{(i+1)\delta}k$, we have
\begin{align*}
    \sigma(S_{i+1})&\geq\frac{1-\varepsilon}k\sigma(S^\ast)+\frac{k-1+\varepsilon}k\left(1-\left(1-\frac1k\right)^i-\varepsilon\right)\sigma(S^\ast)\\
    &=\left(1-\left(1-\frac1k\right)^{i+1}-\varepsilon+\frac\varepsilon{k}\left(1-\varepsilon-\left(1-\frac1k\right)^i\right)\right)\sigma(S^\ast)\tag{elementary calculations}\\
    &\geq\left(1-\left(1-\frac1k\right)^{i+1}-\varepsilon\right)\sigma(S^\ast),\tag{since $\varepsilon\leq\frac1k$ and $i\geq1$}
\end{align*}
which concludes the inductive step.

The theorem concludes by taking $i=k$ in Equation~(\ref{eqn:induction_epsilon_delta_greedy}) and noticing that $1-(1-1/k)^k>1-1/e$.
\end{proof}

We can define the \emph{$(\varepsilon,\delta)$-greedy adaptive policy} similarly.
\begin{definition}\label{def:epsilon_delta_greedy_adaptive}
  An adaptive policy $\pi$ is a \emph{$(\varepsilon,\delta)$-greedy adaptive policy} if, for any $S\subseteq V$ and any partial realization $\varphi$, with probability at least $1-\delta$, we have $\pi(S,\varphi)=v$ for $v$ such that
  $$\E_{\phi\simeq\varphi}\left[\left|I_{G,F}^\phi\left(S\cup\{v\}\right)\right|-\left|I_{G,F}^\phi\left(S\right)\right|\right]\geq(1-\varepsilon)\max_{s\in V}\E_{\phi\simeq\varphi}\left[\left|I_{G,F}^\phi\left(S\cup\{s\}\right)\right|-\left|I_{G,F}^\phi\left(S\right)\right|\right].$$ 
\end{definition}

\subsection{Greedy Adaptivity Gap for $(\varepsilon,\delta)$-Greedy Algorithms}
\label{sect:epsilon_delta_greedy_gap}
We re-exam all the theorems and lemmas in Section~\ref{sect:inf} and Section~\ref{sect:sup}.
Throughout this section, we use $\alg^g_{\varepsilon,\delta}$ to denote the set of all $(\varepsilon,\delta)$-greedy algorithms and
$\Pi^g_{\varepsilon,\delta}$ to denote the set of all $(\varepsilon,\delta)$-greedy adaptive policies.
Since the greedy adaptive policy we are studying now is randomized, for any $(\varepsilon,\delta)$-greedy adaptive policy $\pi^g\in\Pi^g_{\varepsilon,\delta}$, the values $\sigma^\f(\pi^g,k)$ and $\sigma^\m(\pi^g,k)$ are the expected numbers of infected vertices under the full-adoption feedback setting and the myopic feedback setting respectively, where the expectation is taken over \emph{both} the sampling of a realization \emph{and the randomness when implementing $\pi^g$}.
Correspondingly, for a randomized non-adaptive $(\varepsilon,\delta)$-greedy algorithm $A^g\in\alg^g_{\varepsilon,\delta}$, we will slightly abuse the notation and use $\sigma(A^g,k)$ to denote the expected number of infected vertices when $k$ seeds are chosen based on algorithm $A^g$, where the expectation is again taken over both the sampling of a realization and the randomness when implementing $A^g$.

The argument behind all the proofs in this section is the same, which we summarize as follows.
To show that the greedy adaptivity gap remains the same in the $(\varepsilon,\delta)$-greedy setting, we find an $\varepsilon$ that is small enough such that, in the \infmax instance we constructed, requiring the marginal influence being at least a $(1-\varepsilon)$ fraction of the maximum marginal influence is the same as requiring the maximum marginal influence.
This is done by setting $\varepsilon$ to be small enough such that the only seed that produces the marginal influence within $(1-\varepsilon)$ of the maximum marginal influence is the seed that produce the maximum marginal influence.
By definition, the $(\varepsilon,\delta)$-greedy algorithm/policy will behave exactly the same as their exact deterministic counterpart with probability at least $1-\delta$.
By setting $\delta=o(1/k)$ and taking a union bound over all the $k$ iterations, with probability at least $1-o(1)$, the $(\varepsilon,\delta)$-greedy algorithm/policy will behave the same way as the exact deterministic greedy algorithm/policy.

\subsubsection{Infimum of Greedy Adaptivity Gap for $(\varepsilon,\delta)$-Greedy Algorithms}
In this section, we show that, for $\varepsilon=o(1/k)$ and $\delta=o(1/k)$, the infimum of the greedy adaptivity gap for $(\varepsilon,\delta)$-greedy algorithms is between $1-1/e-\varepsilon$ and $1-1/e$.
We will formally state what exactly we mean by this, and we will prove this by showing that Lemma~\ref{lem:tightICM}, Lemma~\ref{lem:tightLTM} and Theorem~\ref{thm:lowerbound} can be adapted in the $(\varepsilon,\delta)$-greedy setting.

The following lemma extends Lemma~\ref{lem:tightICM} to the $(\varepsilon,\delta)$-greedy setting.
\begin{lemma}\label{lem:tightICM_epsilon_delta}
Given any two functions $\varepsilon:\Z^+\to\R^+$ and $\delta:\Z^+\to\R^+$ satisfying $\varepsilon(k)=o(1/k)$ and $\delta(k)=o(1/k)$, for any $\tau>0$, there exists $G,F,k$ such that $I_{G,F}$ is an \ICM and, for any adaptive policy $\pi^g\in\Pi^g_{\varepsilon(k),\delta(k)}$ and any non-adaptive algorithm $A^g\in\alg^g_{\varepsilon(k),\delta(k)}$, we have
$$\frac{\sigma^\f(\pi^g,k)}{\sigma(A^g,k)}\leq1-\frac1e+\tau\qquad\mbox{and}\qquad\frac{\sigma^\m(\pi^g,k)}{\sigma(A^g,k)}\leq1-\frac1e+\tau.$$
\end{lemma}
\begin{proof}
We construct the same \infmax instance $(G=(V,E,w),k+1)$ as it is in the proof of Lemma~\ref{lem:tightICM}, with only one change: set $\Upsilon=\varepsilon(k+1)\cdot W$ instead of the previous setting $\Upsilon=W/k^2$.\footnote{If $\varepsilon(k+1)\cdot W$ is not an integer, we can always make $W$ large enough and find a positive rational number $\varepsilon'<\varepsilon(k+1)$ such that $\varepsilon'W\in\Z^+$. The remaining part of the proof will not be invalidated by this change.}
Notice that $\varepsilon(k+1)$ means the value of the function $\varepsilon(\cdot)$ with input $k+1$, not $\varepsilon$ times $(k+1)$.
To avoid possible confusion, we write $\varepsilon:=\varepsilon(k+1)$ and $\delta:=\delta(k+1)$ for this proof.
The remaining part of the proof is an adaption of the proof of Lemma~\ref{lem:tightICM} to the $(\varepsilon,\delta)$ setting.

We first show that, for any $A^g\in\alg^g_{\varepsilon,\delta}$, with probability at least $1-(k+1)\cdot\delta=1-o(1)$, $A^g$ will output $\{s,u_1,\ldots,u_k\}$.

From Equation~(\ref{eqn:sigma_s}), (\ref{eqn:sigma_ui}), (\ref{eqn:sigma_v1}) and (\ref{eqn:sigma_vi}), by the definition of $(\varepsilon,\delta)$-greedy, with probability at least $1-\delta$, the first seed chosen must have expected influence at least $(1-\varepsilon)\sigma(\{s\})\geq(1-o(1/k))\cdot2W$.
Since any other vertex does not have an influence which is even close to $2W$, the first seed chosen by $A^g$ is $s$ with probability at least $1-\delta$.

Next, we show that, if $A^g$ has chosen $s$ and $i$ vertices from $\{u_1,\ldots,u_k\}$ after $i+1$ iterations, $A^g$ will choose the next seed from $\{u_1,\ldots,u_k\}$ with probability $1-\delta$.
Let $U_i=\{s,u_1,\ldots,u_i\}$ and $U_0=\{s\}$ as before.
Without loss of generality, we only need to show that, supposing $A^g$ has chosen $U_i$ as the first $(i+1)$ seeds, with probability at least $1-\delta$, $A^g$ will choose a vertex from $\{u_{i+1},\ldots,u_k\}$ as the next seed.
By our calculation in Equation~(\ref{eqn:nonadaptivegreedy_marginal_ui+1}), (\ref{eqn:nonadaptivegreedy_marginal_v1}) and (\ref{eqn:nonadaptivegreedy_marginal_v2}), with probability at least $1-\delta$, $A^g$ will choose a seed $x$ such that $$\sigma(U_i\cup\{x\})-\sigma(U_i)\geq(1-\varepsilon)(\sigma(U_i\cup\{u_{i+1}\})-\sigma(U_i))$$
$$\qquad> W+(4k-2)\Upsilon-1.1\varepsilon W=W+(4k-3.1)\Upsilon,$$
where the last inequality is due to $\varepsilon(W+(4k-2)\Upsilon)=\varepsilon(W+o(W))<1.1\varepsilon W=1.1\Upsilon$.
On the other hand, from Equation~(\ref{eqn:nonadaptivegreedy_marginal_v1}) and (\ref{eqn:nonadaptivegreedy_marginal_v2}), choosing $v_1$ or $v_2$ as the next seed does not provide enough marginal gain to $\sigma(\cdot)$:
\begin{align*}
    \sigma(U_i\cup\{v_1\})-\sigma(U_i)&\leq1+W+(4k-4)\Upsilon<(1-\varepsilon)(\sigma(U_i\cup\{u_{i+1}\})-\sigma(U_i)),\\
    \sigma(U_i\cup\{v_2\})-\sigma(U_i)&\leq1+W-\frac{W}k+(4k+5)\Upsilon\\
    &=W+4k\Upsilon-\omega(\Upsilon)\tag{since $\Upsilon=\varepsilon W=o\left(\frac{W}k\right)$}\\
    &<(1-\varepsilon)(\sigma(U_i\cup\{u_{i+1}\})-\sigma(U_i)),
\end{align*}
and the marginal influence of the remaining vertices other than $v_1$ and $v_2$ are even smaller.
Therefore, we conclude that, with probability at least $1-\delta$, $A^g$ will choose a vertex from $\{u_{i+1},\ldots,u_k\}$ as the next seed.

Putting these together, by a union bound, with probability $1-(k+1)\delta=1-o(1)$, $A^g$ will choose $\{s,u_1,\ldots,u_k\}$, which will infected $(k+2)W+o(W)$ vertices in expectation, as calculated in the proof of Lemma~\ref{lem:tightICM}.
Therefore,
$$\sigma(A^g,k+1)\geq(1-(k+1)\delta)((k+2)W+o(W))+(k+1)\delta\cdot0=kW-o(kW).$$

Second, we prove, for any greedy adaptive policy $\pi^g\in\Pi^g_{\varepsilon,\delta}$, $\pi^g$ will choose $\{s,v_1,\ldots,v_k\}$ with probability at least $1-1/k-(k+1)\delta=1-o(1)$.
By the same analysis in the non-adaptive case, with probability $(1-\delta)$, the first seed chosen by $\pi^g$ is $s$.
We assume that $s$ fails to infect $t$ which happens with probability $1-1/k$, and this is given as the feedback to $\pi^g$.
The marginal influence of $v_1$ is then $\w(t)+\w(v_1)+\sum_{i=1}^k\w(w_{1i})=4k\Upsilon+1+W$.
With probability at least $1-\delta$, $\pi^g$ will choose a second seed with marginal influence at least $(1-\varepsilon)(4k\Upsilon+1+W)>W+4k\Upsilon-\varepsilon\cdot 1.1W=W+(4k-1.1)\Upsilon$.
It is easy to see that $v_1$ is the only vertex that can provide enough marginal influence.
In particular, the marginal influence of each of $u_1,\ldots,u_k$ is $1+W+(4k-2)\Upsilon$, which is less than $W+(4k-1.1)\Upsilon$, the marginal influence of $v_2$ is $1+(1-\frac1k)W+k\cdot\frac{4k-2}{k-1}\Upsilon=W+4k\Upsilon-\omega(\Upsilon)$ (notice that $k\cdot\frac{4k-2}{k-1}\Upsilon=4k\Upsilon+\Theta(\Upsilon)$ and $W/k=\omega(\Upsilon)$), which is less than $W+(4k-1.1)\Upsilon$, and the marginal influence of $v_3,\ldots,v_k$ are even smaller than that of $v_2$.

We have shown that the first two seeds are $s$ and $v_1$ with probability at least $1-1/k-2\delta$.
Next, we show that, for each $i=1,\ldots,k-1$, if $\pi^g$ has chosen $s,v_1,\ldots,v_i$, with probability $(1-\delta)$, $\pi^g$ will choose $v_{i+1}$ as the next seed.
Suppose $\pi^g$ has chosen $s,v_1,\ldots,v_i$.
From the proof of Lemma~\ref{lem:tightICM}, we have seen that $v_{i+1}$ has the highest marginal influence, which is $1+(1-1/k)^iW+k\cdot\frac{4k-2}{k-1}\Upsilon$.
With probability at least $1-\delta$, $\pi^g$ will choose a seed with marginal influence at least $(1-\varepsilon)$ fraction of this value:
\begin{align*}
    &(1-\varepsilon)\left(1+\left(1-\frac1k\right)^iW+k\frac{4k-2}{k-1}\Upsilon\right)\\
    >&\left(1-\frac1k\right)^iW+k\frac{4k-2}{k-1}\Upsilon-\varepsilon\left(\left(1-\frac1k\right)^iW+k\frac{4k-2}{k-1}\Upsilon\right)\\
    \geq&\left(1-\frac1k\right)^iW+k\frac{4k-2}{k-1}\Upsilon-\varepsilon\left(\left(1-\frac1k\right)^1W+k\frac{4k-2}{k-1}\Upsilon\right)\\
    \geq&\left(1-\frac1k\right)^iW+k\frac{4k-2}{k-1}\Upsilon-1.1\Upsilon\tag{since $\varepsilon(1-\frac1k)^1W<\varepsilon W=\Upsilon$ and $\varepsilon k\frac{4k-2}{k-1}\Upsilon=\Theta(4k\varepsilon\Upsilon)=o(\Upsilon)<0.1\Upsilon$}.
\end{align*} 
The marginal influence of $v_{i+2}$ is $(1-1/k)^{i+1}W+k\frac{4k-2}{k-1}\Upsilon=(1-1/k)^iW+k\frac{4k-2}{k-1}\Upsilon-\frac1k(1-1/k)^iW<(1-1/k)^iW+k\frac{4k-2}{k-1}\Upsilon-0.63\frac{W}k$, which is less than the value above (as $\Upsilon=o(W/k)$).
The marginal influence of $v_{i+3},\ldots,v_k$ are even smaller, and we do not need to consider them.
The marginal influence of $u_1$ is $\sum_{j=i+1}^k\w(w_{j1})=(1-1/k)^iW+(k-i)\cdot\frac{4k-2}{k-1}\Upsilon\leq(1-1/k)^iW+(k-1)\cdot\frac{4k-2}{k-1}\Upsilon<(1-1/k)^iW+k\frac{4k-2}{k-1}\Upsilon-2\Upsilon$, which is again less than the value above.

Putting these together, by a union bound, with probability $(1-1/k)(1-(k+1)\delta)=o(1)$, $\pi^g$ will choose $\{s,v_1,\ldots,v_k\}$, which can infect, as computed in the proof of Lemma~\ref{lem:tightICM}, $k(1-(1-1/k)^k)W+o(kW)$ vertices.
Therefore,
$$\sigma^\f(\pi^g,k)=\sigma^\m(\pi^g,k)\leq(1-o(1))\left(k\left(1-\left(1-\frac1k\right)^k\right)W+o(kW)\right)+o(1)\cdot|V|$$
$$\qquad=k\left(1-\left(1-\frac1k\right)^k\right)W+o(kW).$$
The theorem concludes by taking the ratio of the computed upper-bound of $\sigma^\f(\pi^g,k)=\sigma^\m(\pi^g,k)$ and the computed lower-bound of $\sigma(A^g,k)$, and then taking the limits $k\rightarrow\infty$ and $W\rightarrow\infty$.
\end{proof}

\begin{remark}
Lemma~\ref{lem:tightICM_epsilon_delta} provides an upper bound for each of $\frac{\sigma^\f(\pi^g,k)}{\sigma(A^g,k)}$ and $\frac{\sigma^\m(\pi^g,k)}{\sigma(A^g,k)}$, where the numerator and the denominator in each ratio represent the number of infected vertices (in the non-adaptive setting and the adaptive setting respectively) \emph{in expectation}.
The same proof for Lemma~\ref{lem:tightICM_epsilon_delta} can be used to show the following stronger version of Lemma~\ref{lem:tightICM_epsilon_delta}.
Let $I_{G,F}^\phi(\pi^g,k)$ be the set of infected vertices when the adaptive policy $\pi^g$ is used and the underlying live-edge realization is $\phi$.
Let $I_{G,F}^\phi(A^g,k)$ have similar meaning corresponding to non-adaptive algorithm $A^g$.
The same proof for Lemma~\ref{lem:tightICM_epsilon_delta} implies that, under the same setting in Lemma~\ref{lem:tightICM_epsilon_delta}, for both full-adoption and myopic feedback models, we have
\begin{equation}\label{eqn:stronger_version_lower_bound}
    \Pr_{\phi\sim F}\left(\frac{|I_{G,F}^\phi(\pi^g,k)|}{|I_{G,F}^\phi(A^g,k)|}\leq1-\frac1e+\tau\right)\geq1-2(k+1)\delta-\frac1k=1-o(1),
\end{equation}
where we have taken a union bound of the ``bad'' events that the ``correct'' seed is not chosen in each of the $(k+1)$ iterations in both $\pi^g$ and $A^g$, as well as the ``bad'' event that $s$ infects $t$ (with probability $1/k$).
\end{remark}

The following lemma extends Lemma~\ref{lem:tightLTM} to the $(\varepsilon,\delta)$-greedy setting.
\begin{lemma}\label{lem:tightLTM_epsilon_delta}
Given any two functions $\varepsilon:\Z^+\to\R^+$ and $\delta:\Z^+\to\R^+$ satisfying $\varepsilon(k)=o(1/k)$ and $\delta(k)=o(1/k)$, for any $\tau>0$, there exists $G,F,k$ such that $I_{G,F}$ is an \LTM and, for any adaptive policy $\pi^g\in\Pi^g_{\varepsilon(k),\delta(k)}$ and any non-adaptive algorithm $A^g\in\alg^g_{\varepsilon(k),\delta(k)}$, we have
$$\frac{\sigma^\f(\pi^g,k)}{\sigma(A^g,k)}\leq1-\frac1e+\tau\qquad\mbox{and}\qquad\frac{\sigma^\m(\pi^g,k)}{\sigma(A^g,k)}\leq1-\frac1e+\tau.$$
\end{lemma}
\begin{proof}
We construct the same \infmax instance $(G=(V,E,w),k+1)$ as it is in the proof of Lemma~\ref{lem:tightLTM}, with only one change: set $\Upsilon=\varepsilon\cdot W$ instead of the previous setting $\Upsilon=W/k^2$.
The remaining part of the proof is exactly the same as how we have adapted the proof of Lemma~\ref{lem:tightICM} to Lemma~\ref{lem:tightICM_epsilon_delta}.
We omit the details here.
\end{proof}

\begin{remark}
Similarly, we can prove a stronger version of Lemma~\ref{lem:tightLTM_epsilon_delta} given by exactly the same equation (\ref{eqn:stronger_version_lower_bound}).
\end{remark}

Next, it is easy to show that Theorem~\ref{thm:lowerbound} holds for the $(\varepsilon,\delta)$-greedy setting.
\begin{theorem}\label{thm:lowerbound_epsilon_delta}
For a triggering model $I_{G,F}$, any $\varepsilon\in(0,1/k]$, any function $\delta:\Z^+\to\R^+ $ such that $\delta(k)=o(1/k)$, and any $\pi^g\in\Pi^g_{\varepsilon,\delta(k)}$, we have
$$\sigma^\f(\pi^g,k)\geq\left(1-\frac1e-\varepsilon\right)\max_{S\subseteq V,|S|\leq k}\sigma(S)\quad\mbox{ and }\quad
\sigma^\m(\pi^g,k)\geq\left(1-\frac1e-\varepsilon\right)\max_{S\subseteq V,|S|\leq k}\sigma(S).$$
\end{theorem}
\begin{proof}
Let $S^\ast\in\argmax_{S\subseteq V,|S|\leq k}\sigma(S)$ be an optimal non-adaptive seed set.
For any $S\subseteq V$, any partial realization $\phi$ that is a valid feedback of $S$ under any feedback model (either full-adoption or myopic), letting $s^\ast\in\argmax_{s'}\E_{\phi\simeq\varphi}[|I^\phi_{G,F}(S\cup\{s'\})|]$ be the vertex which maximizes the expected marginal influence given $\varphi$, with probability at least $1-\delta$, $\pi^g$ will pick the next seed $s=\pi^g(S,\varphi)$ such that
\begin{align*}
&\E_{\phi\simeq\varphi}\left[\left|I_{G,F}^\phi(S\cup\{s\})\right|\right]-\E_{\phi\simeq\varphi}\left[\left|I_{G,F}^\phi(S)\right|\right]\\
\geq&(1-\varepsilon)\left(\E_{\phi\simeq\varphi}\left[\left|I_{G,F}^\phi(S\cup\{s^\ast\})\right|\right]-\E_{\phi\simeq\varphi}\left[\left|I_{G,F}^\phi(S)\right|\right]\right)\tag{Definition~\ref{def:epsilon_delta_greedy_adaptive}}\\
\geq&\frac{1-\varepsilon}k\left(\E_{\phi\simeq\varphi}\left[\left|I_{G,F}^\phi(S\cup S^\ast)\right|\right]-\E_{\phi\simeq\varphi}\left[\left|I_{G,F}^\phi(S)\right|\right]\right)\tag{Proposition~\ref{prop:marginal}}.
\end{align*}
We can then prove that, for any $\ell\in\Z^+$, both $\sigma^\f(\pi^g,\ell)$ and $\sigma^\m(\pi^g,\ell)$ are no less than $(1-(1-1/k)^\ell-\varepsilon)\sigma(S^\ast)$, by using the same arguments in the proof of Proposition~\ref{prop:generallowerbound}.
The theorem concludes by taking $\ell=k$.
\end{proof}

Finally, we formally state and prove that, when $\varepsilon=o(1/k)$ and $\delta=o(1/k)$, the infimum of the adaptivity gap under the $(\varepsilon,\delta)$-greedy setting is between $1-1/e-\varepsilon$ and $1-1/e$, which adapts Theorem~\ref{thm:inf_gap} to the $(\varepsilon,\delta)$-greedy setting.
\begin{theorem}
Given any two functions $\varepsilon:\Z^+\to\R^+$ and $\delta:\Z^+\to\R^+$ such that $\varepsilon(k)=o(1/k)$ and $\delta(k)=o(1/k)$, we have
\begingroup
\allowdisplaybreaks
\begin{eqnarray*}
    \inf_{G,F,k:I_{G,F}\text{ is \ICM}}\left(\inf_{\substack{\pi^g\in\Pi^g_{\varepsilon(k),\delta(k)}\\A^g\in\alg^g_{\varepsilon(k),\delta(k)}}}\frac{\sigma^\f(\pi^g,k)}{\sigma(A^g,k)}\right)&\in& \left(1-\frac1e-\varepsilon(k),1-\frac1e\right),\\
    \inf_{G,F,k:I_{G,F}\text{ is \ICM}}\left(\sup_{\substack{\pi^g\in\Pi^g_{\varepsilon(k),\delta(k)}\\A^g\in\alg^g_{\varepsilon(k),\delta(k)}}}\frac{\sigma^\f(\pi^g,k)}{\sigma(A^g,k)}\right)&\in& \left(1-\frac1e-\varepsilon(k),1-\frac1e\right),\\
    \inf_{G,F,k:I_{G,F}\text{ is \LTM}}\left(\inf_{\substack{\pi^g\in\Pi^g_{\varepsilon(k),\delta(k)}\\A^g\in\alg^g_{\varepsilon(k),\delta(k)}}}\frac{\sigma^\f(\pi^g,k)}{\sigma(A^g,k)}\right)&\in& \left(1-\frac1e-\varepsilon(k),1-\frac1e\right),\\
    \inf_{G,F,k:I_{G,F}\text{ is \LTM}}\left(\sup_{\substack{\pi^g\in\Pi^g_{\varepsilon(k),\delta(k)}\\A^g\in\alg^g_{\varepsilon(k),\delta(k)}}}\frac{\sigma^\f(\pi^g,k)}{\sigma(A^g,k)}\right)&\in& \left(1-\frac1e-\varepsilon(k),1-\frac1e\right),\\
    \inf_{G,F,k}\left(\inf_{\substack{\pi^g\in\Pi^g_{\varepsilon(k),\delta(k)}\\A^g\in\alg^g_{\varepsilon(k),\delta(k)}}}\frac{\sigma^\f(\pi^g,k)}{\sigma(A^g,k)}\right)&\in& \left(1-\frac1e-\varepsilon(k),1-\frac1e\right),\\
    \inf_{G,F,k}\left(\sup_{\substack{\pi^g\in\Pi^g_{\varepsilon(k),\delta(k)}\\A^g\in\alg^g_{\varepsilon(k),\delta(k)}}}\frac{\sigma^\f(\pi^g,k)}{\sigma(A^g,k)}\right)&\in& \left(1-\frac1e-\varepsilon(k),1-\frac1e\right).
\end{eqnarray*}
\endgroup
All the six statements above also hold for the myopic feedback model.
\end{theorem}
\begin{proof}
Since $\sigma(A^g,k)\leq\max_{S\subseteq V,|S|\leq k}\sigma(S)$ always hold, the lower bound $1-\frac1e-\varepsilon$ holds for each of the six statements according to Theorem~\ref{thm:lowerbound_epsilon_delta}.
Then, Lemma~\ref{lem:tightICM_epsilon_delta} implies the first two statements, Lemma~\ref{lem:tightLTM_epsilon_delta} implies the third and the fourth.
Finally, since both \ICM and \LTM are special cases of the triggering model, the left-hand side of the fifth statement is at most the left-hand side of the first (or the third), and the left-hand side of the sixth statement is at most the left-hand side of the second (or the fourth).
Thus, the first four statements imply the last two.
\end{proof}

\subsubsection{Supremum of Greedy Adaptivity Gap for $(\varepsilon,\delta)$-Greedy Algorithms}
We first show that Lemma~\ref{lem:ltm_precribed} holds for the $(\varepsilon,\delta)$-greedy setting for very mild restrictions on $\varepsilon$ and $\delta$.
the following lemma shows that, under the $(\varepsilon,\delta)$-greedy setting where $\varepsilon=O(\log\log k/\log k)$ and $\delta=o(1/k)$, the greedy adaptivity gap for \infmax with prescribed seed candidates with \LTM and full-adoption feedback is $2^{\Omega(\log|V|/\log\log|V|)}$.\footnote{Lemma~\ref{lem:ltm_precribed} also says that the adaptivity gap under the same setting is infinity. Since the adaptivity gap is about \emph{optimal} algorithm/policy which is irrelevant to the greedy algorithm, this result holds as always, and it makes no sense to ``adapt'' it into the setting in this section. Similarly, Theorem~\ref{thm:adaptivityGap} is also irrelevant here, and it holds as always.}
Notice that Lemma~\ref{lem:ltm_prescribed_epsilon_delta} below is a stronger claim: it says that the greedy adaptive policy significantly outperforms any non-adaptive \infmax algorithm, including the non-adaptive greedy algorithm and even the optimal non-adaptive algorithm.

\begin{lemma}\label{lem:ltm_prescribed_epsilon_delta}
There exists a constant $c>0$ such that, given any two functions $\varepsilon:\Z^+\to\R^+$ and $\delta:\Z^+\to\R^+$ such that $\varepsilon(k)\leq\frac{c\log\log k}{\log k}$ and $\delta(k)=o(1/k)$, for \infmax with prescribed seed candidates with \LTM, there exists $k$ such that, for any valid seed set $S$ (i.e., $S$ is a subset of the candidate set $\overline{V}$ and $|S|\leq k$) and any $\pi^g\in\Pi^g_{\varepsilon(k),\delta(k)}$, we have
$$\frac{\sigma^\f(\pi^g,k)}{\sigma(S)}\geq2^{c\log(|V|)/\log\log(|V|)}.$$
\end{lemma}
\begin{proof}
The sketch of the proof of this lemma follows the proof of Lemma~\ref{lem:ltm_precribed}.
We construct the same \infmax instance with the same $k,\overline{V},d,W$ as given in the proof of Lemma~\ref{lem:ltm_precribed}.
Again, by Lemma~\ref{lem:additive}, choosing any $k$ vertices among $\overline{V}$ infects the same number of vertices in expectation.
We can reach the same conclusion that $\sigma(S)<\frac{|S|W}{d^d}+d^{d+1}$ by the same arguments.
It then remains to analyze the greedy adaptive policy.

Consider an arbitrary greedy adaptive policy $\pi^g\in\Pi^g_{\varepsilon(k),\delta(k)}$.
Let the three status ``unexplored'', ``explored'' and ``dead'' have the same meanings as they are in the proof of Lemma~\ref{lem:additive}.
Correspondingly, we will show that, \emph{with probability at least $1-k\delta(k)=1-o(1)$, if the root node is not infected yet, at any iteration of the greedy adaptive policy, each internal level of the tree can contain at most one explored node.}

Let $v,v',u,\ell_u,d_u$ have the same meaning as in the proof of Lemma~\ref{lem:ltm_precribed}.
We have already seen that, at the current iteration, choosing $v'$ is suboptimal, and choosing $v'$ yields a marginal influence which is at most a fraction
$$\frac{1}{d^{\ell_u-1}d_u}/\frac1{d^{\ell_u-2}(d-1)d_u}=1-\frac1d$$
of the marginal influence of $v$.
Therefore, if we set $\varepsilon$ such that $\varepsilon<\frac1d$, the next seed chosen by the policy $\pi^g$ will not be a leaf that is a descendent of an unexplored node with probability at least $1-\delta$.
By a union bound, with probability at least $1-k\delta=1-o(1)$, if the root node is not infected yet, it will never happen that, at an iteration, there are more than one explored node in the same level.

The remaining part of the proof is almost the same.
Since this crucial claim holds with probability at least $1-o(1)$, the same induction argument shows that $\sigma^\f(\pi^g,k)\geq\frac12W$, and
we have $\frac{\sigma^\f(\pi^g,k)}{\sigma(S)}=\Omega(2^d)=2^{\Omega(\log|V|/\log\log|V|)}$ as long as $\varepsilon<\frac1d$.
Noticing that $d=\Omega(\frac{\log k}{\log\log k})$ (in particular, since $k=2(\frac{d+1}2)^d$, $\log k=d\log d + O(d)$ and $\log\log k=\log d +o(\log d)$, we have $d=\Omega(\frac{\log k}{\log\log k})$), implying $\frac1d=O(\frac{\log \log k}{\log k})$.
The lemma holds with a sufficiently small $c$.
\end{proof}

Finally, we extend Theorem~\ref{thm:sup_gap} to the $(\varepsilon,\delta)$-greedy setting.
\begin{theorem}\label{thm:sup_gap_epsilon_delta}
For any constant $c>0$, given any two functions $\varepsilon:\Z^+\to\R^+$ and $\delta:\Z^+\to\R^+$ such that $\varepsilon(k)=O(\frac1{k^{2+c}})$ and $\delta(k)=o(1/k)$, there exists a triggering model $I_{G,F}$ and $k$ such that, for any valid seed set $S$ (i.e., $S$ is a subset of the candidate set $\overline{V}$ and $|S|\leq k$) and any $\pi^g\in\Pi^g_{\varepsilon(k),\delta(k)}$, we have
$$\frac{\sigma^\f(\pi^g,k)}{\sigma(S)}\geq2^{c'\log(|V|)/\log\log(|V|)},$$
where $c'>0$ is a universal constant.
\end{theorem}
\begin{proof}
The sketch of the proof follows from Section~\ref{sect:proofs}.
We construct the same \infmax instance with the same triggering model that is a mixture of \ICM and \LTM.
The analysis for the non-adaptive algorithms is the same.
We have $\sigma(S)\leq M(k\cdot\frac1{d^d}W+d^{d+1})$ for any $S$ with $|S|\leq k$.

Consider any $\pi^g\in\Pi^g_{\varepsilon(k),\delta(k)}$.
Consider an arbitrary iteration.
Let $a_\z\in A$ be the seed that maximizes the marginal influence, which will be the one picked by the exact greedy adaptive policy.
Recall that a vertex in $A$ corresponds to the selection of a seed among the leaves in each of $T_1,\ldots,T_M$.
Naturally, $a_\z$ makes the optimal selection in all the $M$ trees.
From the argument in the proof of Lemma~\ref{lem:ltm_precribed}, in each tree $T_i$ and in each iteration, as long as the seed selected in $T_i$ satisfies that there is at most one explored node at each level of $T_i$, we will have the greedy adaptivity gap being $2^{\Omega(\log|V|/\log\log|V|)}$ on the subgraph $T_i$, and the same argument in Section~\ref{sect:proofs} shows that the greedy adaptivity gap overall is $2^{\Omega(\log\log|V|/\log\log\log|V|)}$.
To conclude the proof of this theorem, we will show that $\varepsilon=O(\frac1{k^{2+c}})$ is sufficient to make sure that this will happen for all $T_1,\ldots,T_M$.

To show this, we consider a suboptimal $a_\z'\in A$ such that, at some tree $T_i$, there are more than one explored node at some level of $T_i$, and we find a lower bound of the difference between the marginal influence of $a_\z$ and the marginal influence of $a_\z'$.
Let $v,v',u,\ell_u,d_u$ have the same meaning as in the proof of Lemma~\ref{lem:ltm_precribed}.
Let $p_u$ be the probability that, given the feedback at the current iteration, the path connecting from $u$ to the root contains only live edges (i.e., $u$ will infect the root).
Then the marginal influence of $v'$ is at most
$$\frac1{d^{\ell_u-1}d_u}\cdot p_u\cdot W+d+1$$
(where the second term $d$ is the number of nodes on the path from $v'$ to the root, which can potentially be infected),
and the marginal influence of $v$ is at least
$$\frac1{d^{\ell_u-2}(d-1)d_u}\cdot p_u\cdot W.$$
The difference is at least
$$Wp_u\cdot\frac1{d^{\ell_u-2}d_u}\left(\frac1{d-1}-\frac1d\right)-d-1\geq \frac{W}{d^d(d-1)}-d-1,$$
where we used the fact that $d_u\leq d$ and $p_u\geq\frac1{d^{d-\ell_u}}$.

On the other hand, the marginal influence of $a_\z$ is at most $M(W+d+1)$ (we have assumed the root is infected at this iteration for each of $T_1,\ldots,T_M$).
It suffices to find an $\varepsilon$ such that
$$\frac{W}{d^d(d-1)}-d-1>\varepsilon M(W+d+1).$$
Since we have set $W=M=d^{d+10}$, this is equivalent to
$$\frac{d^{10}}{d-1}-d-1>\varepsilon d^{d+10}\left(d^{d+10}+d+1\right),$$
which implies
$$\varepsilon=O\left(\frac1{d^{2d+11}}\right).$$
Finally, recalling that $k=2(\frac{d+1}2)^d$, elementary calculations shows that $\varepsilon=O(\frac1{k^{2+c}})$ is a sufficient condition to the above:
$$\frac1{k^{2+c}}=\frac1{2^{2+c}}\left(\frac2{d+1}\right)^{2d+cd}<\left(\frac1d\right)^{2d+cd}\cdot2^{2d+cd}=\frac1{d^{2d+11}}\cdot\frac{2^{2d+cd}}{2^{(cd-11)\log d}},$$
and the second term in the product above tends to $0$ as $d\rightarrow\infty$. 
\end{proof}

\section{A Variant of Greedy Adaptive Policy}
\label{sect:variant}
Although we have seen that the adaptive version of the greedy algorithm can perform worse than its non-adaptive counterpart, in general, we would still recommend the use of it as long as it is feasible, as it can also perform significantly better than the non-adaptive greedy algorithm (Theorem~\ref{thm:sup_gap}) while never being too bad (Theorem~\ref{thm:lowerbound}).
As we remarked, the adaptivity may be harmful because exploiting the feedback may make the seed-picker too myopic.
In this section, we propose a less aggressive risk-free version of the greedy adaptive policy, $\pi^{g-}$, in that it balances between the exploitation of the feedback and the focus on the average in the conventional non-adaptive greedy algorithm.

First, we apply the non-adaptive greedy algorithm with $|V|$ seeds to obtain an order $\mathcal{L}$ on all vertices.
Then for any $S\subseteq V$ and any partial realization $\varphi$, $\pi^{g-}(S,\varphi)$ is defined to be the first vertex $v$ in $\mathcal{L}$ that is not known to be infected.
Formally, $v$ is the first vertex in $\mathcal{L}$ that are not reachable from $S$ when removing all edges $e$ with $\varphi(e)\in\{\block,\unknown\}$.
This finishes the description of the policy.

This adaptive policy is always no worse than the non-adaptive greedy algorithm, as it is easy to see that
those seeds chosen by $\pi^g$ are either seeded or infected by previously selected seeds in $\pi^{g-}$.

However, $\pi^{g-}$ can sometimes be conservative.
It is possible that $\pi^{g-}$ has the same performance as the non-adaptive greedy algorithm, but $\pi^g$ is much better.
Especially, when there is no path between any two vertices among the first $k$ vertices in $\mathcal{L}$, $\pi^{g-}$ will make the same choice as the non-adaptive greedy algorithm.
The \infmax instance in Section~\ref{sect:proofs} is an example of this.

We have seen that $\pi^{g-}$ sometimes performs better than $\pi^g$ (e.g., in those instances constructed in the proofs of Lemma~\ref{lem:tightICM} and Lemma~\ref{lem:tightLTM}) and sometimes performs worse than the $\pi^g$ (e.g., in the instance constructed in Section~\ref{sect:proofs}).
Therefore, given a \emph{particular} \infmax instance, for deciding which of $\pi^{g-}$ and $\pi^g$ to be used (we should never consider the non-adaptive greedy algorithm if adaptivity is available, as it is always weakly worse than $\pi^{g-}$), we recommend a comparison of the two policies by simulations.
Notice that the seed-picker can randomly sample a realization $\phi$ and simulate the feedback the policy will receive.
Thus, given $I_{G,F}$, both $\pi^{g-}$ and $\pi^g$ can be estimated by taking an average over the numbers of infected vertices in a large number of simulations.
In the next section, we evaluate the three algorithms---the non-adaptive greedy algorithm, the greedy adaptive policy $\pi^g$ and the conservative greedy adaptive policy $\pi^{g-}$---empirically by experiments on social networks in our daily lives.

\section{Empirical Experiments}
\label{sect:experiments}
In this section, we compare the three algorithms---the non-adaptive greedy algorithm, the greedy adaptive policy $\pi^g$ and the conservative greedy adaptive policy $\pi^{g-}$---empirically by experiments on the social networks in our daily lives.
Below is a quick summary of the results obtained from our experiments.
\begin{enumerate}
    \item The greedy adaptive policy $\pi^g$ outperforms the conservative greedy adaptive policy $\pi^{g-}$ and the non-adaptive greedy algorithm in most scenarios.
    \item The conservative greedy adaptive policy $\pi^{g-}$ always outperforms the non-adaptive greedy algorithm.
    \item Occasionally, the greedy adaptive policy $\pi^g$ is outperformed by the conservative adaptive policy $\pi^{g-}$, or even the non-adaptive greedy algorithm.
\end{enumerate}
Notice that our results in Section~\ref{sect:inf} support the third observation, and the second observation follows easily from our definition of $\pi^{g-}$ in the last section.

In the proofs of Lemma~\ref{lem:tightICM} and Lemma~\ref{lem:tightLTM} in Sect.~\ref{sect:inf_gap_tight}, we have constructed two graphs where the greedy adaptive policy performs worse than the non-adaptive greedy algorithm.
We test the performances of the three algorithms empirically on the two graphs in Sect.~\ref{sect:experiments_constrcutedgraphs}.

\subsection{Reverse Reachable Sets}
\label{sect:RR_set}
In this section, we discuss a popular type of greedy-based algorithm---the reverse-reachable-set-based algorithms.
Our experiments have also made use of reverse reachable sets.

In all those reverse-reachable-set-based algorithms, including RIS~\shortcite{borgs2014maximizing}, $\text{TIM}^+$~\shortcite{tang2014influence}, IMM~\shortcite{tang2015influence}, EPIC~\shortcite{han2018efficient}, a sufficient number of reverse reachable sets are sampled.
Each reverse reachable set is sampled as follows: first, a vertex $v$ is sampled uniformly at random; second, sample the live edges in the graph where each vertex chooses a triggering set according to the triggering model (undirected graphs are treated as directed graphs with anti-parallel edges); lastly, the reverse reachable set consists of exactly those vertices from which $v$ is reachable.

After collecting sufficiently many reverse reachable sets, the algorithms choose $k$ seeds that attempt to cover as many reverse reachable sets as possible (we say a reverse reachable set is covered if it contains at least $1$ seed), and this is done by a greedy maximum coverage way: iteratively select the seed that maximizes the extra number of reverse reachable sets covered by this seed.

The meat of those reverse-reachable-set-based algorithms is that, given a seed set $S$, the probability that a randomly sampled reverse reachable set is covered by $S$ is exactly the probability that a vertex selected uniformly at random from the graph is infected by $S$.
Therefore, when sufficiently many reverse reachable sets are sampled, the fraction of the reverse reachable sets covered by $S$ is a good approximation to $\sigma(S)/|V|$.

\subsection{Experiments Setup}
We implement the experiments on four undirected graphs, shown in Table~\ref{tab:data}.
All of our datasets come from~\shortcite{snapnets}, and these networks are also popular choices in other empirical work.
We implement the three algorithms with $k=200$ seeds.

\begin{table}
    \centering
    \begin{tabular}{lrrr}
    \hline
    Dataset     &  Number of Vertices & Number of Edges & Average Degree\\
    \hline
      Nethept & 15,233 & 31,387 & 4.12 \\
      CA-HepPh & 12,008 & 118,505 & 19.73 \\
      DBLP & 317,080 & 1,049,866 & 6.62\\
      com-YouTube & 1,134,890 & 2,987,624 & 5.26\\
    \hline 
    \end{tabular}
    \caption{Datasets for experiments}
    \label{tab:data}
\end{table}

For the diffusion model, we implement both \ICM and \LTM.
For \ICM, the weight of each edge is set to $0.01$.
For \LTM, each undirected edge $(u,v)$ is viewed as two anti-parallel directed edges such that $w(u,v)=1/\deg(v)$ and $w(v,u)=1/\deg(u)$.
For each dataset, we sample three realizations $\phi_1,\phi_2,\phi_3$ as the ``ground-truth''.
Therefore, a total of six experiments are performed for each dataset: the two models \ICM and \LTM for each of the three realizations.
For each of those six experiments, when a seed $s$ is chosen, all vertices that are reachable from $s$ in the ground-truth realization are considered infected, and given as the feedback.
In particular, we consider the full-adoption feedback in our experiments.

To implement the three algorithms, we sample 1,000,000 reverse reachable sets, and perform the greedy maximum coverage algorithm described in the last sub-section which iteratively selects the seed that maximizes the number of extra reverse reachable sets covered by this seed.
We iteratively select seeds in this way until a sufficient number of seeds are selected (we decided to select 10,000 seeds, which turns out to be sufficient), and we ordered them in a list.
Naturally, the non-adaptive greedy algorithm choose the first $k=200$ seeds in this list.
The conservative greedy adaptive policy iteratively select the first not-yet-selected seed in the list that is not known to be infected, as described in Section~\ref{sect:variant}.

As for the greedy adaptive policy, the first seed is the same as the one for non-adaptive greedy algorithm and the conservative greedy adaptive policy.
In each future iteration, the vertices that are infected (given as the feedback) are removed from the graph, and 1,000,000 new reverse reachable sets are sampled on the remainder graph.
Notice that, for \LTM, the degrees of the vertices in the remainder graph may decrease, which increases the weights of the incoming edges of these vertices.
Then, a seed that covers a maximum number of reverse reachable sets is selected as the next seed.

We remark that removing infected vertices from the graph and sampling reverse reachable sets on the remainder graph is the correct way to implement the algorithm.
Since we are considering the full-adoption feedback, we know that there is no directed live edge from an infected vertex to an uninfected vertex, for otherwise the uninfected vertex should have been infected.
When sampling the reverse reachable set, the triggering set of any uninfected vertex should not intersect with any infected vertex.
Given an arbitrary uninfected vertex $v$ and letting $X$ be the set of all infected vertices, $v$ should include each vertex in $\Gamma(v)\setminus X$ to its triggering set with probability $0.01$ independently under \ICM, and $v$ should include exactly one vertex chosen uniformly at random in $\Gamma(v)\setminus X$ to its triggering set under \LTM.
Consequently, for both \ICM and \LTM, we can and we should remove those infected vertices from the graph and sample reverse reachable sets in the remainder graph.

\subsection{Results}
As we mentioned, for each dataset, we have six figures corresponding to \ICM and \LTM for each of the three realizations $\phi_1,\phi_2,\phi_3$.
In each figure, the $x$-axis is the number of seeds, and the $y$-axis is the number of infected vertices in the realization.
The three curves correspond to the outcomes of the three algorithms.
Figure~\ref{fig:Nethept}, \ref{fig:CA-HepPh}, \ref{fig:DBLP} and \ref{fig:YouTube} correspond to the datasets Nethept, CA-HepPh, DBLP, and com-YouTube respectively.
The three observations mentioned at the beginning of this section can be easily observed from the figures.

\subsection{Experiments on Graphs Constructed in Sect.~\ref{sect:inf_gap_tight}}
\label{sect:experiments_constrcutedgraphs}
For \ICM, based on the construction in the proof of Lemma~\ref{lem:tightICM}, we build five graphs with the parameter $W$ set to $10k^{2k}(k-1)$ and the parameter $k$ set to $10,15,20,25$ and $30$ respectively.
We implement the experiment with weighted vertices (see Remark~\ref{remark:weightedvertices}).
That is, we set the weight of the vertex $u$ to be $W$ instead of creating $W-1$ vertices that are connected from $u$.
The expected number of infections for a given seed set is evaluated with 10,000 Monte Carlo simulations.
Notice that all the edges in our construction have weight $1$ except for the edge $(s,t)$, which has weight $1/k$.
Thus, there can be only two possible realizations $\phi_{(s,t)}$ and $\phi_{\neg(s,t)}$ where the edge $(s,t)$ is live and blocked respectively.
In addition, $\phi_{(s,t)}$ and $\phi_{\neg(s,t)}$ are realized with probability $1/k$ and $1-1/k$ respectively.
For the greedy adaptive policy, we test the performances of the policy on both graphs, and the overall performance is then given by the weighted average.
In all the tests, the non-adaptive greedy algorithm selects the seed set $\{s,u_1,\ldots,u_k\}$ and the greedy adaptive policy selects the seed set $\{s,v_1,\ldots,v_t\}$, exactly as we predicted in the proof of Lemma~\ref{lem:tightICM}.
Since each vertex $\{s,u_1,\ldots,u_k\}$ has in-degree $0$, the seeds selected by the non-adaptive greedy algorithm cannot infect each other.
Therefore, the conservative greedy adaptive policy selects exactly the same seed set as the non-adaptive greedy algorithm, and we omit the test for the conservative greedy adaptive policy.
The experimental results are shown in Table~\ref{tab:ICMexperiment}.
The first two rows show the expected numbers of infected vertices for the two algorithms, and the last row shows the ratios between the performance of the greedy adaptive policy and the performance of the non-adaptive greedy algorithm.
We can see that the greedy adaptive policy performs worse than its non-adaptive counterpart in all cases, and the gap becomes larger as the size of the graph grows (with $k$ increasing).

\begin{table}
    \centering
    \begin{tabular}{lcccccc}
    \hline
       $k$  & $10$ & $15$ & $20$ & $25$ & $30$ \\
    \hline
        non-adaptive  & $1.426\times 10^{23}$ & $5.607\times 10^{38}$ & $5.413\times 10^{55}$ & $5.855\times 10^{73}$ & $4.418\times 10^{92}$ \\
        adaptive & $1.176\times 10^{23}$ & $4.338\times10^{38}$ & $4.030\times10^{55}$ & $4.247\times10^{73}$ & $3.144\times 10^{92}$\\
        gap & $0.825$ & $0.774$ & $0.744$ & $0.725$ & $0.712$\\
    \hline
    \end{tabular}
    \caption{The results for \ICM with the graph in Lemma~\ref{lem:tightICM}.}
    \label{tab:ICMexperiment}
\end{table}

For \LTM, based on the construction in the proof of Lemma~\ref{lem:tightLTM}, we build three graphs with the parameter $W$ set to $10k^{2k}(k-1)$ and the parameter $k$ set to $10,15,20$ respectively.
For each $k$, we randomly sample $k$ realizations $\phi_1,\ldots,\phi_k$.
Since $t$ has two incoming edges $(s,t)$ and $(v_1,t)$ with weights $1/k$ and $1-1/k$ respectively, we let $(s,t)$ appears in $\phi_1$ and $(v_1,t)$ appears in $\phi_2,\ldots,\phi_k$.
(Recall that the greedy adaptive policy will select seeds that are of worse quality than the non-adaptive greedy algorithm in the case it receives the feedback that $(s,t)$ is blocked.)
The remaining edges are sampled randomly based on \LTM in each $\phi_i$.

We observe that the non-adaptive greedy algorithm always selects $\{s,u_1,\ldots,u_k\}$ as we expected.
However, for the greedy adaptive policy, it sometimes selects the ``bad'' seed set $\{s,v_1,\ldots,v_k\}$ as predicted in the proof of Lemma~\ref{lem:tightLTM}, and it sometimes selects the ``good'' seed set $\{s,u_1,\ldots,u_k\}$.
The chance it selects the bad seed set (i.e., more accurately reflects our prediction) becomes higher when the number of Monte-Carlo simulations increases.
This is natural: when the number of Monte Carlo simulations increases, the greedy adaptive policy is more likely to choose the locally optimal seed at each iteration; if it does, it will lead to a ``bad'' seed set.
On the other hand, with less Monte Carlo simulations, the policy sometimes chooses a locally sub-optimal seed, which is actually beneficial in that it deviates from the ``path'' leading to a ``bad'' seed set.
To be more specific, we have seen that, after selecting $s$ as the first seed, if we receive the feedback that $t$ is not infected by $s$, the marginal influence of $v_1$ is slightly larger than any of $u_1,\ldots,u_k$.
The choice of the second seed is pivotal: we reach a ``bad'' seed set if we select $v_1$ as the second seed, and we reach a ``good'' seed set if $u_1$ is selected.
The policy will select $v_1$ as the second seed only when its marginal influence is greater than all of $u_1,\ldots,u_k$.
However, this theoretical fact is not always empirically observed due to the errors in the Monte Carlo estimation.

The experimental results for $k=10,15,20$ with various numbers of Monte Carlo simulations are shown in Table~\ref{tab:LTMexperiment}.
We can see that the gap becomes larger when the number of Monte Carlo simulations increases (with the only exception where $k=15$ with $10^6$ Monte Carlo simulations).
Again, since the non-adaptive greedy algorithm always selects the seed set $\{s,u_1,\ldots,u_k\}$ and the vertices in this set have in-degree $0$, the conservative greedy adaptive policy always selects the same seed set, and we omit the corresponding results.

In conclusion, we have observed that the greedy adaptive policy, to different extents, consistently performs worse than the non-adaptive greedy algorithm (as well as the conservative greedy adaptive policy) in the graphs we constructed in Lemma~\ref{lem:tightICM} and Lemma~\ref{lem:tightLTM}.
We have empirically verified our theoretical postulation that the greedy adaptive policy can perform worse than the non-adaptive greedy algorithm in some special scenarios.

\begin{table}
    \centering
    \begin{tabular}{lcccccc}
    \hline
       \# of MC  & $10^2$ & $10^3$ & $10^4$ & $10^5$ & $10^6$ \\
    \hline
        non-adaptive  & $1.422\times 10^{23}$ & $1.422\times 10^{23}$ & $1.422\times 10^{23}$ & $1.422\times 10^{23}$ & $1.422\times 10^{23}$ \\
        adaptive & $1.335\times 10^{23}$ & $1.275\times 10^{23}$ & $1.146\times 10^{23}$ & $1.121\times 10^{23}$ & $1.121\times 10^{23}$ \\
        gap & $0.939$ & $0.896$ & $0.806$ & $0.788$ & $0.788$\\
    \hline
    \end{tabular}
    The results for $k=10$
    
    \vspace{0.5cm}
    
    \begin{tabular}{lcccccc}
    \hline
       \# of MC  & $10^2$ & $10^3$ & $10^4$ & $10^5$ & $10^6$ \\
    \hline
        non-adaptive  & $5.602\times 10^{38}$ & $5.602\times 10^{38}$ & $5.602\times 10^{38}$ & $5.602\times 10^{38}$ & $5.602\times 10^{38}$ \\
        adaptive & $5.276\times 10^{38}$ & $5.212\times 10^{38}$ & $4.910\times 10^{38}$ & $4.203\times 10^{38}$ & $4.241\times 10^{38}$ \\
        gap & $0.942$ & $0.930$ & $0.877$ & $0.750$ & $0.757$\\
    \hline
    \end{tabular}
    The results for $k=15$
    
    \vspace{0.5cm}
    
    \begin{tabular}{lcccccc}
    \hline
       \# of MC  & $10^2$ & $10^3$ & $10^4$ & $10^5$ & $10^6$ \\
    \hline
        non-adaptive  & $5.411\times 10^{55}$ & $5.411\times 10^{55}$ & $5.411\times 10^{55}$ & $5.411\times 10^{55}$ & $5.411\times 10^{55}$ \\
        adaptive & $5.183\times 10^{55}$ & $4.864\times 10^{55}$ & $4.612\times 10^{55}$ & $4.227\times 10^{55}$ & $3.924\times 10^{55}$ \\
        gap & $0.958$ & $0.899$ & $0.852$ & $0.781$ & $0.725$\\
    \hline
    \end{tabular}
    The results for $k=20$
    \caption{The results for \LTM with the graph in Lemma~\ref{lem:tightLTM}}
    \label{tab:LTMexperiment}
\end{table}

\begin{figure}
    \centering
    \includegraphics[width=0.45\textwidth]{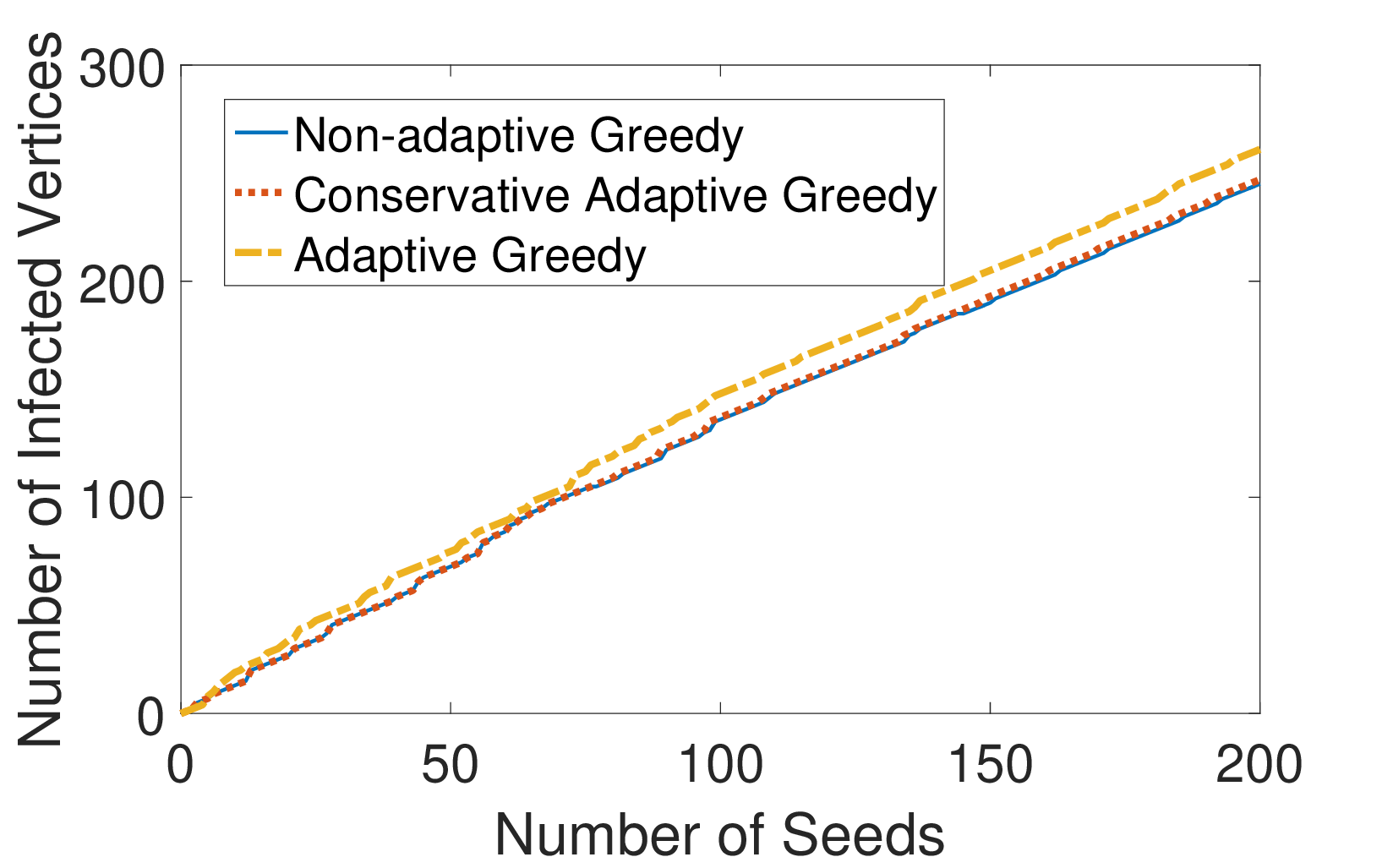}
    \includegraphics[width=0.45\textwidth]{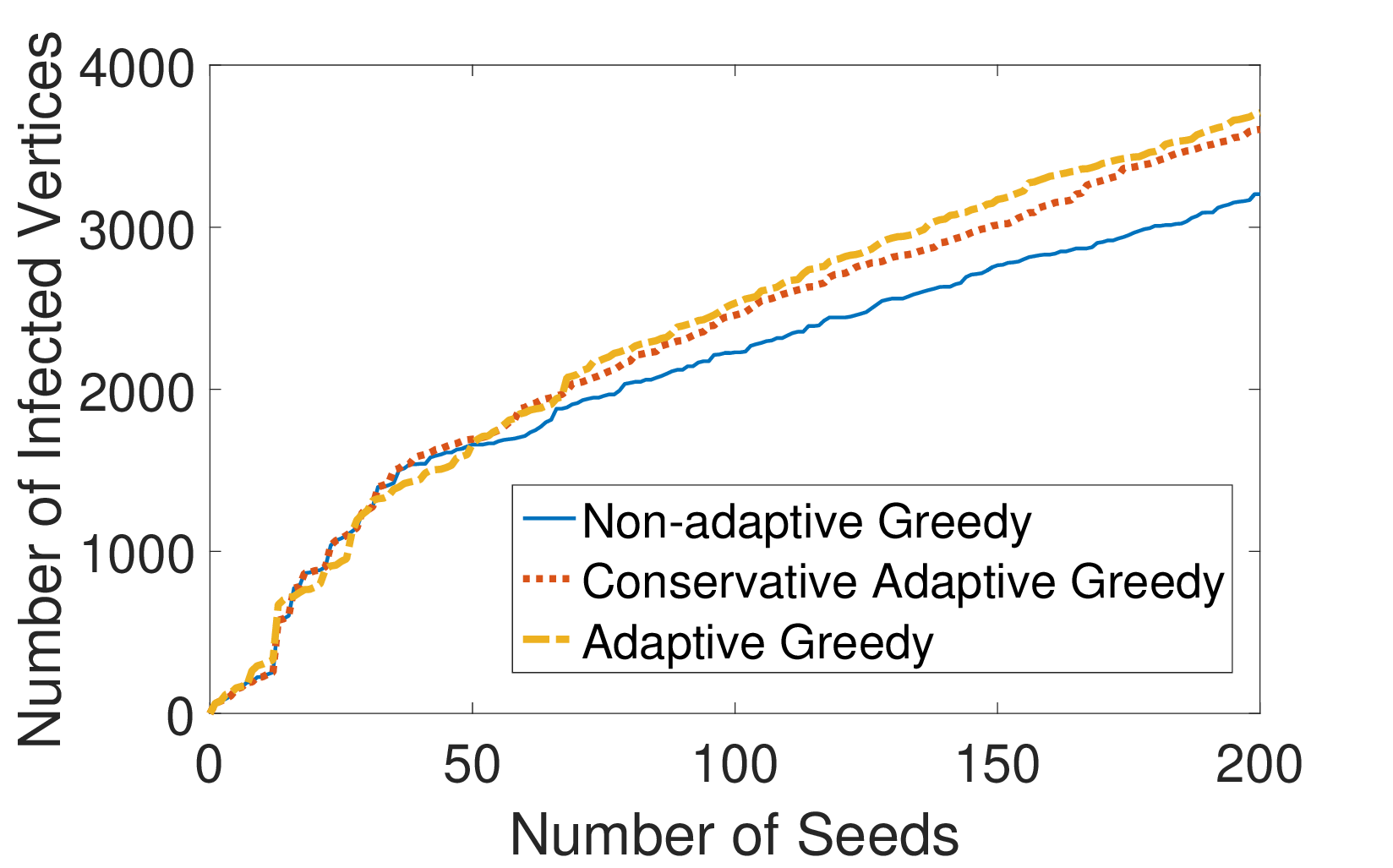}
    \includegraphics[width=0.45\textwidth]{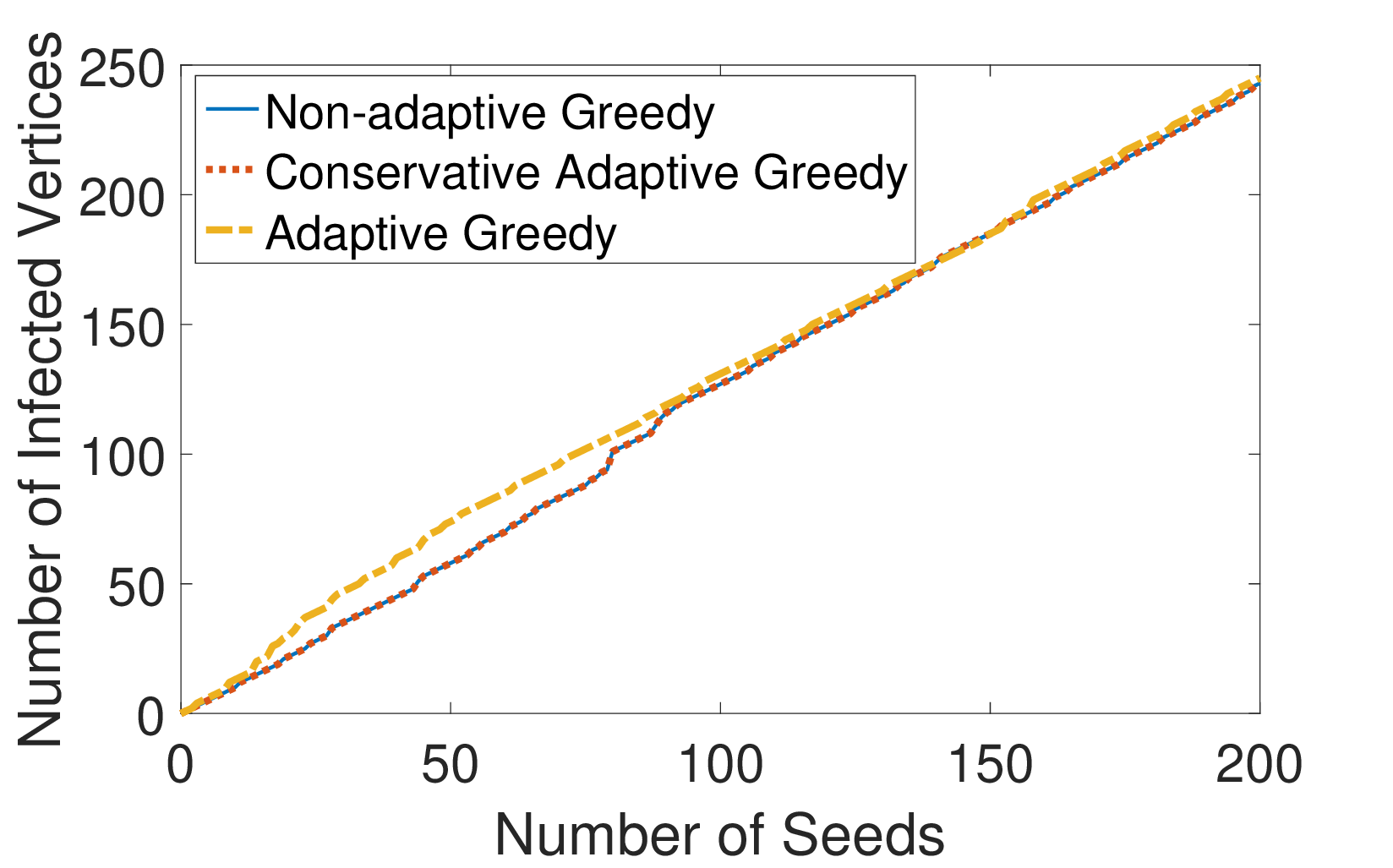}
    \includegraphics[width=0.45\textwidth]{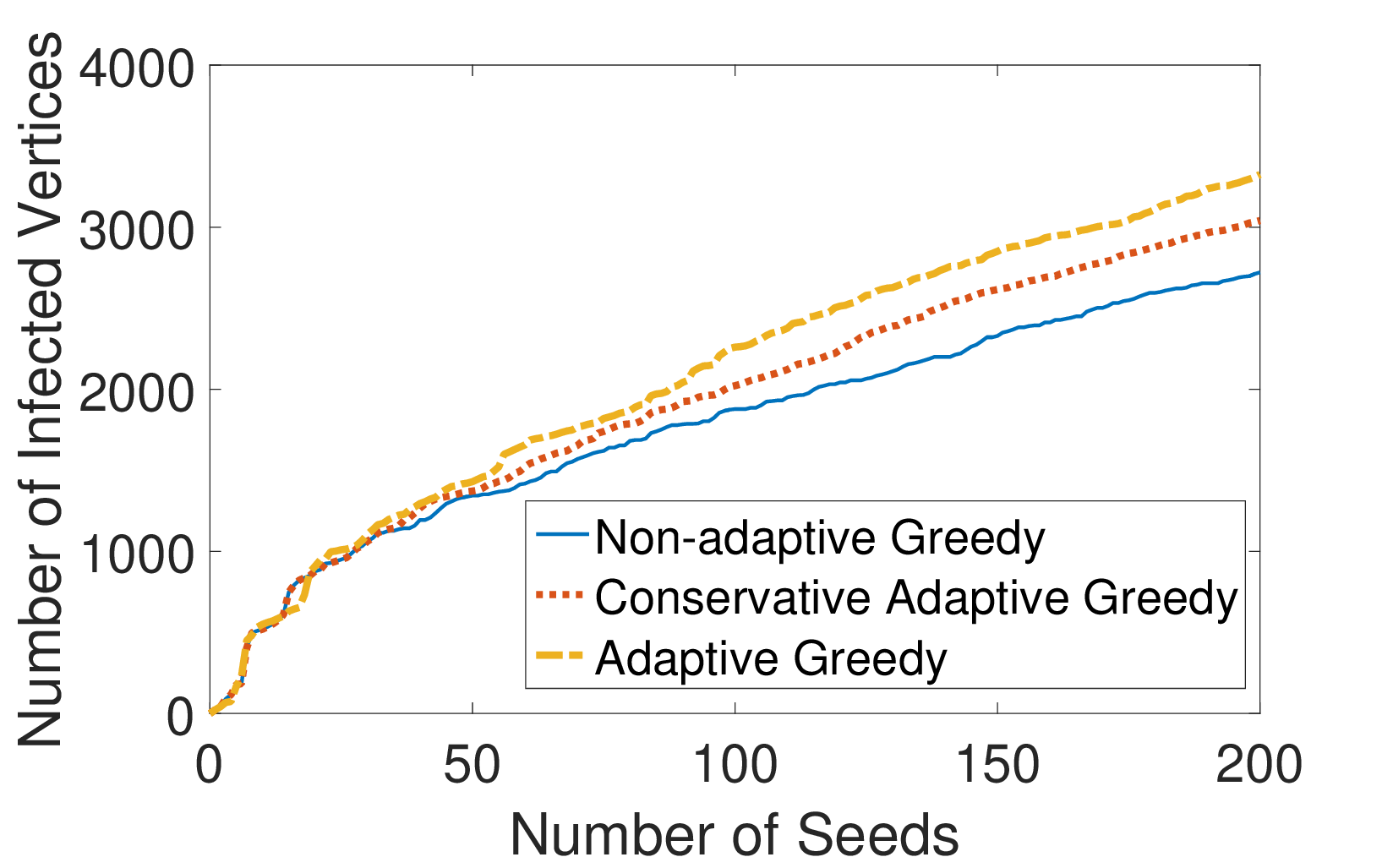}
    \includegraphics[width=0.45\textwidth]{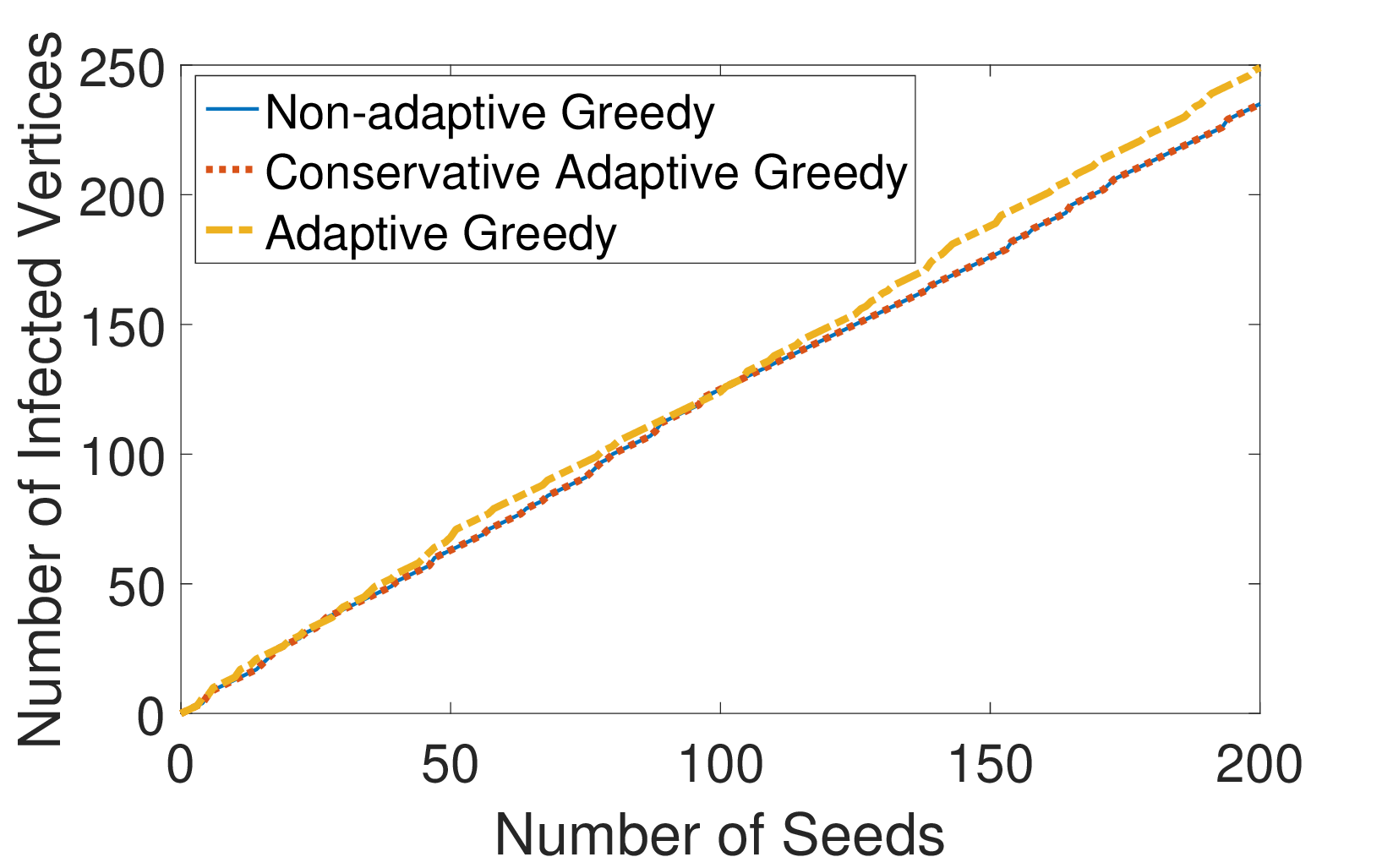}
    \includegraphics[width=0.45\textwidth]{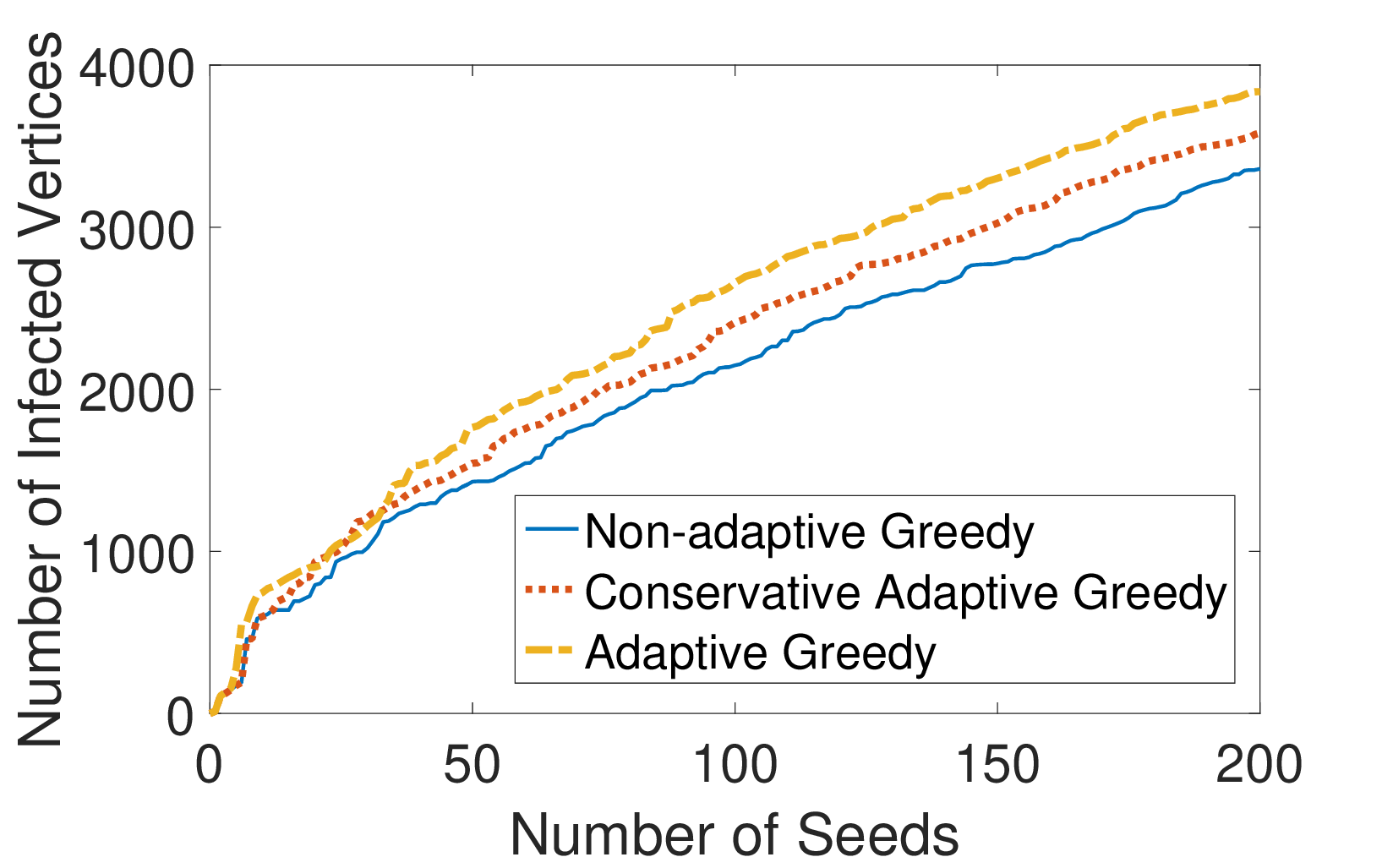}
    \caption{The results for the dataset Nethept. The three rows correspond to the three realizations $\phi_1,\phi_2,\phi_3$, the left column is for \ICM, and the right column is for \LTM.}%
    \label{fig:Nethept}
\end{figure}

\begin{figure}
    \centering
    \includegraphics[width=0.45\textwidth]{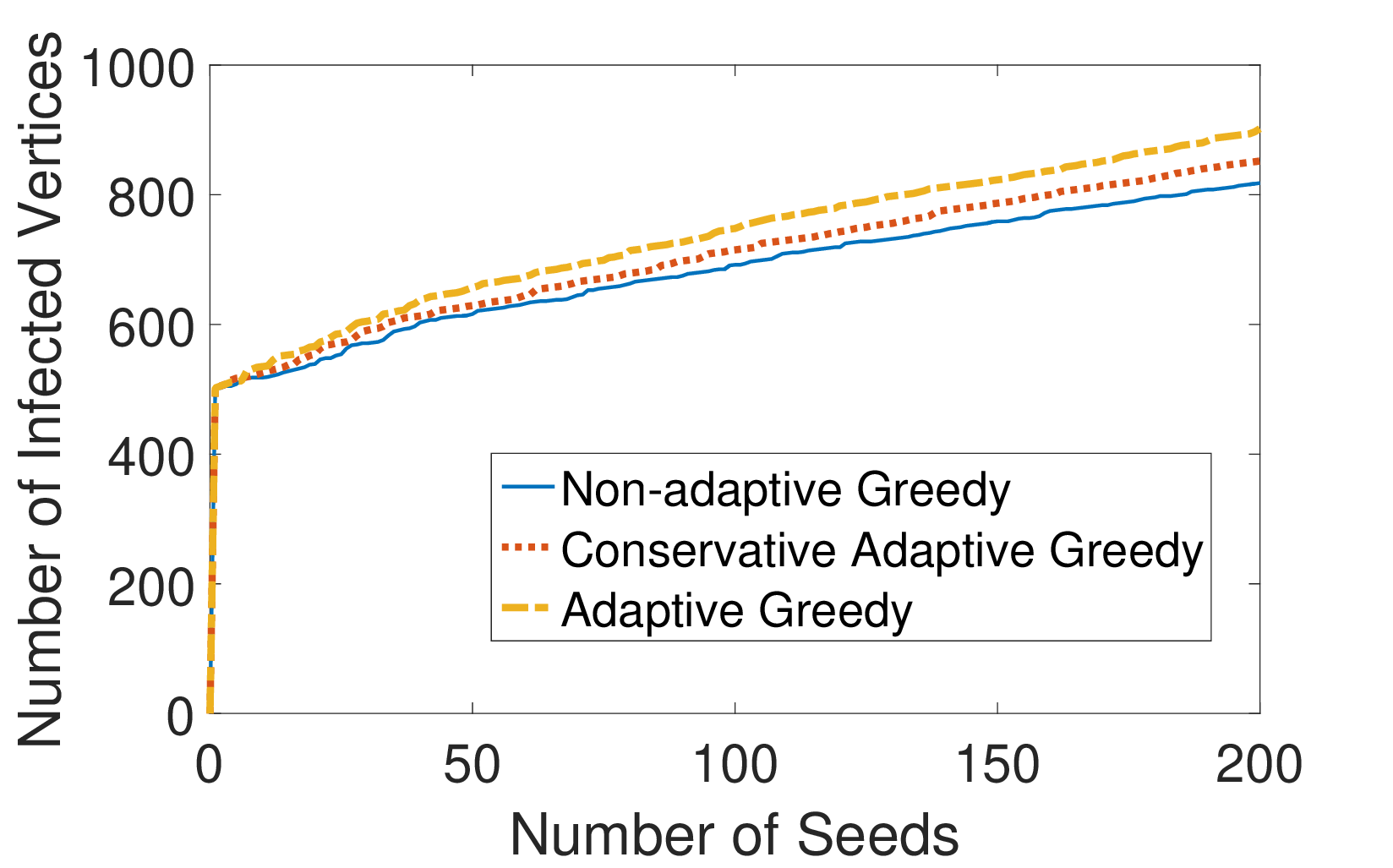}
    \includegraphics[width=0.45\textwidth]{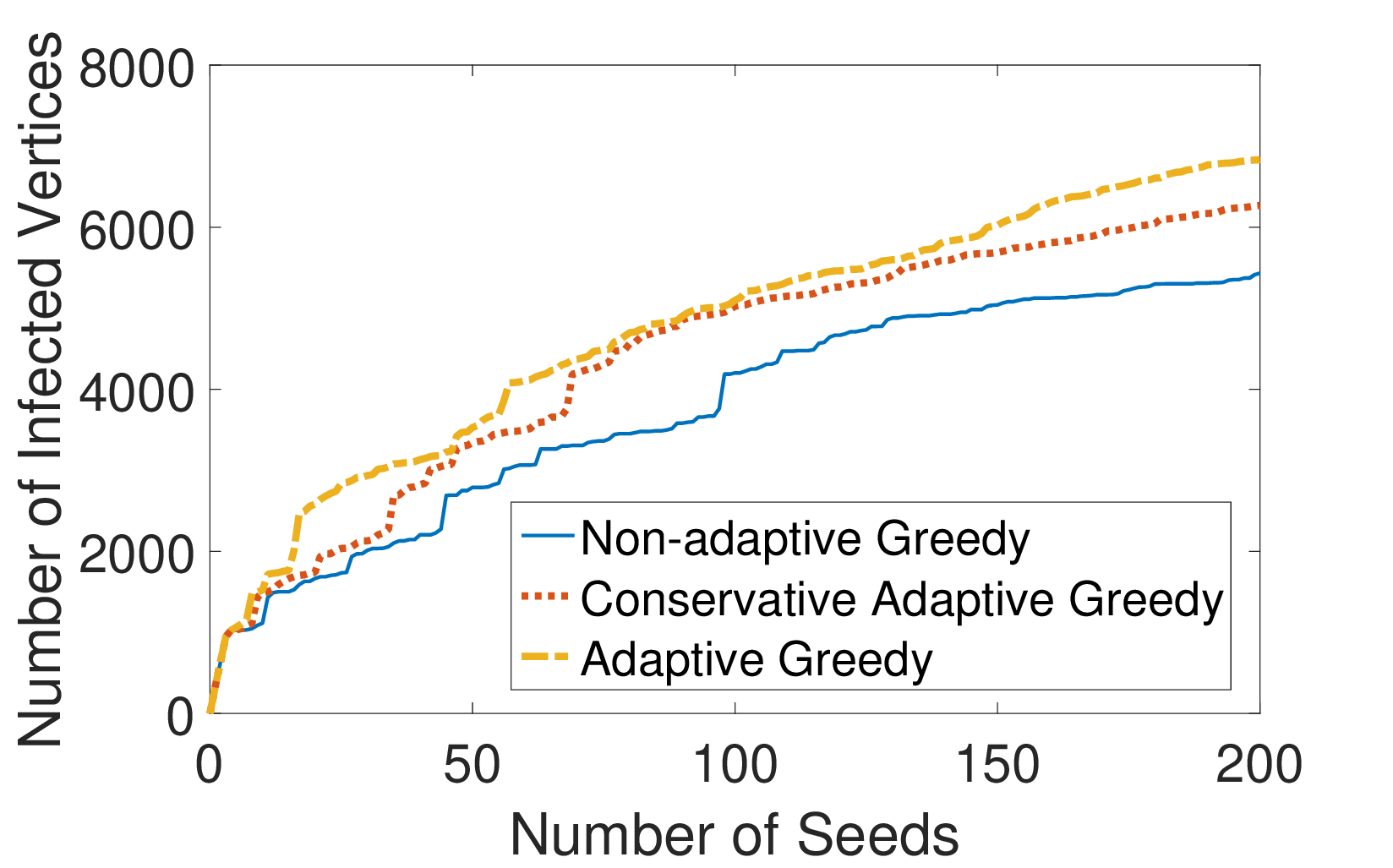}
    \includegraphics[width=0.45\textwidth]{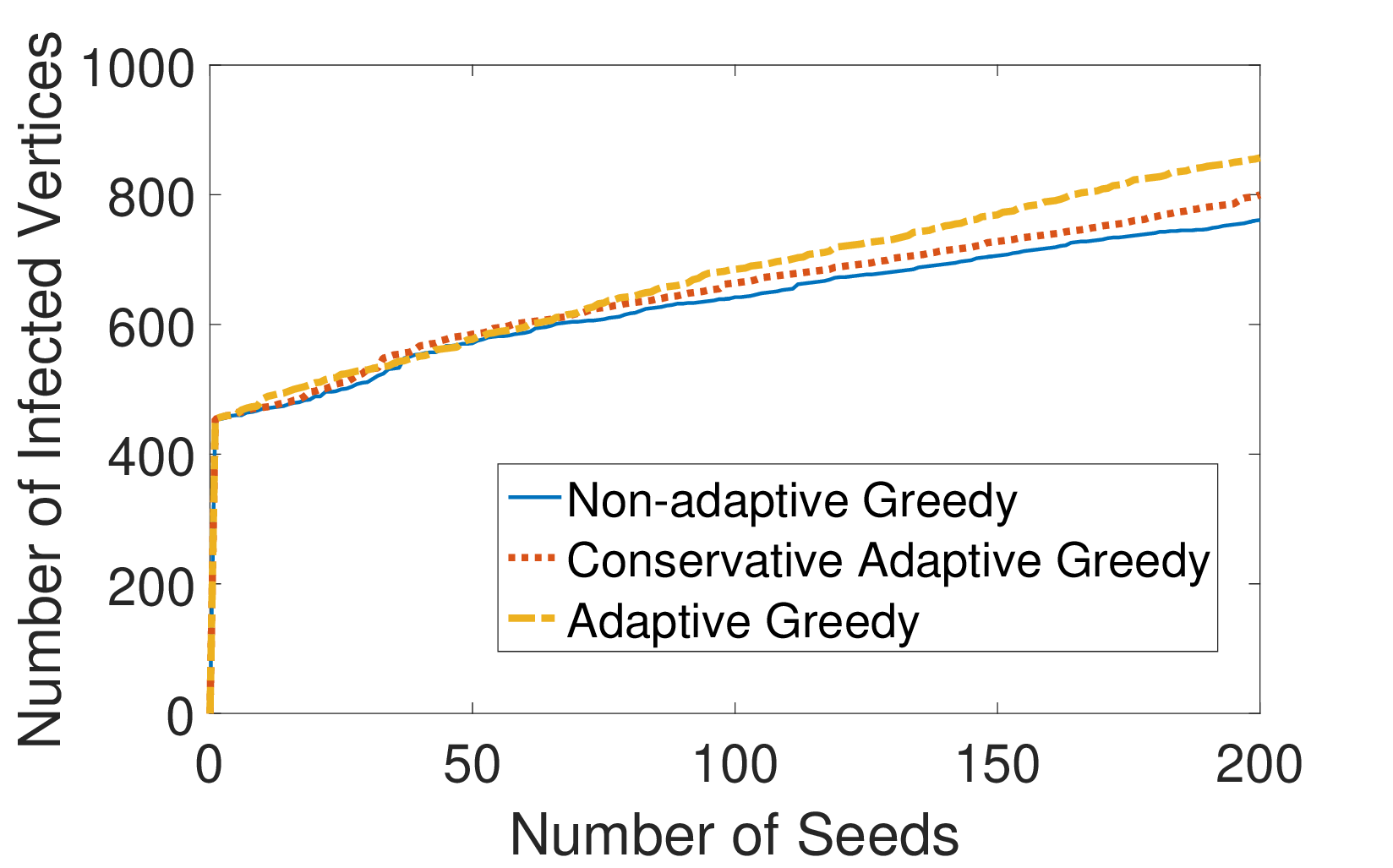}
    \includegraphics[width=0.45\textwidth]{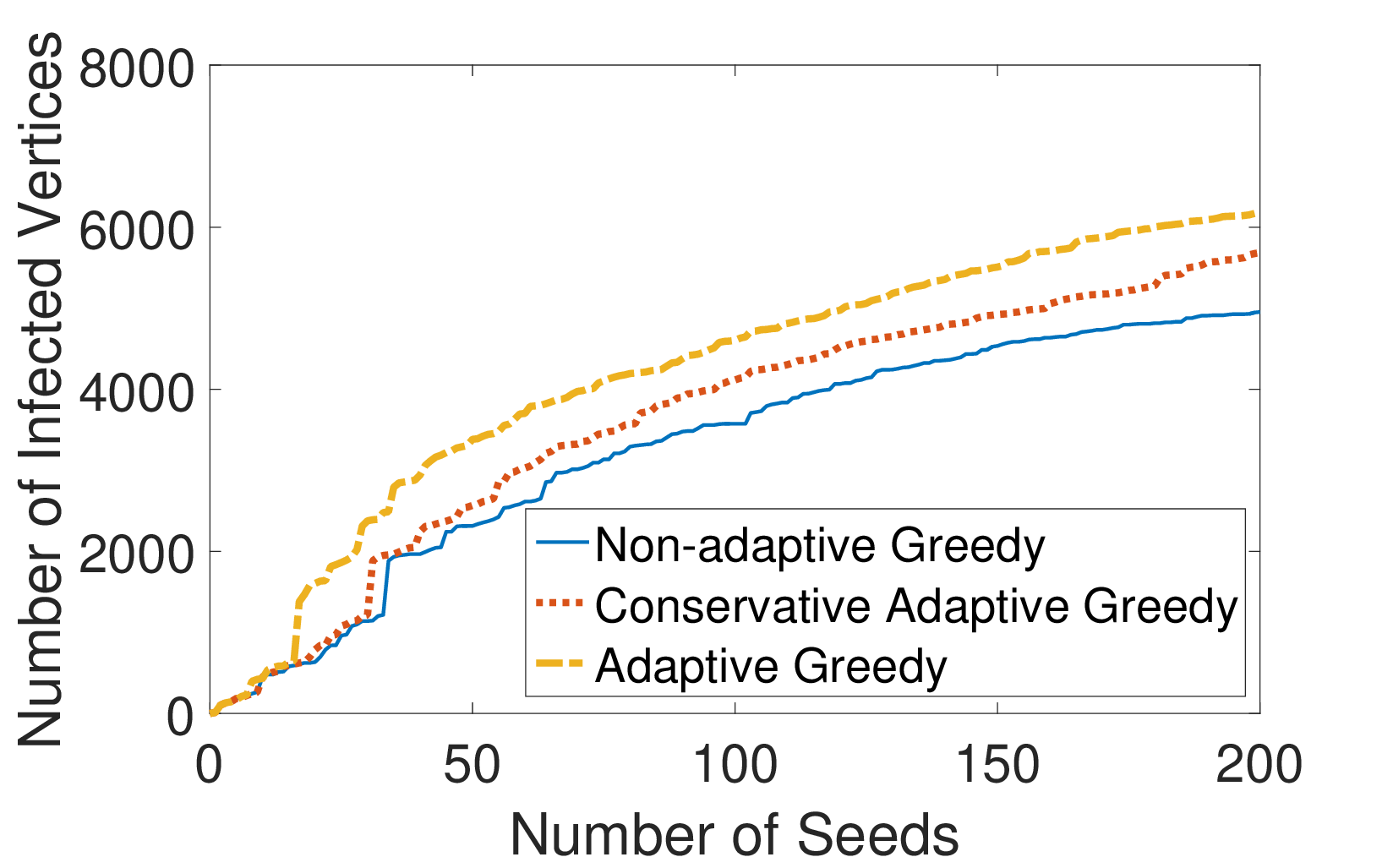}
    \includegraphics[width=0.45\textwidth]{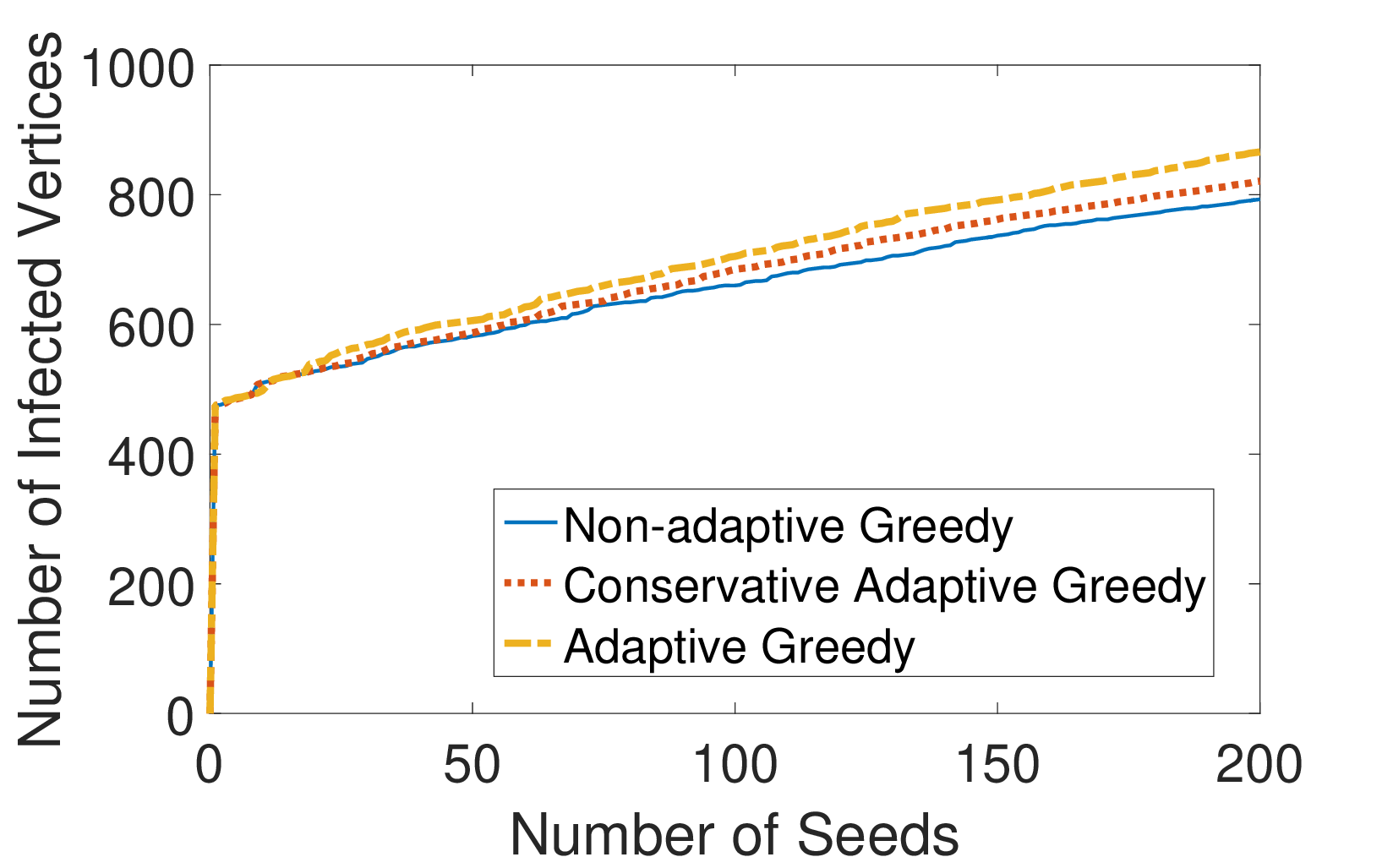}
    \includegraphics[width=0.45\textwidth]{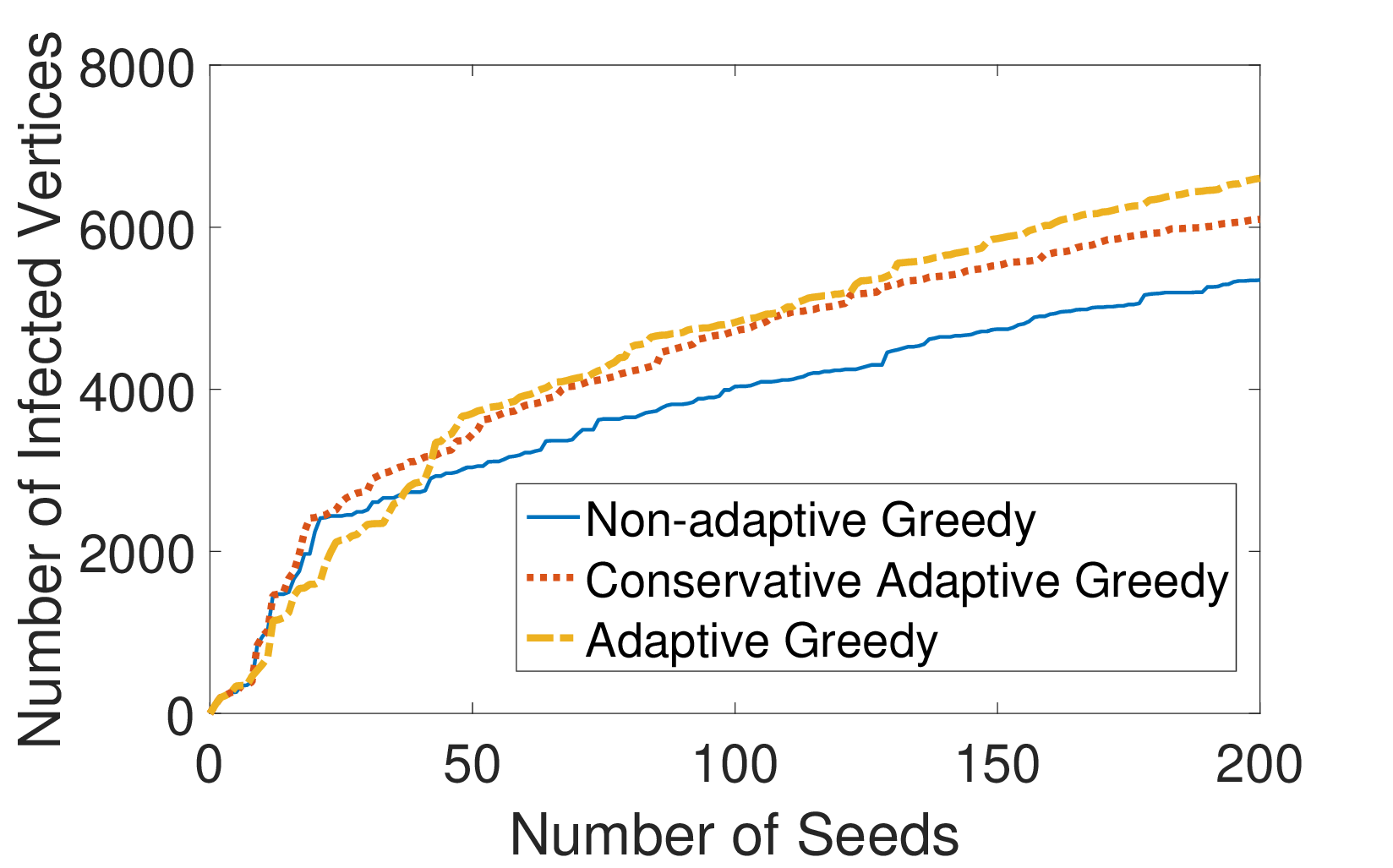}
    \caption{The results for the dataset CA-HepPh. The three rows correspond to the three realizations $\phi_1,\phi_2,\phi_3$, the left column is for \ICM, and the right column is for \LTM.}%
    \label{fig:CA-HepPh}
\end{figure}

\begin{figure}
    \centering
    \includegraphics[width=0.45\textwidth]{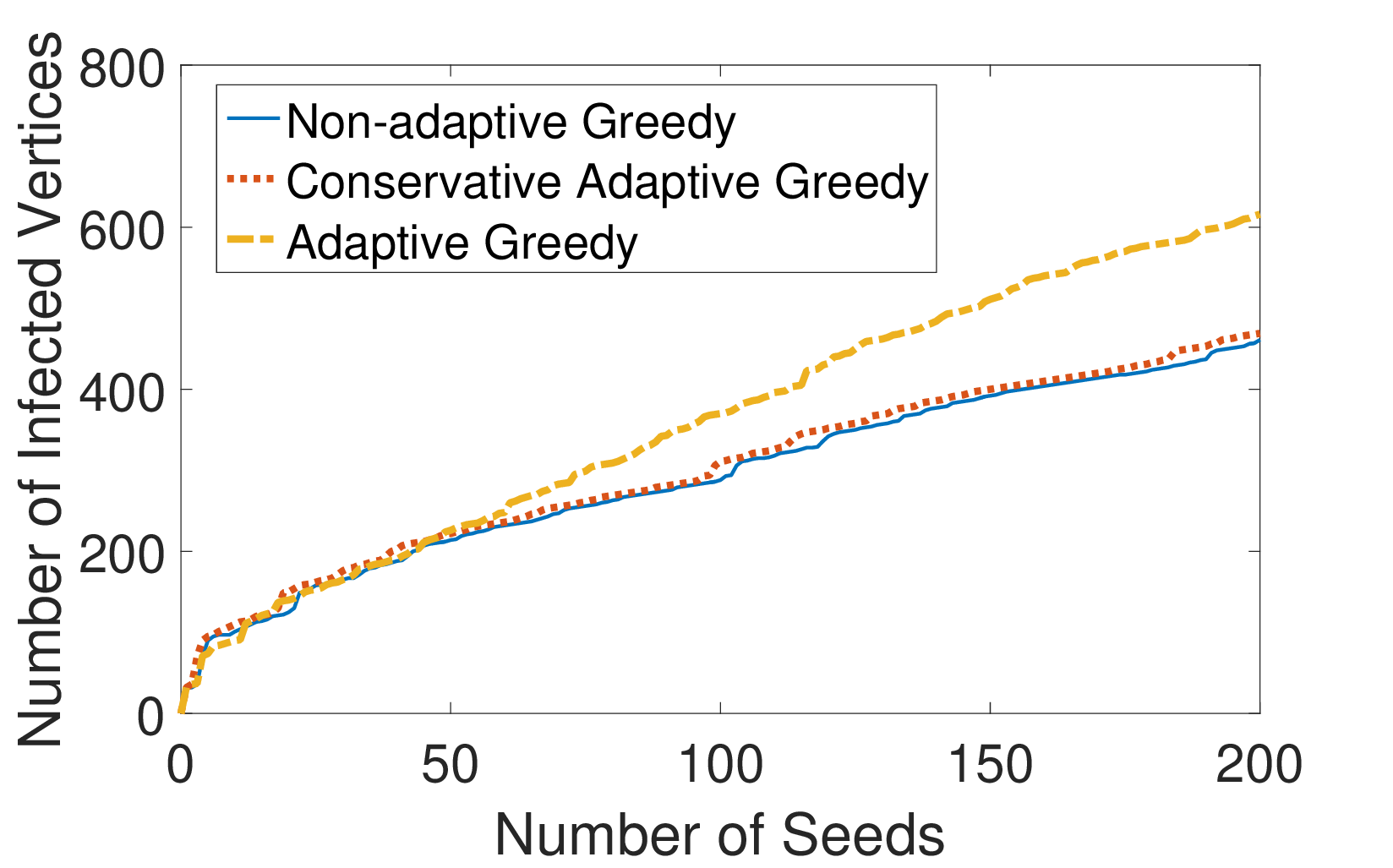}
    \includegraphics[width=0.45\textwidth]{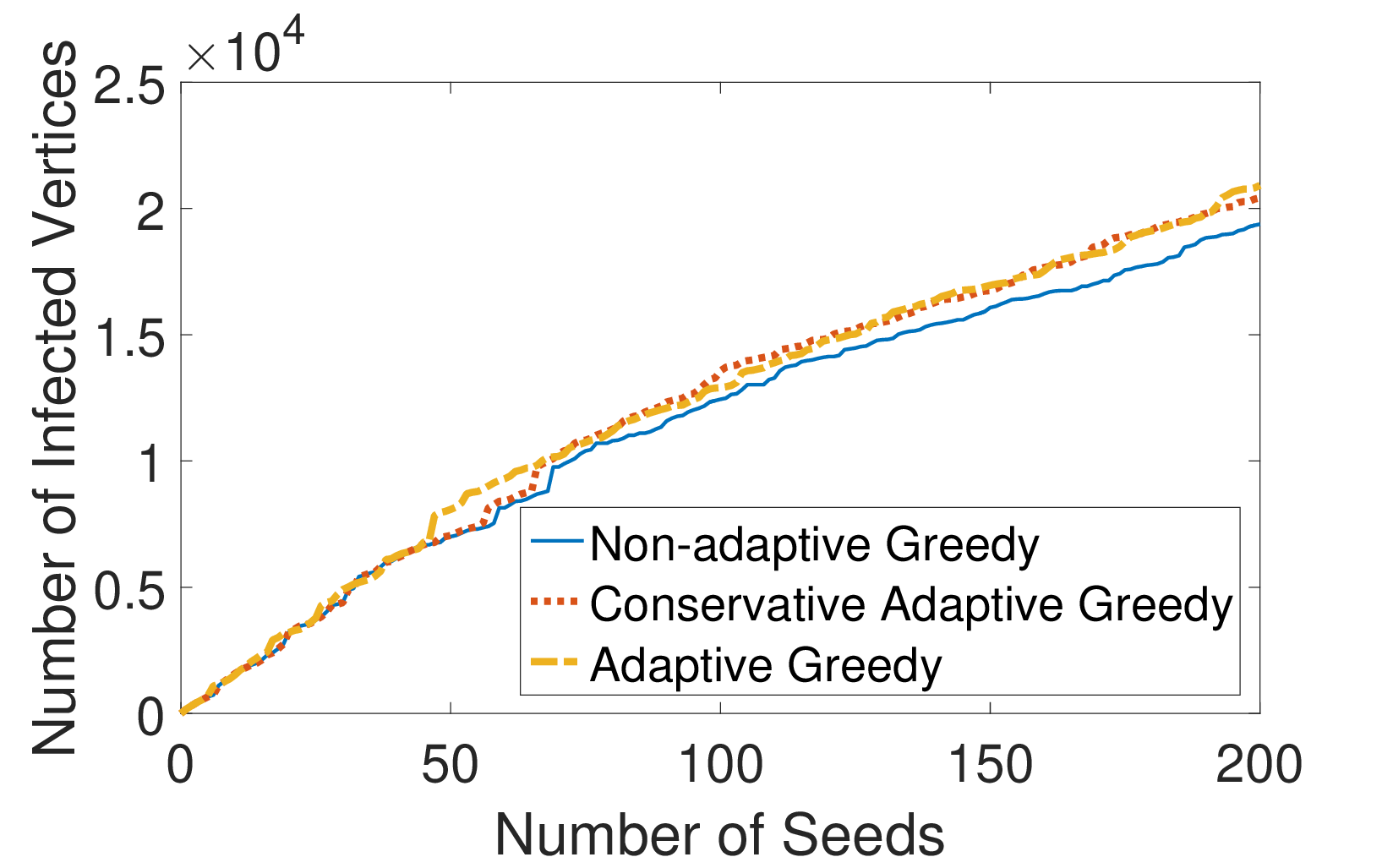}
    \includegraphics[width=0.45\textwidth]{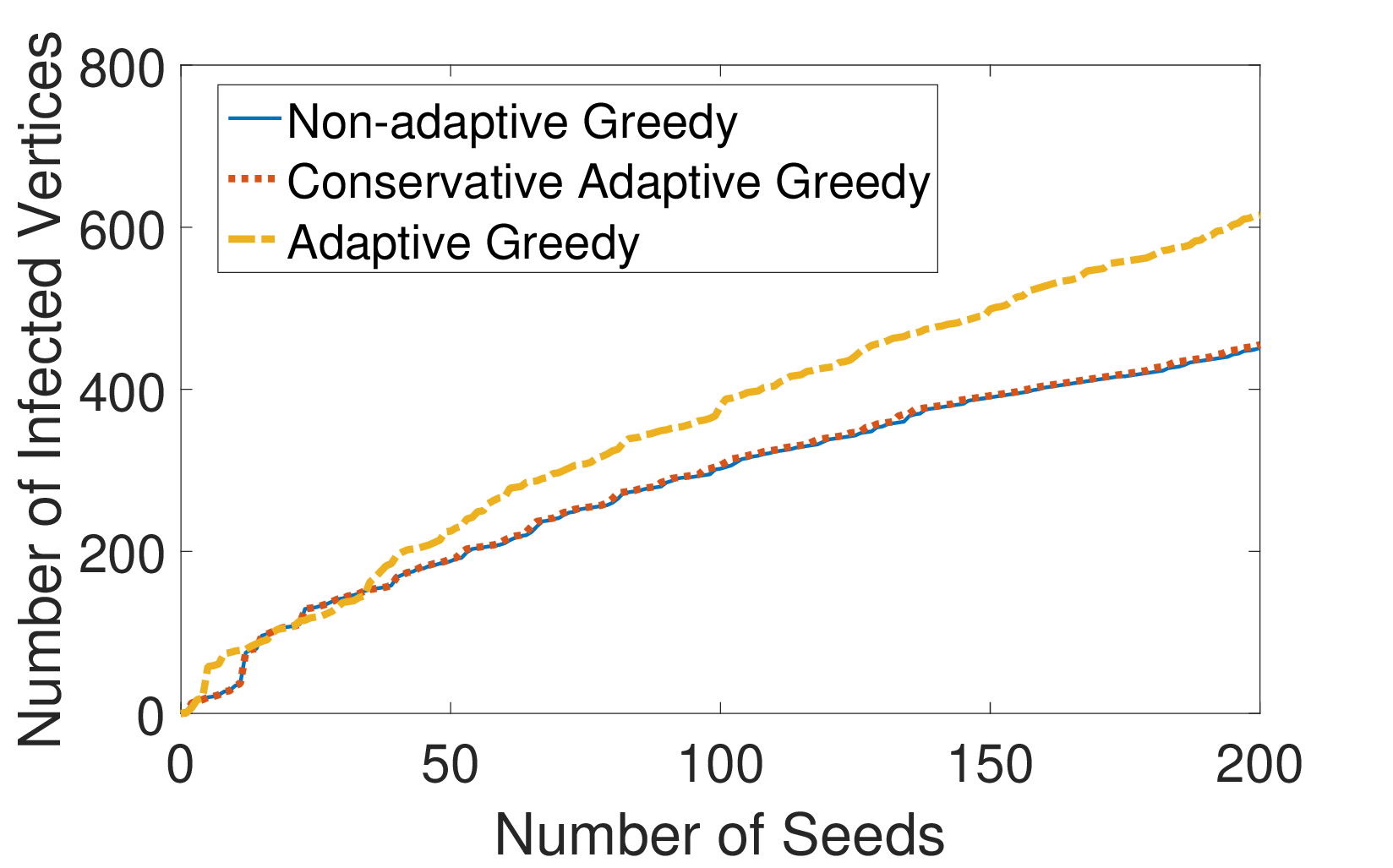}
    \includegraphics[width=0.45\textwidth]{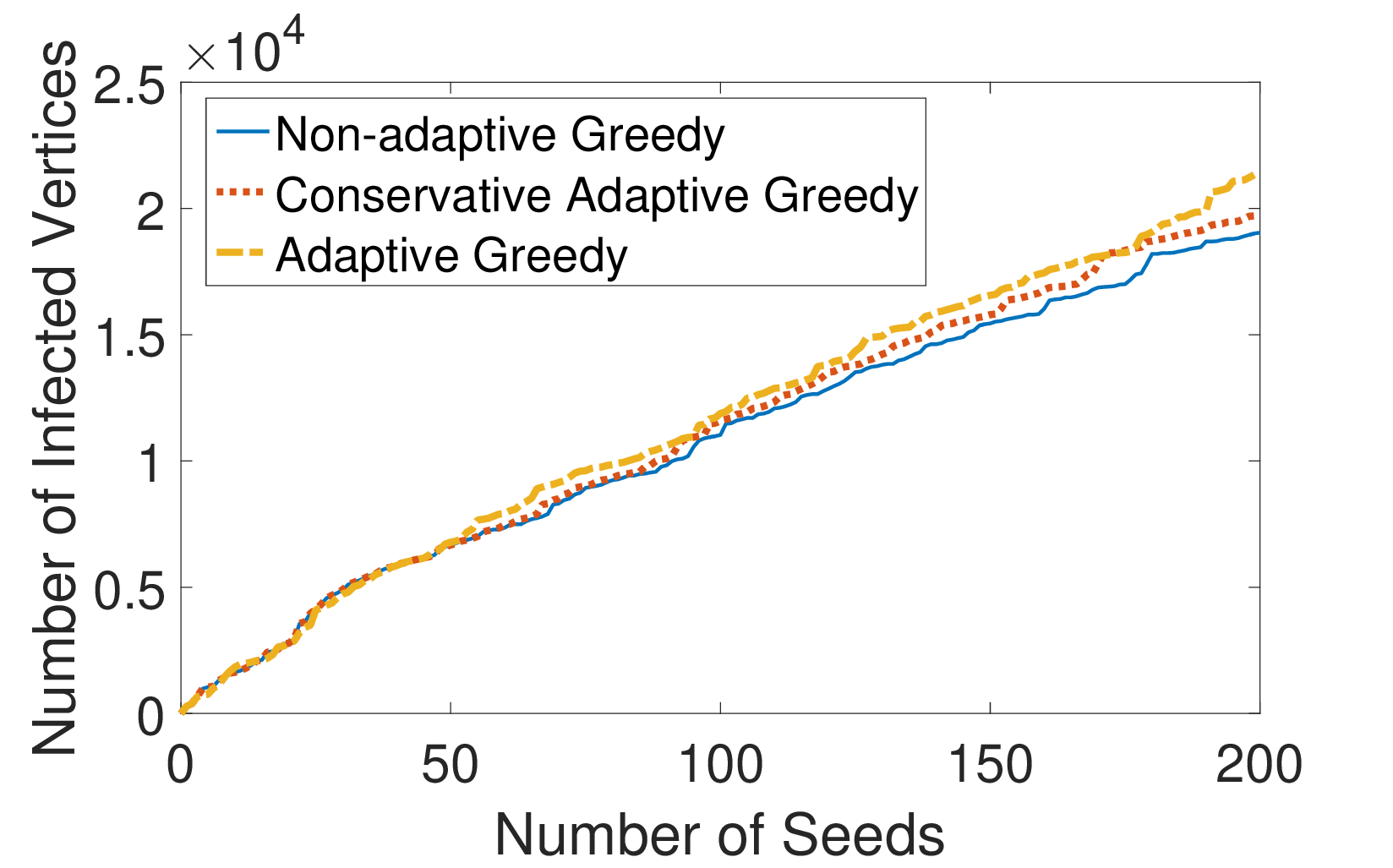}
    \includegraphics[width=0.45\textwidth]{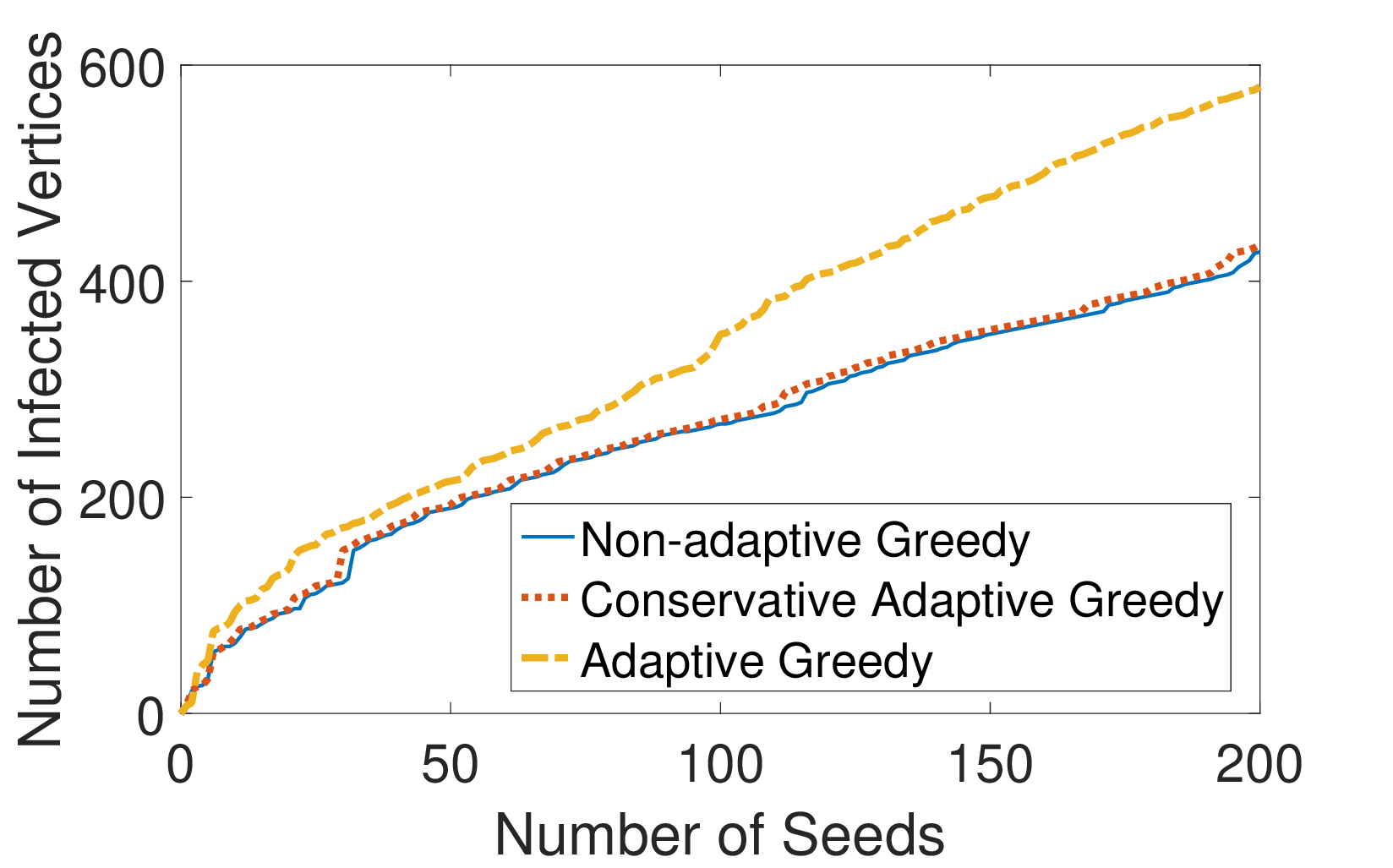}
    \includegraphics[width=0.45\textwidth]{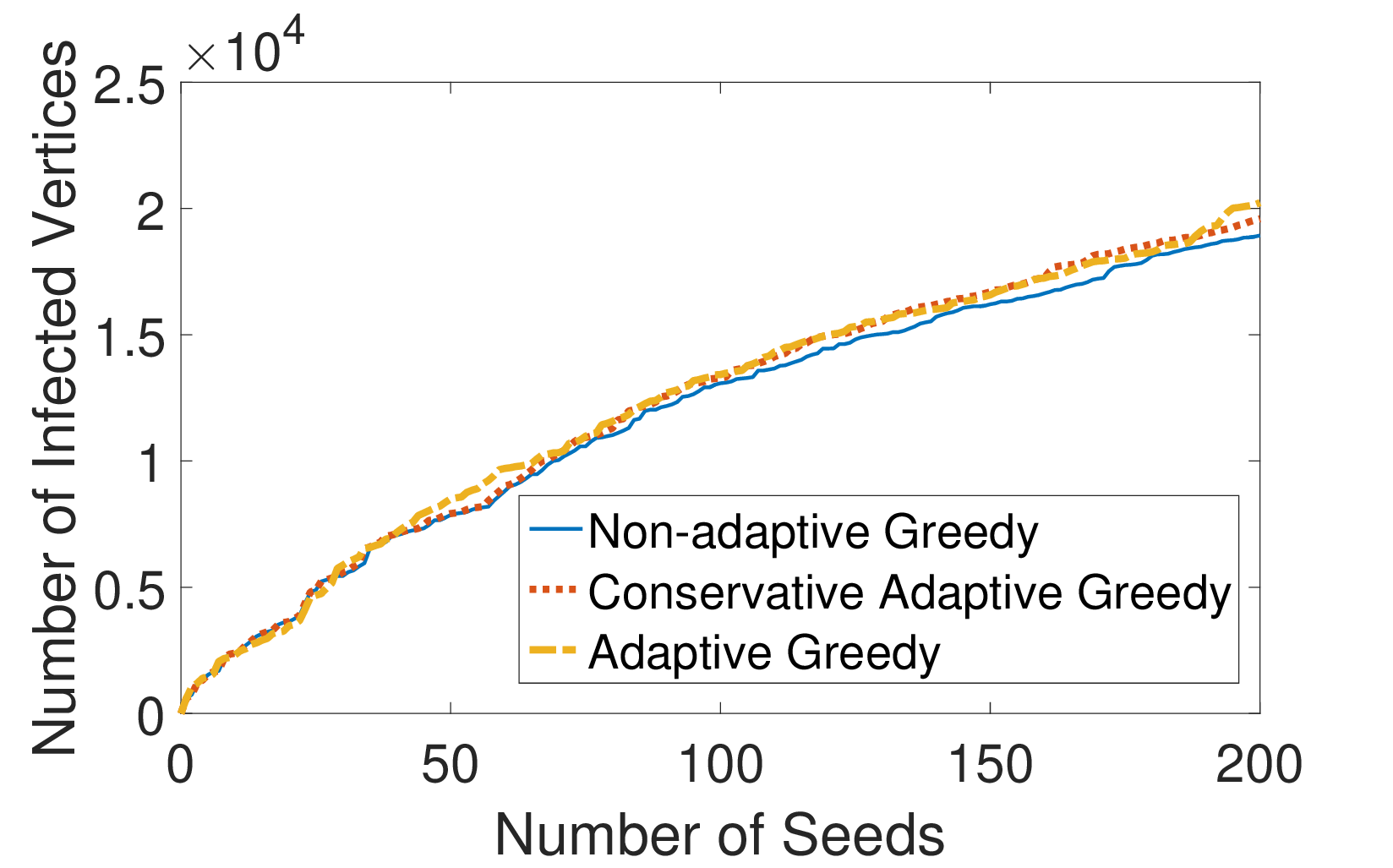}
    \caption{The results for the dataset DBLP. The three rows correspond to the three realizations $\phi_1,\phi_2,\phi_3$, the left column is for \ICM, and the right column is for \LTM.}%
    \label{fig:DBLP}
\end{figure}

\begin{figure}
    \centering
    \includegraphics[width=0.45\textwidth]{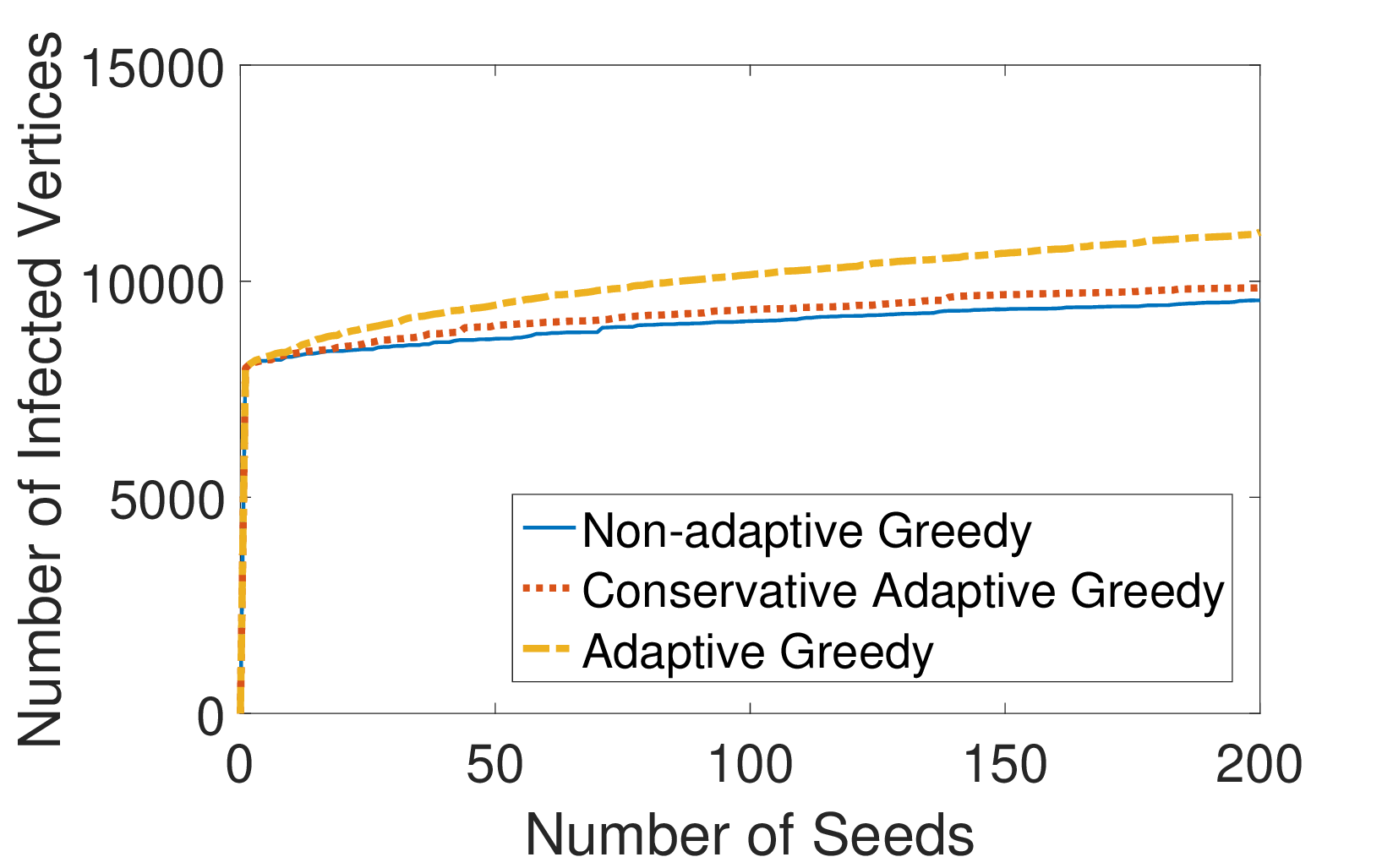}
    \includegraphics[width=0.45\textwidth]{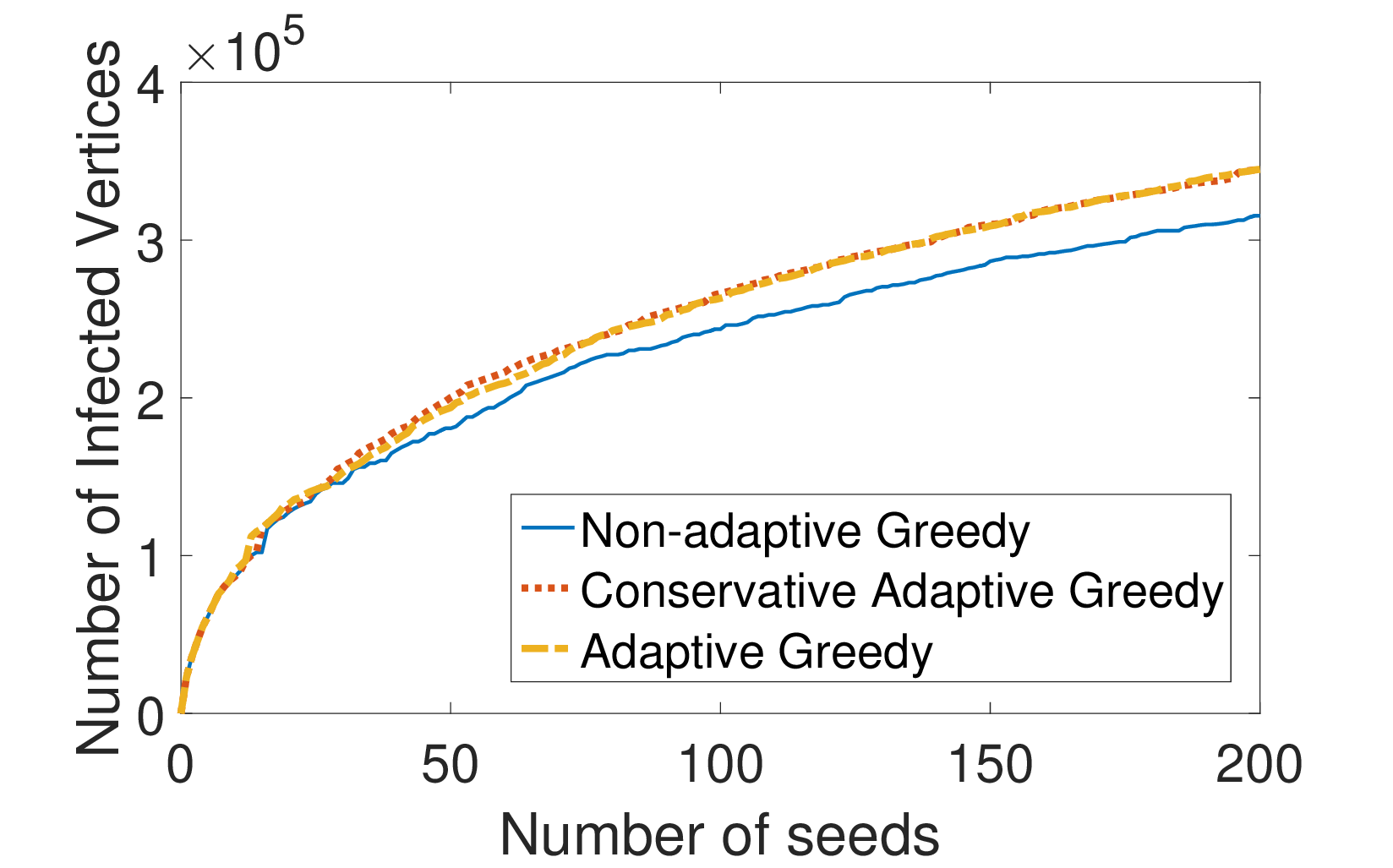}
    \includegraphics[width=0.45\textwidth]{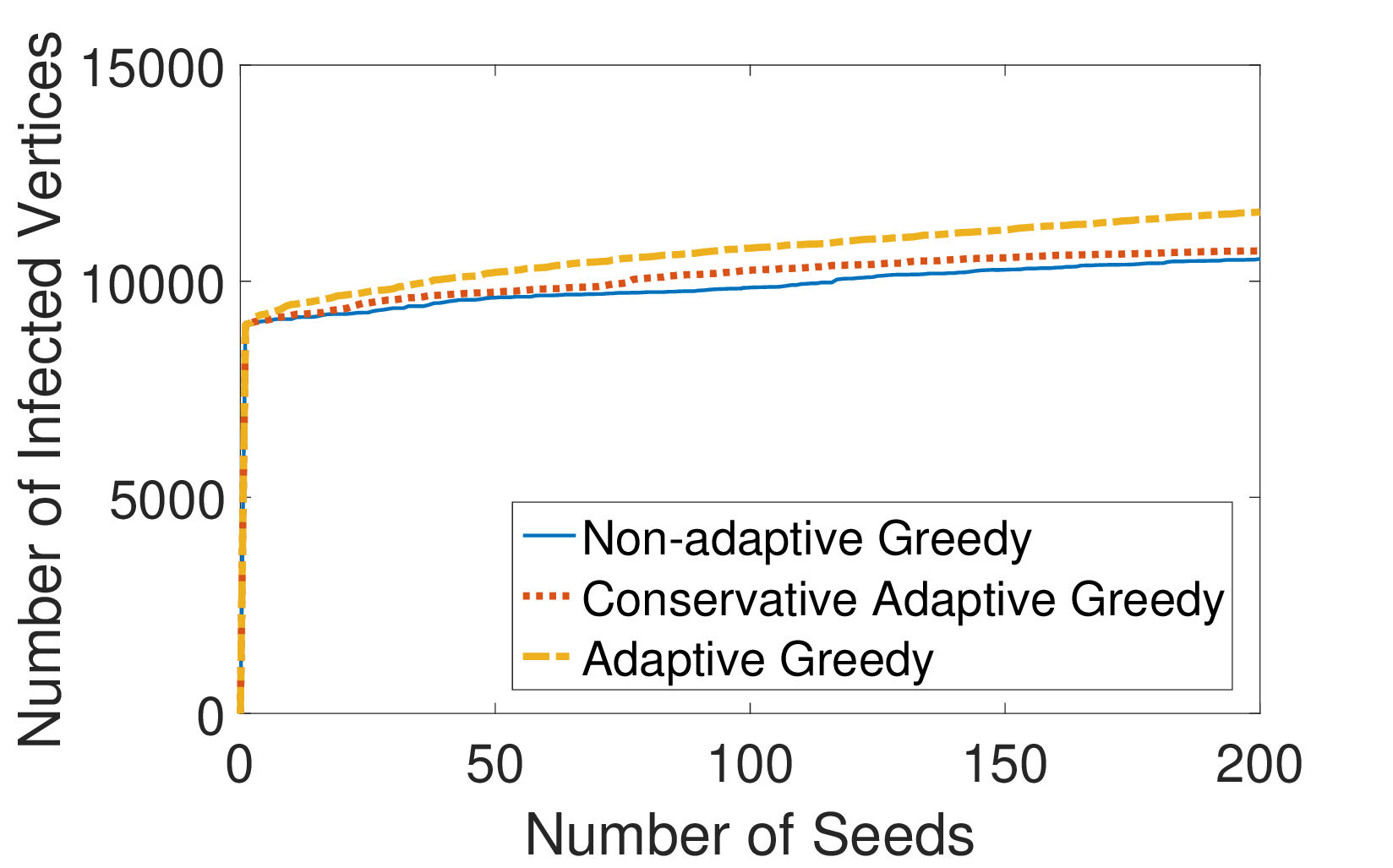}
    \includegraphics[width=0.45\textwidth]{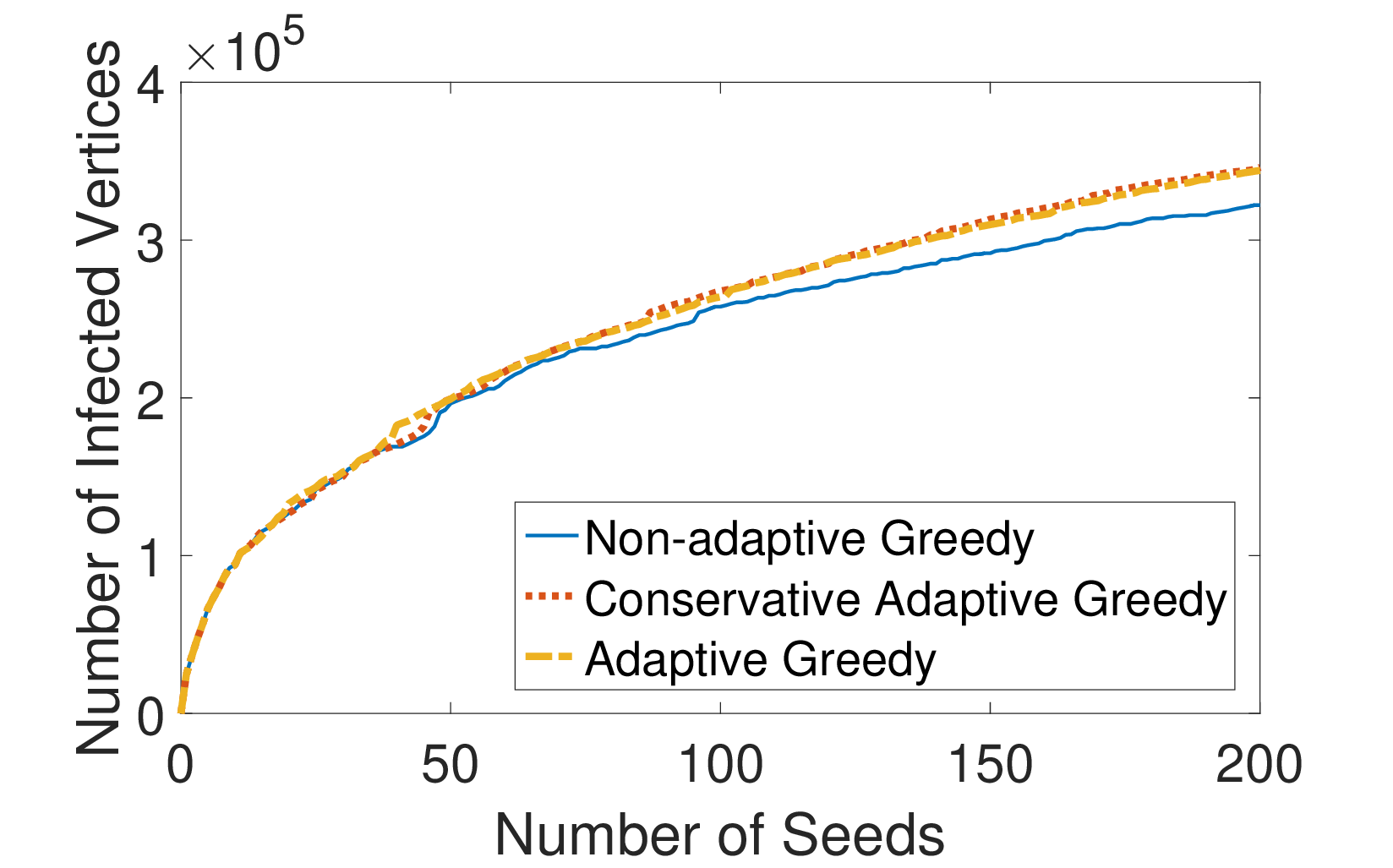}
    \includegraphics[width=0.45\textwidth]{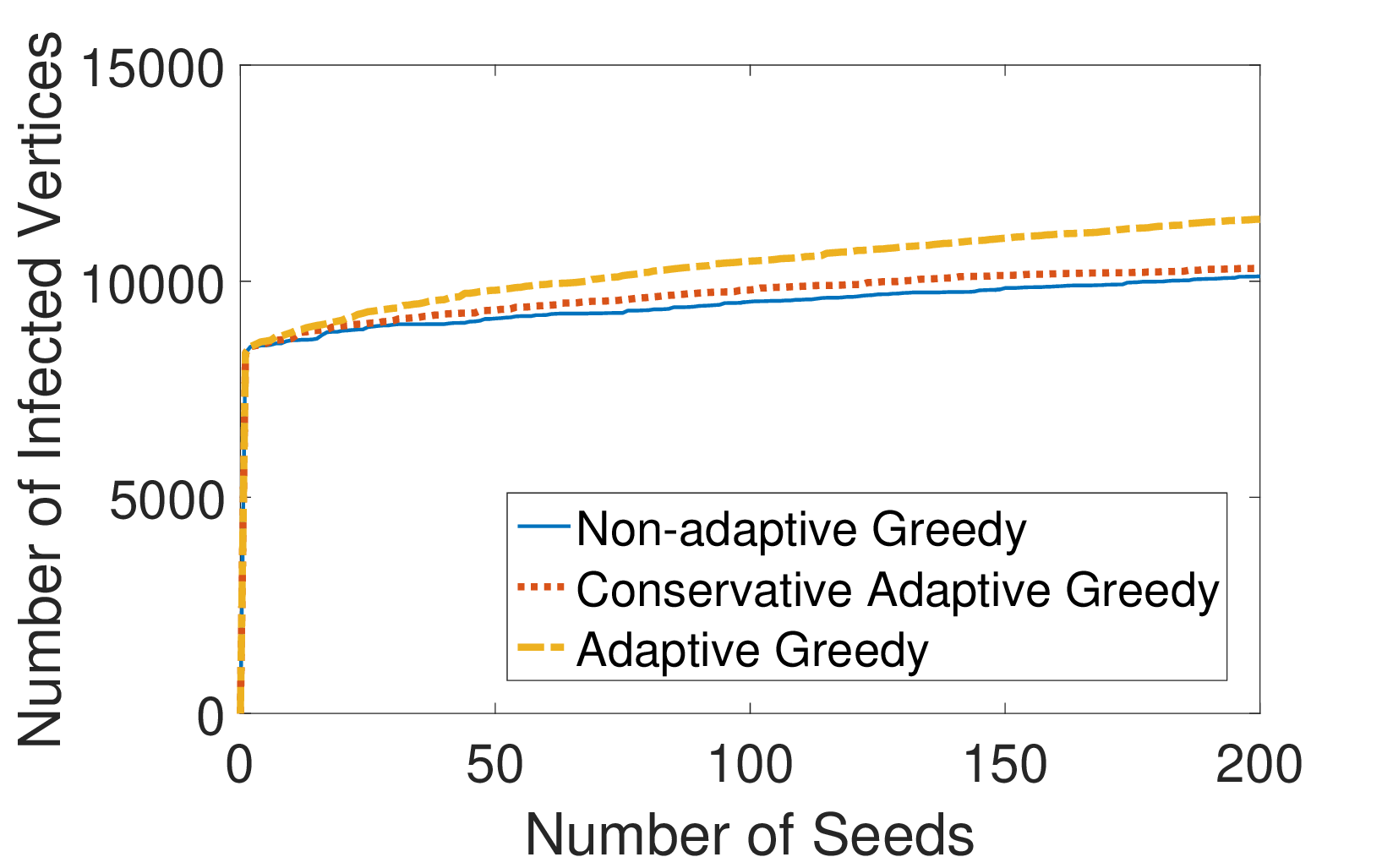}
    \includegraphics[width=0.45\textwidth]{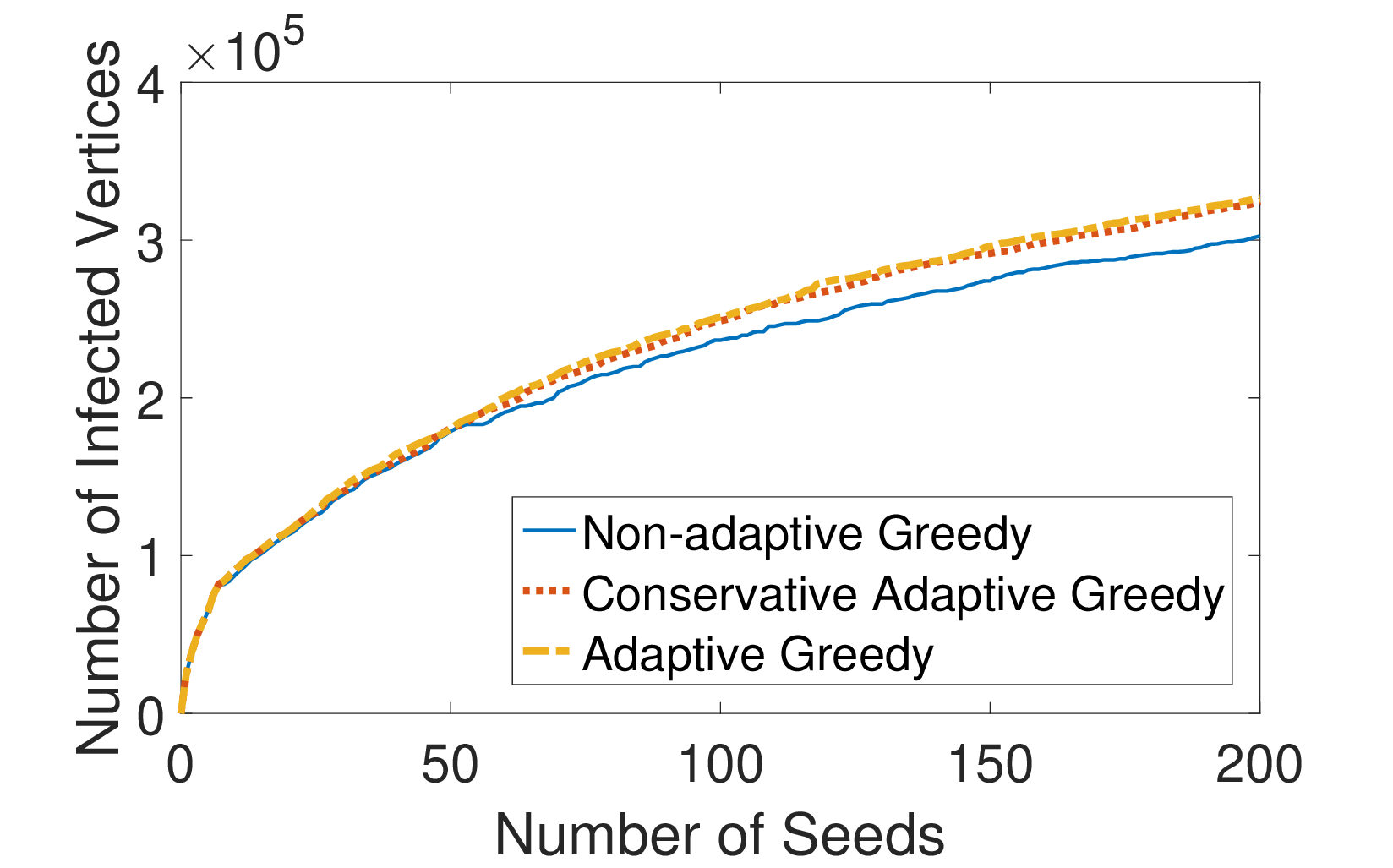}
    \caption{The results for the dataset com-YouTube. The three rows correspond to the three realizations $\phi_1,\phi_2,\phi_3$, the left column is for \ICM, and the right column is for \LTM.}%
    \label{fig:YouTube}
\end{figure}

\section{Conclusion and Open Problems}
\label{sect:conclusion}
We have seen that the infimum of the greedy adaptivity gap is exactly $(1-1/e)$ for \ICM, \LTM, and general triggering models with both the full-adoption feedback model and the myopic feedback model.
We have also seen that the supremum of this gap is infinity for the full-adoption feedback model.
One natural open problem is to find the supremum of the greedy adaptivity gap for the myopic feedback model.
Another natural open problem is to find the supremum of the greedy adaptivity gap for the more specific \ICM and \LTM.  

The greedy adaptivity gap studied in this paper is closely related to the adaptivity gap studied in the past.
Since the non-adaptive greedy algorithm is always a $(1-1/e)$-approximation of the non-adaptive optimal solution, a constant adaptivity gap implies a constant greedy adaptivity gap.
For example, the adaptivity gap for \ICM with myopic feedback is at most $4$ \shortcite{peng2019adaptive}, so the greedy adaptivity gap in the same setting is at most $\frac4{1-1/e}$.
In addition, the greedy adaptive policy is known to achieve a $(1-1/e)$-approximation to the adaptive optimal solution for \ICM with full-adoption feedback~\shortcite{golovin2011adaptive},
so the adaptivity gap and the greedy adaptivity gap could either be both constant or both unbounded for \ICM with full-adoption feedback model, but it remains open which case is true.
The adaptivity gap for \ICM with full-adoption feedback, as well as the adaptivity gap for \LTM with both feedback models, are all important open problems.
We believe these problems can be studied together with the greedy adaptivity gap.

\section*{Acknowledgements}
A significantly shorter version of this paper is available in AAAI'20:

https://ojs.aaai.org/index.php/AAAI/article/view/5398.

Grant Schoenebeck and Biaoshuai Tao are pleased to acknowledge the support of National Science Foundation AitF \#1535912 and CAREER \#1452915.
Biaoshuai Tao is pleased to acknowledge the support of National Natural Science Foundation of China Grant 62102252.

\appendix
\section{Original Definitions and Some Intuitions for \ICM and \LTM}
\label{append:originaldefinition}
In the original definition, \ICM is defined such that each vertex $u$ attempts only once to infect each of its not-yet-infected out-neighbor $v$ with probability $w(u,v)$.

\begin{definition}
The \emph{independent cascade model} $IC_G$ is defined by a directed edge-weighted graph $G=(V,E,w)$ such that $w(u,v)\leq1$ for each $(u,v)\in E$.
On input seed set $S\subseteq V$, $IC_G(S)$ outputs a set of infected vertices as follows:
\begin{enumerate}
    \item Initially, only vertices in $S$ are infected.
    \item In each subsequent round, each vertex $u$ infected in the previous round infects each (not yet infected) out-neighbor $v$ with probability $w(u,v)$ independently.
    \item After a round where there is no additional infected vertices, $IC_G(S)$ outputs the set of infected vertices.
\end{enumerate}
\end{definition}
It is straightforward to see that this definition is equivalent to Definition~\ref{def:ICM}.

The basic idea behind the original \LTM is that the influence from the in-neighbors of a vertex is additive.
\begin{definition}
The \emph{linear threshold model} $LT_G$ is defined by a directed edge-weighted graph $G=(V,E,w)$ such that $\sum_{u:u\in\Gamma(v)}w(u,v)\leq1$ for each $v\in V$.
On input seed set $S\subseteq V$, $LT_G(S)$ outputs a set of infected vertices as follows:
\begin{enumerate}
    \item Initially, only vertices in $S$ are infected, and for each vertex $v$ a \emph{threshold} $\theta_v$ is sampled uniformly at random from $[0,1]$ independently.
    \item In each subsequent round, a vertex $v$ becomes infected if 
    $$\sum_{u:u\in\Gamma(v)\text{ and }u\text{ is infected}}w(u,v)\geq\theta_v.$$
    \item After a round where there is no additional infected vertices, $LT_G(S)$ outputs the set of infected vertices.
\end{enumerate}
\end{definition}
\shortciteA{KempeKT03} showed that the definition above is equivalent to Definition~\ref{def:LTM}.
For an intuition of this, consider a not-yet-infected vertex $v$ and a set of its infected neighbors $\IN_v\subseteq\Gamma(v)$.
$v$ will be infected by vertices in $\IN_v$ with probability $\sum_{u:u\in\IN_v}w(u,v)$, as  $\Pr\left(\theta_v\leq\sum_{u:u\in \IN_v}w(u,v)\right)=\sum_{u:u\in \IN_v}w(u,v)$.
In the case where $v$ becomes infected, we can attribute its infection to exactly one of its infected neighbors.  The infection will be attributed to neighboring infected vertex $u$ with probability equal to $w(u, v)$ (in which case $T_v=\{u\}$).
Overall, the probability that $v$ includes an incoming edge from $\{(u,v):u\in\IN_v\}$ is exactly $\sum_{u:u\in\IN_v}w(u,v)$.

\section{On General Threshold Model}
\label{append:GTM}
In this section, we show that all our theoretical results in this paper hold for submodular general threshold model, a model that is more general than the triggering model.\footnote{In particular, a diffusion model that is captured by the general threshold model with submodular local influence functions but not the triggering model, named \emph{decreasing cascade model}, was discovered in the full version of \shortcite{KempeKT03}; this indicates that even the general threshold model with submodular local influence functions is strictly more general than the triggering model. \shortciteA{salek2010you} completely characterized the necessary and sufficient condition under which a general threshold model can be captured by a triggering model.}
In Section~\ref{append:GTM_definition}, we define the general threshold model, and we define the two feedback models, the full-adoption and the myopic, based on the general threshold model.
In Section~\ref{append:extension}, we justify that all our results in this paper hold for submodular general threshold model.
Notice that, however, our empirical results in Section~\ref{sect:experiments} depend on the reverse reachable set technique, which is only compatible with the triggering model.

\subsection{General Threshold Model and Feedback}
\label{append:GTM_definition}
\begin{definition}[\shortciteA{KempeKT03}]\label{def:GTM2}
  The \emph{general threshold model}, $I_{G,F}$, is defined by a graph $G=(V,E)$ and for each vertex $v$
  a monotone \emph{local influence function} $f_v:\{0, 1\}^{|\Gamma(v)|} \to[0,1]$ with $f_v(\emptyset) = 0$.
  Let $F=\{f_v\mid v\in V\}$.

  On an input seed set $S \subseteq V$,  $I_{G,F}(S)$ outputs a set of infected vertices as follows:
  \begin{enumerate}[noitemsep,nolistsep]
    \item[1.] Initially, only vertices in $S$ are infected, and for each vertex $v$ the threshold $\theta_v$ is sampled uniformly at random from the interval $(0,1]$ independently.
    \item[2.]  In each subsequent round, a vertex $v$ becomes infected if the influence of its infected in-neighbors, $\IN_v\subseteq\Gamma(v)$, exceeds its threshold: $f_v(\IN_v)\geq\theta_v$.
    \item[3.] After a round where no additional vertices are infected, the set of infected vertices is the output.
  \end{enumerate}
\end{definition}

$I_{G,F}$ in Definition~\ref{def:GTM2} can be viewed as a random function $I_{G,F}:\{0,1\}^{|V|}\to\{0,1\}^{|V|}$.
In addition, if the thresholds of all the vertices are fixed, this function becomes deterministic.
Correspondingly, we define a \emph{realization} of a graph $G=(V,E)$ as a function $\phi:V\to(0,1]$ which encodes the thresholds of all vertices.
Let $I_{G,F}^\phi:\{0,1\}^{|V|}\to\{0,1\}^{|V|}$ be the deterministic function corresponding to the general threshold model $I_{G,F}$ with vertices' thresholds following realization $\phi$.
We will interchangeably consider $\phi$ as a function defined above or a $|V|$ dimensional vector in $(0,1]^{|V|}$, and we write $\phi\sim(0,1]^{|V|}$ to mean a random realization is sampled such that each $\theta_v$ is sampled uniformly at random and independently as it is in Definition~\ref{def:GTM2}.

Like the triggering model, the general threshold model also captures the independent cascade and linear threshold models.
\begin{itemize}
    \item \ICM is a special case of the general threshold model $I_{G,F}$ where $G=(V,E,w)$ is an edge-weighted graph with $w(u,v)\in(0,1]$ for each $(u,v)\in E$ and $f_v(\IN_v)=1-\prod_{u\in \IN_v}(1-w(u,v))$ for each $f_v\in F$.
    \item \LTM is a special case of the general threshold model $I_{G,F}$ where $G=(V,E,w)$ is an edge-weighted graph with $w(u,v)>0$ for each $(u,v)\in E$ and $\sum_{u\in\Gamma(v)}w(u,v)\leq1$ for each $v\in V$ and $f_v(\IN_v)=\sum_{u\in \IN_v}w(u,v)$ for each $f_v\in F$.
\end{itemize}

Given a general threshold model $I_{G,F}$, the global influence function is then defined as $\sigma_{G,F}(S)=\E_{\phi\sim(0,1]^{|V|}}[|I_{G,F}^\phi(S)|]$.
\shortciteA{MosselR10} showed that $\sigma(\cdot)$ is monotone and submodular if each $f_v(\cdot)$ is monotone and submodular.
We normally say that a general threshold model $I_{G,F}$ is submodular if each $f_v\in F$ is submodular.
Notice that this implies $\sigma(\cdot)$ is submodular.

In the remaining part of this section, we define the \emph{full-adoption feedback model} and the \emph{myopic feedback model} corresponding to the general threshold model.

When the seed-picker sees that a vertex $v$ is not infected ($v$ may be a vertex adjacent to $I_{G,F}^\phi(S)$ in the full-adoption feedback model, or a vertex adjacent to $S$ in the myopic feedback model), the seed-picker has certain partial information about $v$'s threshold.
Specifically, let $\IN_v$ be $v$'s infected in-neighbors that are observed by the seed-picker.
By seeing that $v$ is not infected, the seed-picker knows that the threshold of $v$ is in the range $(f_v(\IN_v),1]$, and the posterior distribution of $\theta_v$ is the uniform distribution on this range.

Let the \emph{level} of a vertex $v$, denoted by $o_v$, be a value which either equals a character $\checkmark$ indicating that it is infected, or a real value $\vartheta_v\in[0,1]$ indicating that $\theta_v\in(\vartheta_v,1]$.
Let $O=\{\checkmark\}\cup[0,1]$ be the space of all possible levels.
A \emph{partial realization} $\varphi$ is a function specifying a level for each vertex: $\varphi:V\to O$.
We say that a partial realization $\varphi$ \emph{is consistent with} the full realization $\phi$, denoted by $\phi\simeq\varphi$, if $\phi(v)>\varphi(v)$ for any $v\in V$ such that $\varphi(v)\neq\checkmark$.

\begin{definition}
  Given a general threshold model $I_{G=(V,E),F}$ with a realization $\phi$, the \emph{full-adoption feedback} is a function $\Phi_{G,F,\phi}^{\f}$ mapping a seed set $S\subseteq V$ to a partial realization $\varphi$ such that
  \begin{itemize}
      \item $\varphi(v)=\checkmark$ for each $v\in I_{G,F}^\phi(S)$, and
      \item $\varphi(v)=f_v(I_{G,F}^\phi(S)\cap\Gamma(v))$ for each $v\notin I_{G,F}^\phi(S)$.
  \end{itemize}
\end{definition}

\begin{definition}
  Given a general threshold model $I_{G=(V,E),F}$ with a realization $\phi$, the \emph{myopic feedback} is a function $\Phi_{G,F,\phi}^\m$ mapping a seed set $S\subseteq V$ to a partial realization $\varphi$ such that
  \begin{itemize}
      \item $\varphi(v)=\checkmark$ for each $v\in S$, and
      \item for each $v\notin S$, $\varphi(v)=\checkmark$ if $f_v(S\cap\Gamma(v))\geq\phi(v)$, and $\varphi(v)=f_v(S\cap\Gamma(v))$ if $f_v(S\cap\Gamma(v))<\phi(v)$.
  \end{itemize}
\end{definition}

Notice that, in both definitions above, a vertex $v$ that does not have any infected neighbor (i.e., $v\notin S$ such that $I_{G,F}^\phi(S)\cap\Gamma(v)=\emptyset$ for the full-adoption feedback model or $S\cap\Gamma(v)=\emptyset$ for the myopic feedback model) always satisfies $\varphi(v)=0$, as $f_v(\emptyset)=0$ by Definition~\ref{def:GTM2}.

After properly defining the two feedback models, the definition of the adaptive policy $\pi$, as well as the definitions of the functions $\Seed^\f(\cdot,\cdot,\cdot),\Seed^\m(\cdot,\cdot,\cdot),\sigma^\f(\cdot,\cdot),\sigma^\m(\cdot,\cdot)$, are exactly the same as they are in Section~\ref{sect:Prelim_infmax}.
The definitions of the adaptivity gap and the greedy adaptivity gap are also the same as they are in Section~\ref{sect:prelim_AG}.

\subsection{Extending of Our Results to General Threshold Model}
\label{append:extension}
We will show in this section that all our results can be extended to the submodular general threshold model.
Recall that a general threshold model is submodular means that all the local influence functions $f_v$'s are submodular.
In this section, whenever we write $I_{G,F}$, we refer to the general threshold model in Definition~\ref{def:GTM2}, not the triggering model in Definition~\ref{def:GTM}.

\subsubsection{Infimum of Greedy Adaptivity Gap}
Theorem~\ref{thm:inf_gap} is extended as follows.
\begin{theorem}\label{thm:inf_gap2}
For the full-adoption feedback model,
$$\inf_{G,F,k:\mbox{ }I_{G,F}\text{ is \ICM}}\frac{\sigma^\f(\pi^g,k)}{\sigma(S^g(k))}=\inf_{G,F,k:\mbox{ }I_{G,F}\text{ is \LTM}}\frac{\sigma^\f(\pi^g,k)}{\sigma(S^g(k))}$$
$$\qquad=\inf_{G,F,k:\mbox{ }I_{G,F}\text{ is submodular}}\frac{\sigma^\f(\pi^g,k)}{\sigma(S^g(k))}=1-\frac1e.$$
The same result holds for the myopic feedback model.
\end{theorem}

Recall that Theorem~\ref{thm:inf_gap} can be easily implied by Lemma~\ref{lem:tightICM}, Lemma~\ref{lem:tightLTM} and Theorem~\ref{thm:lowerbound}.
Since Lemma~\ref{lem:tightICM} and Lemma~\ref{lem:tightLTM} are for specific models \ICM and \LTM which are compatible with both the triggering model and the general threshold model, their validity here is clear.
Following the same arguments, Theorem~\ref{thm:inf_gap2} can be implied by Lemma~\ref{lem:tightICM}, Lemma~\ref{lem:tightLTM} and the following theorem which is the counterpart to Theorem~\ref{thm:lowerbound}.

\begin{theorem}\label{thm:lowerbound2}
If $I_{G,F}$ is a submodular general threshold model, then we have both
$$\sigma^\f(\pi^g,k)\geq\left(1-\frac1e\right)\max_{S\subseteq V,|S|\leq k}\sigma(S)
\qquad\mbox{and}\qquad
\sigma^\m(\pi^g,k)\geq\left(1-\frac1e\right)\max_{S\subseteq V,|S|\leq k}\sigma(S).$$
\end{theorem}

Similar to the proof of Theorem~\ref{thm:lowerbound}, Theorem~\ref{thm:lowerbound2} can be proved by showing the three propositions: Proposition~\ref{prop:partial_submodular}, Proposition~\ref{prop:marginal} and Proposition~\ref{prop:generallowerbound}.
It is straightforward to check that Proposition~\ref{prop:marginal} and Proposition~\ref{prop:generallowerbound} hold for the general threshold model with exactly the same proofs.
Now, it remains to extend Proposition~\ref{prop:partial_submodular} to the general threshold model, which is restated and proved below.

\begin{proposition}
Given a submodular general threshold model $I_{G,F}$, any $S\subseteq V$, any feedback model (either full-adoption or myopic) and any partial realization $\varphi$ that is a valid feedback of $S$ (i.e., $\exists\phi:\varphi=\Phi_{G,F,\phi}^\f(S)$ or $\exists\phi:\varphi=\Phi_{G,F,\phi}^\m(S)$, depending on the feedback model considered), the function $\mathcal{T}:\{0,1\}^{|V|}\to\R_{\geq0}$ defined as $\mathcal{T}(X)=\E_{\phi\simeq\varphi}[|I_{G,F}^\phi(S\cup X)|]$ is submodular.
\end{proposition}
\begin{proof}
Fix a feedback model, $S\subseteq V$ and $\varphi$ that is a valid feedback of $S$.
Let $T=\{v\mid\varphi(v)=\checkmark\}$ be the set of infected vertices indicated by the feedback of $S$.
We consider a new general threshold model $I_{G',F'}$ defined as follows:
\begin{itemize}
    \item $G'$ is obtained by removing vertices in $T$ from $G$ (and the edges connecting from/to vertices in $T$ are also removed);
    \item For any $v\in V'=V\setminus T$, $\Gamma(v)\cap T$ is the set of in-neighbors of $v$ that are removed. Define $f_v'(Y)=\frac{f_v((\Gamma(v)\cap T)\cup Y)-\varphi(v)}{1-\varphi(v)}$ for each subset $Y$ of $v$'s in-neighbors in the new graph $G'$: $Y\subseteq \Gamma(v)\cap V'$.
\end{itemize}
Notice that $f_v'$ is a valid local influence function.
$f_v'$ is clearly monotone.
For each $v\in V'$, we have $\varphi(v)=f_v(\Gamma(v)\cap T)$, as this is exactly the feedback received from the fact that $v$ has not yet infected.
It is then easy to see that $f_v'$ is always non-negative and $f_v'(\emptyset)=0$.

A simple coupling argument can show that
\begin{equation}\label{eqn:coupling2}
    \E_{\phi\simeq\varphi}\left[\left|I_{G,F}^\phi(S\cup X)\right|\right]=
    \sigma_{G',F'}(X\setminus T)+|T|.
\end{equation}
To define the coupling, for each $v\in V'$, the threshold of $v$ in $G$, $\theta_v$, is coupled with the threshold of $v$ in $G'$ as $\theta_v'=\frac{\theta_v-\varphi(v)}{1-\varphi(v)}$.
This is a valid coupling: by $\phi\simeq\varphi$, we know that $\theta_v$ is sampled uniformly at random from $(\varphi(v),1]$, which indicates that the marginal distribution of $\theta_v'$ is the uniform distribution on $(0,1]$, which makes $I_{G',F'}$ a valid general threshold model.

Under this coupling, on the vertices $V'$, the cascade in $G$ with seeds $S\cup X$ and partial realization $\varphi$ is identical to the cascade in $G'$ with seeds $X\setminus T$.
To see this, consider an arbitrary vertex $v\in V'$ and let $\IN_v$ and $\IN_v'$ be $v$'s infected neighbors in $G$ and $G'$ respectively.
For induction hypothesis, suppose the two cascade processes before $v$'s infection are identical.
We have $\IN_v=\IN_v'\cup(\Gamma(v)\cap T)$ and $\IN_v'\cap(\Gamma(v)\cap T)=\emptyset$.
It is easy to see from our construction that $v$ is infected in $G$ if and only if $v$ is infected in $G'$:
$$f_v(\IN_v)\geq\theta_v\Leftrightarrow f_v'(\IN_v')=\frac{f_v(\IN_v)-\varphi(v)}{1-\varphi(v)}\geq\theta_v'.$$
This proves Equation~(\ref{eqn:coupling2}).

Finally, since each $f_v(\cdot)$ is assumed to be submodular, it is easy to see that each $f_v'(\cdot)$ is submodular by our definition.
Thus, $I_{G',F'}$ is a submodular model.
This, combined with Equation~(\ref{eqn:coupling2}), proves the proposition.
\end{proof}

\subsubsection{Supremum of Greedy Adaptivity Gap}
All the results in Section~\ref{sect:sup} about the supremum of the greedy adaptivity gap can be extended easily to the submodular general threshold model.
In particular, Lemma~\ref{lem:additive} and Lemma~\ref{lem:ltm_precribed} are under \LTM, which is compatible with the submodular general threshold model.
Theorem~\ref{thm:sup_gap} and Theorem~\ref{thm:adaptivityGap} are proved by providing an example with a diffusion model that is a combination of \ICM and \LTM, and the diffusion model constructed in Definition~\ref{def:LTMc} can be easily described in the formulation of the general threshold model, since both \ICM and \LTM can be described in the general threshold model.

\bibliography{reference}
\bibliographystyle{theapa}

\end{document}